\documentclass[11pt]{article}

\usepackage{amsmath,mathtools,amsthm,amssymb}
\usepackage{fullpage}
\usepackage{thm-restate}
\usepackage{enumerate}
\usepackage[toc,page]{appendix}
\usepackage{booktabs}
\usepackage{hyperref}

\newtheorem{theorem}{Theorem}[section]
\newtheorem{maintheorem}{Theorem}
\newtheorem{lemma}[theorem]{Lemma}
\newtheorem{claim}[theorem]{Claim}
\newtheorem{construction}[theorem]{Construction}
\newtheorem{definition}[theorem]{Definition}
\newtheorem{corollary}[theorem]{Corollary}
\newtheorem{proposition}[theorem]{Proposition}
\theoremstyle{remark} 
\newtheorem*{example*}{Example}
\newtheorem*{remark*}{Remark}
\newtheorem{remark}[theorem]{Remark}

\renewcommand{\epsilon}{\varepsilon}

\allowdisplaybreaks

\DeclareMathOperator{\rank}{\mathrm{rank}}

\DeclareMathOperator*{\Ex}{\mathbb{E}} 
\DeclareMathOperator{\codim}{\mathrm{codim}}
\DeclareMathOperator{\Tr}{\mathrm{Tr}}
\newcommand{\ord}{\mathrm{ord}}
\newcommand{\coeff}{\mathrm{coeff}}
\newcommand{\supp}{\mathrm{supp}}
\newcommand{\Ext}{\mathsf{Ext}}

\newcommand{\arrbegin}{\begin{eqnarray}}
\newcommand{\arrend}{\end{eqnarray}}

\newcommand{\barr}{\overline}

\newcommand{\RR}{\mathbb{R}}
\newcommand{\CC}{\mathbb{C}}
\newcommand{\ZZ}{\mathbb{Z}}
\newcommand{\NN}{\mathbb{N}}
\renewcommand{\AA}{\mathbb{A}}

\newcommand{\PP}{\mathbb{P}}
\newcommand{\FF}{\mathbb{F}}
\newcommand{\KK}{\mathbb{K}}

\newcommand{\eps}{\epsilon}

\newcommand{\poly}{\mathrm{poly}}

\newcommand{\Lone}[1]{\left\Vert #1 \right\Vert_1}
\newcommand{\Ltwo}[1]{\left\Vert #1 \right\Vert_2}
\newcommand{\Linf}[1]{\left\Vert #1 \right\Vert_{\infty}}

\DeclarePairedDelimiterX\braket[2]{\langle}{\rangle}{#1 \delimsize\vert #2}

\DeclarePairedDelimiterX{\infdivx}[2]{(}{)}{%
  #1\;\delimsize|\delimsize|\;#2%
}

\DeclarePairedDelimiter\ceil{\lceil}{\rceil}

\newcommand{\set}[1]{\left\{ #1 \right\}}

\usepackage{tikz-cd}

\begin{document}

\title{
Extractors for Images of Varieties
}

\author{
\begin{tabular}{cc}
\begin{tabular}[t]{c}
Zeyu Guo\thanks{Supported in part by a Simons Investigator Award (\#409864, David Zuckerman).}\\CSE Department\\ Ohio State University \\ \texttt{zguotcs@gmail.com}
\end{tabular} &
\begin{tabular}[t]{c}
Ben Lee Volk\\Efi Arazi School of Computer Science\\ Reichman University \\ \texttt{benleevolk@gmail.com} 
\end{tabular}\\ \addlinespace[4ex]
\begin{tabular}[t]{c}
Akhil Jalan\\Computer Science Department\\ University of Texas at Austin \\ \texttt{akhiljalan@utexas.edu}
\end{tabular} &
\begin{tabular}[t]{c}
David Zuckerman\thanks{Supported in part by NSF Grants CCF-1705028 and CCF-2008076, a Simons Investigator Award (\#409864), and the Center of Mathematical Sciences and Applications at Harvard University.}
\\ Computer Science Department\\ University of Texas at Austin \\ \texttt{diz@utexas.edu}
\end{tabular}
\end{tabular}
}

\newcommand\CoAuthorMark{\footnotemark[\arabic{footnote}]} 
\date{}                  
\maketitle
\pagenumbering{roman}

\begin{abstract}
We construct explicit deterministic
extractors for \emph{polynomial images of varieties}, that is, distributions sampled by applying a low-degree polynomial map $f : \FF_q^r \to \FF_q^n$ to an element sampled uniformly at random from a $k$-dimensional variety $V \subseteq \FF_q^r$. This class of sources generalizes both \emph{polynomial sources}, studied by Dvir, Gabizon and Wigderson (FOCS 2007, Comput. Complex. 2009), and \emph{variety sources}, studied by Dvir (CCC 2009, Comput. Complex. 2012).

Assuming certain natural non-degeneracy conditions on the map $f$ and the variety $V$, which in particular ensure that the source has enough min-entropy, we extract almost all the min-entropy of the distribution.
Unlike the Dvir--Gabizon--Wigderson and Dvir results, our construction works over large enough finite fields of arbitrary characteristic.
One key part of our construction is an improved deterministic rank extractor for varieties.
As a by-product, we obtain explicit Noether normalization lemmas for affine varieties and affine algebras.

Additionally, we generalize a construction of affine extractors with exponentially small error due to Bourgain, Dvir and Leeman (Comput. Complex. 2016) by extending it to all finite prime fields of quasipolynomial size.
\end{abstract}
\thispagestyle{empty}

\newpage

\hypersetup{linkcolor={blue},citecolor={blue},urlcolor={blue}}  
{
  \hypersetup{linkcolor=blue}
  \tableofcontents
}

\newpage

\pagenumbering{arabic} 

\section{Introduction}

Randomness is a powerful resource in computing.  There are many useful randomized algorithms, and randomness is provably necessary in cryptography and distributed computing.  Naturally, these uses of randomness assume access to uniformly random bits.  However, it can be expensive or impossible to obtain such high-quality randomness.
A randomness extractor converts low-quality randomness into high-quality randomness.

Low-quality random sources can arise in several ways.  First, natural sources of randomness may be defective.  Second, in cryptography, if an adversary gains information about a string, then conditioned on this information, the string is weakly random.  Third, in constructing pseudorandom generators, a similar situation arises when we condition on the state of the computation.  Besides the computer science motivation, randomness extraction questions are natural mathematically.

We model a weak source as a class $\mathcal{D}$  of distributions over a finite set $\Omega$.  A randomness extractor for $\mathcal{D}$
is a deterministic function that extracts randomness from any distribution in $\mathcal{D}$.

\begin{definition}
An extractor for a class $\mathcal{D}$ of distributions with error $\eps$, or an $\eps$-extractor, is a function $\Ext : \Omega \to B$ such that
for any $D \in \mathcal{D}$, the distribution $\Ext(D)$ is $\eps$-close, in statistical distance, to the uniform distribution over $B$.
\end{definition}
Typically the codomain $B$ will be $\set{0,1}^m$.

The most general class of distributions is the set of distributions with high min-entropy, i.e., distributions that do not place much probability on any string.  However, it is not hard to show that it is impossible to extract from such sources.
It is possible to extract using an auxiliary seed, and there are many applications of such seeded extractors (see \cite{vadhan-book} for a survey).
It is also possible to extract from two independent general weak sources (e.g., \cite{cz}).
However, if we want to avoid adding a seed and only have one source, we must restrict the class of distributions further.

Various models of weak sources have been studied.  It is not hard to show that if there are not too many distributions in the class, then most functions are extractors with excellent parameters.  Of course, we really want efficiently-computable extractors.

Models of weak sources tend to be either complexity-theoretic or algebraic.
In this work, we focus on \emph{algebraic sources}. That is, we consider distributions over subsets $\Omega$ which have a ``nice'' algebraic structure.

\subsection{Algebraic Sources of Randomness}

Suppose $\FF$ is a finite field and $\Omega=\FF^n$. The simplest class of algebraic sources is the set of \emph{affine sources}. An affine source is simply the uniform distribution over an affine subspace $V \subseteq \FF^n$ of dimension $k$. Note that since $|V|=|\FF|^k$, the single parameter $k$ also determines the min-entropy of the uniform distribution over the source.

Gabizon and Raz \cite{gabizon-raz} constructed an explicit extractor $\Ext : \FF^n \to \FF^{k-1}$, assuming the field size is bounded from below by a large enough polynomial in $n$. For a large enough field size $q$, their construction extracts almost all of the randomness from the source and has error $\eps = 1/\poly(q)$.

The last feature is slightly undesirable, as ideally, one would like the error to decrease exponentially with $k$, the dimension of the source. Such a construction was given by Bourgain, Dvir and Leeman \cite{bourgain-dvir-leeman-2016}, albeit their construction requires the field size to be slightly super-polynomial in $n$, and only works for certain fields.

Over smaller fields, constructing affine extractors for small min-entropy is a more challenging task. Further, it is possible to show that any function $f:\FF_2^n \to \FF_2$ is constant on some affine subspace of dimension $\Omega(\log n)$ (see, e.g., Lemma 6.7 of the arXiv version of \cite{AggarwalBGS21}), and thus one cannot hope to extract even a single bit when the min-entropy is smaller than $\log n$ (compare this with the fact that over large fields, the Gabizon--Raz extractor works for any $k$).

Bourgain \cite{bourgain-2007-affine-extractor} constructed an extractor that works over $\FF_2$ for min-entropy $k=cn$ for a small constant $c$. This result was slightly improved by Yehudayoff \cite{yehudayoff-2011} and Li \cite{li-affine-extractors-11}.
Li \cite{li-affine-extractors} then presented a much improved construction which works when the min-entropy is as small as $k=\log^C(n)$ for some constant $C$, which was improved by \cite{chattopadhyay-goodman-liao-affine-extractors-2021} to $k=\log^{1+o(1)} (n)$. However, one drawback of the last two constructions is that the error parameter $\varepsilon$ is either constant or polynomially small, whereas one would hope for it to be exponentially small in $k$, as in the earlier constructions of Bourgain, Yehudayoff and Li.

There are several natural ways to generalize affine sources, but some care is needed when defining those generalizations. As we remarked earlier, for an affine subspace, the single parameter $k$ determines its size and hence the min-entropy of the corresponding source. For more complicated algebraic sets, however, as we shall now see, there are multiple parameters controlling their ``complexity,'' and the connection between those parameters and the min-entropy of the source is not always obvious.

Dvir, Gabizon and Wigderson \cite{dvir-gabizon-wigderson} considered \emph{polynomial sources}, which are defined by applying a low-degree polynomial map $P : \FF^k \to \FF^n$ on a uniformly random input from $\FF^k$. (Note that affine sources are a special case of polynomial sources when the degree equals one.) They further impose the algebraic condition that the Jacobian matrix of the map is of full rank, which in particular guarantees that the min-entropy of the source is high, assuming the characteristic of the field is large enough. The field size required by the construction of \cite{dvir-gabizon-wigderson} is $\poly(k,d,n)^{k}$.

Dvir \cite{dvir-varieties} studied a different generalization called \emph{variety sources}, which are uniform distributions over sets $V \subseteq \FF^n$ that are the common zeros of a set of low-degree polynomials. Varieties also have an associated concept of dimension, but unlike the affine case, over finite fields having a large dimension does not guarantee by itself that the set $V$ is large, and thus this condition must be imposed explicitly. Dvir presented two constructions. The first requires exponentially large fields and works for any dimension $k$.
The second requires the variety to have size larger than $|\FF|^{n/2}$, but the field size depends only polynomially on the degree $d$ of the polynomials defining $V$.

Over $\FF_2$, the situation is much more mysterious. This setting is well motivated, since it turns out that explicit constructions of extractors (or even dispersers) for varieties with various parameters would imply new circuit lower bounds. Golovnev, Kulikov and Williams \cite{GKW21} proved multiple such results. One is that explicit extractors for varieties of size at least $2^{\varepsilon n}$ defined by constant degree polynomials would imply lower bounds for general circuits of the form $Cn$ for larger constants $C$ than what is currently known.
They also showed that extractors for varieties of size at least $2^{0.99n}$ defined by polynomials of degree at most $n^{0.01}$ would imply super-linear lower bounds for boolean circuits of depth $O(\log n)$, a long-standing challenge in complexity theory (see also \cite{hrubes-rao-15}).

As for constructions over $\FF_2$, Li and Zuckerman \cite{li-zuckerman-varieties-2019} showed how to use correlation bounds against low-degree polynomials to obtain extractors for variety sources defined by degree $d$ polynomials for $d=O(1)$ and size at least $2^{(1-c_d)n}$ for some constant $c_d$ that depends on $d$. Remscrim \cite{remscrim-varieties-2016} proved that the majority function is an extractor for varieties defined by polynomials of degree at most $n^{\alpha}$ and size at least $2^{n-n^{\beta}}$, assuming $\alpha+\beta < 1/2$. Thus, all the known constructions are not strong enough to imply new circuit lower bounds.

\subsection{Our Results}
\label{sec:results}

\subsubsection{Extractor for Polynomial Images of Varieties}

In this paper, we study the class of \emph{polynomial images of varieties}, which generalizes both variety sources and polynomial sources. Informally, the source is specified by a variety $V \subseteq \FF^r$ and a polynomial map $f: V \to \FF^{n}$, 
and a sample from the source is a random variable $X$ computed by uniformly at random picking an element $x \in V$ and outputting $f(x)$.
We would like to construct an efficient extractor $\Ext : \FF^{n} \to \set{0,1}^m$
that has small error $\eps$ and large output length $m$.
The largest $m$ we can hope for is the min-entropy of the input, which is approximately
$k \log q$, where  $q = |\FF|$ and $k$ is the dimension of the variety $V$ (see Section \ref{sec:prelim-AG} for a definition of this notion). Our main result is a construction of an extractor with $m \approx k\log q$.

Formally defining such sources takes some care, since varieties and their associated complexity parameters are easier to define over algebraically closed fields. As in previous work, we further need to assume some natural non-degeneracy conditions on the variety $V$ and the map $f$. We now describe those sources in more detail.

\paragraph{Polynomial images of variety sources.}
Let $\FF$ be a field.
For $h_1,\dots,h_s\in \FF[X_1,\dots,X_n]$, define 
\[
\mathcal{L}_{h_1,\dots,h_s,\FF}:=\{c_0+c_1 h_1+\dots+c_s h_s: c_0,\dots,c_s\in\FF\}\subseteq \FF[X_1,\dots,X_n],
\]
i.e., $\mathcal{L}_{h_1,\dots,h_s,\FF}$ is the linear span of $h_1,\dots, h_s$ and $1$ over $\FF$.

Denote by $\overline{\FF}$ the algebraic closure of $\FF$.
An \emph{affine variety} $V\subseteq \overline{\FF}^n$ over $\FF$ is the set of common zeros of a set of polynomials in $\FF[X_1,\dots,X_n]$. Two parameters naturally associated with a variety $V$ are its \emph{dimension}, denoted $\dim V$, which equals the length of the maximal chain with respect to inclusion of distinct irreducible subvarieties, and its \emph{degree}, denoted $\deg V$, which is the number of intersection points of the variety with an affine subspace of codimension $\dim V$  in general position (we refer to Section \ref{sec:prelim-AG} for more formal definitions).

\begin{definition}[$(n,k,d)$ algebraic source]\label{defn:algebraic-source}
Let $n,d\in\NN^+$ and $k\in\NN$. 
We say a distribution $D$ over $\FF_q^{n}$ is an \emph{$(n,k,d)$ algebraic source over $\FF_q$} if there exist $r\in\NN$, an affine variety $V\subseteq\overline{\FF}_q^r$ over $\FF_q$, polynomials $h_1,\dots,h_s\in \FF_q[X_1,\dots,X_r]$
with $\deg h_1\geq \dots \geq \deg h_s$, 
and $f_1,\dots,f_n\in \mathcal{L}_{h_1,\dots,h_s,\FF_q}$ such that
$D=f(U_{V(\FF_q)})$, where $f:\overline{\FF}_q^r\to \overline{\FF}_q^{n}$ is the polynomial map defined by $f_1,\dots,f_{n}$, and $U_{V(\FF_q)}$ is the uniform distribution over $V(\FF_q):=V \cap \FF_q^r$, and further, the following conditions hold:
\begin{enumerate}
    \item At least one irreducible component of $V$ of dimension $\dim V$ is absolutely irreducible.
    \item For every irreducible component $V_0$ of dimension $\dim V$ that is absolutely irreducible, the dimension of $\overline{f(V_0)}$ is at least $k$, where $\overline{f(V_0)}\subseteq \overline{\FF}_q^n$ denotes the closure of $f(V_0)$, i.e., the smallest affine variety over $\FF_q$ containing $f(V_0)$.
    \item $\deg V\cdot \prod_{i=1}^{k} \deg h_i \leq d$.\footnote{Note that $\dim \overline{f(V)}\geq k$ by previous conditions. So we necessarily have $s\geq k$ and $\deg h_i\geq 1$ for $i\in [k]$. This also implies $\deg V\leq d$.}
\end{enumerate}
In addition, we say $D$ is an \emph{irreducible} $(n,k,d)$ algebraic source over $\FF_q$ if $V$ can be chosen to be irreducible.
We say $D$ is a \emph{minimal} $(n,k,d)$ algebraic source over $\FF_q$ if $V$ can be chosen to have dimension $k$. 
Finally, we say $D$ is an \emph{irreducibly minimal} $(n,k,d)$ algebraic source over $\FF_q$ if $V$ can be chosen to be irreducible of dimension $k$. 
\end{definition}

The conditions in Definition \ref{defn:algebraic-source} may look a bit contrived at first glance. However, as we now explain, they are quite natural, and indeed some form of them, as observed in previous work, is necessary.

The third condition is simply a convenient way to ``pack'' multiple ``complexity'' parameters of the components of the source that arise in the analysis. That is, $d$ is a single complexity parameter that, in particular, bounds the degree of the variety $V$ and the product of degrees of the polynomial map $f$. Having $d$ as a single parameter simplifies the statements of our theorems and clarifies the dependence between the various parameters: the larger $d$ is, the larger the field size we require and the smaller the output length of the extractor.

The purpose of the first two conditions is to guarantee that our source has enough min-entropy. As observed in previous work \cite{dvir-gabizon-wigderson, dvir-varieties}, it is quite easy to come up with simple varieties $V$ (even of high dimension) or polynomial maps $f$ (even of low degree) such that sources arising as $f(V)$ would have very few points in $\FF_q^n$, so that there will be little to no randomness to extract.

The first condition is analogous to (and, as shown in Appendix \ref{sec:misc}, roughly equivalent to) Dvir's \cite{dvir-varieties} condition that the variety $V$ contains enough points in $\FF_p^n$. The second condition is analogous to (and, over fields of large characteristic, implied by) the full-rank Jacobian condition of Dvir, Gabizon and Wigderson \cite{dvir-gabizon-wigderson}. Thus, not only is some form of conditions 1 and 2 necessary for proving any meaningful results, but moreover, these conditions naturally generalize the conditions imposed by previous related works.

Finally, we note that the name ``$(n,k,d)$ algebraic sources'' suppresses the dependence on the parameter $r$ in the definition, which is the ambient dimension in which the variety $V$ lies. This is because our result, stated next, has no dependence on $r$. Even in the case where $r$ is very large with respect to $n$, $k$ and $d$, our results only depend on the latter three parameters. Further, note that when $r$ is very large, $\dim V$ can also be very large compared with $n$ and $k$. However, as the definition hints, we will reduce this case to the case where $\dim V=k$.

We can now state our main theorem.

\begin{maintheorem}
\label{thm:intro:extractor-main}
Let $n,d\in\NN^+$, $k\in\NN$, and $\epsilon\in (0,1/2]$. 
Let $q$ be a power of a prime $p$. Suppose $q\geq (nd/\epsilon)^c$, where $c>0$ is a large enough absolute constant. Then there exists an explicit $\epsilon$-extractor $\Ext: \FF_q^{n}\to\{0,1\}^m$ for $(n,k,d)$ algebraic sources over $\FF_q$
with output length $m \geq k\log q-4\log\log p-O(\log (nd/\epsilon))$.
\end{maintheorem}

It can be shown that any $(n,k,d)$ algebraic source $D$ over $\FF_q$, where $q\geq (kd)^c$ for a sufficiently large constant $c>0$, is (close to) a distribution with min-entropy at least $k\log q-O(\log d)$. Moreover, this estimate of the min-entropy is tight up to an additive term $O(\log d)$ if $D$ is not an $(n,k+1,d)$ algebraic source over $\FF_q$. See Lemma~\ref{lem:entropy-lower-bound} and Proposition~\ref{prop:entropy-upper-bound}. Therefore, the extractor in Theorem~\ref{thm:intro:extractor-main} extracts most of the min-entropy from $(n,k,d)$ algebraic sources.
In addition, Theorem~\ref{thm:intro:extractor-main} works over finite fields of any characteristic, while the extractors by Dvir, Gabizon, and Wigderson \cite{dvir-gabizon-wigderson} and Dvir \cite{dvir-varieties} 
require large enough characteristics.

As is standard in the literature, by ``explicit'' we mean that the output of the extractor is computable in time $\poly(n, \log q)$ (note that the input length to the extractor is $n \log q$).

Along the way to proving Theorem \ref{thm:intro:extractor-main}, we construct several other algebraic pseudorandom objects which are interesting on their own. We mention some of these constructions when we give an overview of our construction in Section \ref{sec:techniques}.

\subsubsection{Affine Extractors for Quasipolynomally Large Fields with Exponentially Small Error}

Recall that an explicit affine extractor is 
an efficiently computable function $\Ext : \FF^n \to \FF^m$ such that for every affine subspace $V \subseteq \FF^n$ of dimension $k$, and a random variable $X$ uniformly sampled from $V$, $\Ext(X)$ is close to the uniform distribution over $\FF^m$. We would like $m$ to be as close to $k$ as possible and, ideally, the error parameter $\epsilon$ to be exponentially small in $k$.

As mentioned earlier, the extractor of Gabizon and Raz \cite{gabizon-raz} achieves $m=k-1$ and error $\varepsilon$ only polynomially small in the field size $q$. In particular, the error does not decrease with $k$. Bourgain, Dvir and Leeman \cite{bourgain-dvir-leeman-2016} constructed an affine extractor with $m$ arbitrarily close to $k/2 $ and error $q^{-\Omega(k)}$. However, their construction requires $q$ to be slightly super-polynomial in $n$, namely $q=n^{\Omega(\log \log n)}$, and furthermore only works for ``most'' prime fields $\FF_q$. We improve the analysis of their construction and present a construction with identical parameters that works for \emph{all} prime fields, assuming $q=n^{\Omega(\log \log n)}$.

\begin{maintheorem}
\label{thm:affine-extractor-intro}
For every $0 < \beta < 1/2$, there exists a constant $C$ such that the following holds: Let $k \le n$ be integers and $\FF$ be a prime field of size $q \ge n^{C \log \log n}$. Let $m=\beta k$. There exists an efficiently computable function $E : \FF^n \to \FF^m$ which is an affine extractor for min-entropy $k$ with error $q^{-\Omega(k)}$. 
\end{maintheorem}

\subsection{Techniques}
\label{sec:techniques}

Our construction from Theorem \ref{thm:intro:extractor-main} combines several techniques used in previous related constructions, as well as several new ideas which are required to successfully apply these techniques.
It is convenient to think of the construction as proceeding in several steps.

\paragraph*{Preliminary step: decomposing the sources.}
Our definition for algebraic sources (Definition \ref{defn:algebraic-source}) is quite general, and it is convenient to work with slightly ``nicer'' sources. We start by approximating general $(n,k,d)$ algebraic sources as convex combinations of \emph{irreducibly minimal} $(n,k,d)$ algebraic sources. Recall that this means that the variety $V$ is irreducible and has dimension~$k$.

This step is done in Section \ref{sec:decomposition}: we first decompose a general source into a convex combination of irreducible sources in a manner that follows naturally from the decomposition of $V$ itself as a union of irreducible components. We then decompose an irreducible source into irreducibly minimal sources roughly by intersecting it with a linear space of the appropriate dimension. Both parts of the arguments incur a small error.

\paragraph*{First step: extracting a short seed.}
Having reduced to the case of irreducibly minimal sources, we first design an extractor that extracts a small number of bits from the source. One commonly used technique for doing that is to show that the source is an $\epsilon$-biased distribution, i.e., a distribution whose nontrivial Fourier coefficients are all small. Similar methods work when the source is close to such a distribution or to a convex combination of such distributions. Analyzing and bounding the Fourier coefficients is often done using bounds on exponential sums from algebraic geometry, such as Bombieri's estimate (Theorem \ref{thm:bombieri}). We follow this general paradigm as well.

However, the case where the field characteristic is small presents some unique challenges to overcome. We first prove an extension of Bombieri's theorem for small characteristic $p$.  This extension bounds the corresponding exponential sums save for possibly a small set of ``bad'' characters. Hence, we then define and study a more general class than $\epsilon$-biased distributions: $(\epsilon,e)$-biased distributions, which are distributions in which all but at most $e$ of the Fourier coefficients have absolute value at most $\epsilon$. We show that the sources we consider are close to convex combinations of such distributions (for meaningful values of $\epsilon$ and $e$), and construct extractors for such distributions.

Previously, the XOR lemma has been used to construct extractors for $\eps$-biased sources; see, e.g., Rao \cite{rao-2007-bourgain-ext}.
We extend these ideas to the more general and challenging setting of $(\eps, e)$-biased distributions. On the technical level, we construct explicit functions $f: \FF_p^n \to \FF_p^t$ with the following properties: for every nontrivial character $\psi$ of $\FF_p^t$, both the $L_1$ and the $L_\infty$ norms of the Fourier transform of $\psi \circ f$ (which is a function from $\FF_p^n$ to $\mathbb{C}$) are upper bounded by sufficiently small quantities. We in fact present two constructions of such functions $f$. The first is based on standard error-correcting codes over $\FF_p$, and the second is an improved construction based on \emph{rank-metric codes}. Those constructions appear in Section \ref{sec:xor-lemma}.

\paragraph*{Second step: applying a seeded extractor.}

Having extracted a small number of bits, we wish to use them as a \emph{seed} in an application of a seeded extractor on the source to extract almost all the min-entropy. The challenge, of course, is that the seed is correlated with the source, whereas a seeded extractor requires the seed to be independent of the source. Techniques for dealing with these problems were developed in \cite{GRS06, gabizon-raz}, as this is also the general methodology in their extractor constructions. This is done by analyzing the conditional distribution of the source conditioned on any possible output of the seeded extractor with a fixed seed, and showing that it maintains some nice properties. We first analyze the case where the image $f(V)$ of the polynomial map is of full rank inside $\FF^k$, using the \emph{effective fiber dimension theorem}. We then consider the general case. In order to reduce to that case, we apply a \emph{rank extractor} for varieties, a notion we define and develop in this work, building upon previous work which developed rank extractors for linear spaces.

\paragraph*{Rank extractor for varieties.}

Let $V \subseteq \FF^n$ be a $k$-dimensional variety. We would like to obtain a map $E : \FF^n \to \FF^k$ which ``extracts'' all the rank from $V$, in the sense that $E(V) \subseteq \FF^k$ is $k$-dimensional. The first obvious challenge is that $E(V)$ need not necessarily be a variety. It is thus natural in this case to consider the closure of $E(V)$ in $\overline{\FF}^n$ where $\overline{\FF}$ is the algebraic closure of $\FF$.

Previous work has considered the case where $V$ is a linear subspace. In this case, observe that if $E$ is linear, then $E(V)$ is also a linear subspace. However, there clearly cannot be a single map $E$ that preserves the dimension of all linear subspaces, as given any fixed $E$, one could take $V$ to be the kernel of $E$. Therefore, a natural relaxation is to consider \emph{seeded} linear rank extractors, which are collections of linear maps $E_1, \ldots, E_t$ such that for every $V$, most of the maps preserve the dimension. Such objects were first defined and constructed by Gabizon and Raz \cite{gabizon-raz}. Improved and optimal parameters (in terms of the ``seed length,'' i.e., the number of maps) were obtained by Forbes and Shpilka \cite{FS12}, and a systematic study of these objects appears in \cite{forbes-guruswami15}.

In this work, we observe that seeded linear rank extractors for extractors are also seeded linear rank extractors for varieties (see Section \ref{sec:seeded-rank-extractor}). The key insight is that rank extractors (for linear subspaces) preserve the dimensions of the tangent spaces at nonsingular points of the variety, which turns out to be a sufficient criterion.

Linear rank extractors are very useful because they enable us to condense sources that are not full-rank to full-rank sources without increasing the degrees of the polynomial maps. However, it turns out that it is also possible to construct \emph{deterministic} rank extractors for varieties, which we do in Section \ref{sec:deterministic-rank-extractor}. Such extractors are obviously not linear maps, although in our constructions, they are polynomials of fairly small degrees (polynomial in $n$ and in the degree $d$ of the variety). We remark that Dvir \cite{dvir-varieties} constructed such an extractor for one-dimensional varieties, and his extractor is a polynomial of degree exponential in $n$.  In addition, Dvir, Gabizon and Wigderson \cite{dvir-gabizon-wigderson} constructed rank extractors for polynomial sources using a different technique.

Our construction adapts the construction of Dvir, Koll{\'a}r and Lovett \cite{dvir-kollar-lovett-2014}, who constructed different pseudorandom objects called \emph{variety evasive sets}. By modifying their proof, we are able to show that a similar construction yields a deterministic rank extractor for varieties. This essentially follows because their map $\varphi$ satisfies the property that for every low-degree variety $V$ and every point $b \in \FF^k$, the intersection $\varphi^{-1}(b) \cap V$ is a finite set.  Dvir, Koll{\'a}r and Lovett prove it only for the case $b=\mathbf{0}$, but it is not hard to extend it to general $b$.

\paragraph*{Explicit Noether normalization lemmas.} 
As a by-product of the above construction of deterministic rank extractors for varieties, we prove explicit \emph{Noether normalization lemmas} for affine varieties and affine algebras. The Noether normalization lemma \cite{Noe26, Nag62} is a classical result in commutative algebra and algebraic geometry, which states that any affine variety of dimension $k$ admits a surjective \emph{finite morphism} to an affine space of dimension $k$.
We show that the construction in \cite{dvir-kollar-lovett-2014} in fact gives a direct construction of such a finite morphism.
In contrast, the textbook proof of Nagata \cite{Nag62} is iterative and uses polynomials of degrees that are at least doubly exponential in the number of steps of the iteration.

Our proof is inspired by a geometric argument of  Koll{\'a}r, R{\'o}nyai and Szab{\'o} \cite{KRS96}.
See Section~\ref{sec_rank_variety} and Appendix~\ref{sec:NNL} for more details.

\paragraph*{Affine extractors with exponentially small error.}

Our proof of Theorem \ref{thm:affine-extractor-intro} follows a very similar route to the proof of the main theorem of Bourgain, Dvir and Leeman \cite{bourgain-dvir-leeman-2016}, who constructed such an extractor for prime fields $\FF_q$ for ``typical'' primes $q$. Our main contribution is an improved number-theoretic lemma (Proposition \ref{prop:good-degrees}) which shows how to find $n$ distinct integers $d_1, \ldots, d_n$ with desirable number theoretic properties. The proof then proceeds by estimating the Fourier coefficient of the distribution obtained by applying our extractor to a linear subspace using an exponential sum estimate of Deligne, much in the same way as \cite{bourgain-dvir-leeman-2016}.

\subsection{Comparison with Previous Work}
The two works closest to ours are by Dvir \cite{dvir-varieties} and Dvir, Gabizon and Wigderson \cite{dvir-gabizon-wigderson}, both of which construct extractors for sources with algebraic structures.

As mentioned earlier, Dvir, Gabizon and Wigderson \cite{dvir-gabizon-wigderson} study \emph{polynomial sources}, defined by picking an element $x \in \FF_q^k$ uniformly at random and applying a polynomial map $f : \FF_q^k \to \FF_q^n$ of degree at most $d$. This is a special case of the sources we consider when the variety $V$ is taken to be $\FF_q^k$.

They further add the non-degeneracy condition that the \emph{Jacobian} of the mapping $f$, namely, its matrix of partial derivatives, has full rank. This in particular guarantees that the source has a high enough min-entropy. Their main theorem gives an explicit extractor that outputs a constant fraction of the min-entropy over prime fields $\FF_p$ of cardinality $\poly(n,d)^{Ck}$ for some constant $C$. Our construction in Theorem \ref{thm:intro:extractor-main}, on the other hand, works for a larger class of sources, outputs almost all the min-entropy, and works over finite fields of small characteristics as well.

Dvir \cite{dvir-varieties} considers \emph{variety sources}, which he defines as uniform distributions over sets of the type $\set{x : f_1(x)=f_2(x) = \cdots = f_t(x)=0}$ in $\FF_p^n$, where $\deg f_i \le d$ for all $i$.
These sources are also a special case of the type of sources we consider. One should note, however, the different usage of the term ``degree'' in our definitions: Dvir always refers to the degree $\deg f_i$ of the polynomials which define the variety $V$, whereas we refer to the degree $\deg V$ of $V$ as an affine variety, which is often much larger.

Assuming $\dim V=k$ and $|V| \ge p^{k-c}$ for some small constant $c>0$, Dvir's extractor \cite{dvir-varieties} outputs a constant fraction of the min-entropy over prime fields of characteristic $p>d^{Cn^2}$ for some constant $C$. Again, Dvir uses the parameter $d$ differently than we do in Theorem \ref{thm:intro:extractor-main}. In particular, in our construction, the field size $q$ is only polynomial in the parameter $d$ (but $d$ might be exponential in $n$).

As mentioned in the discussion after Definition \ref{defn:algebraic-source}, our assumptions are weaker than those of \cite{dvir-gabizon-wigderson} and \cite{dvir-varieties}. Thus, as our sources is more general, the characteristic in our results can be arbitrary, and our conclusions are stronger (since we extract more output bits), it follows that in particular our result subsumes the extractors of \cite{dvir-gabizon-wigderson} and \cite{dvir-varieties}.

Dvir \cite{dvir-varieties} also presents a different construction that outputs a very small number of bits from very large varieties over small fields. This construction is incomparable with our results.

On the more technical level, we discuss a particular feature of our proof that distinguishes it from \cite{dvir-gabizon-wigderson, dvir-varieties} and, in particular, allows us to extend the output length.

For simplicity, consider the case of $(1,1,d)$ algebraic sources. As mentioned in Section \ref{sec:techniques}, we first prove an extension of Bombieri's estimate that holds even if the characteristic $p$ is small: if $p$ is small, this result implies that a $(1,1,d)$ algebraic source $D$ over $\FF_q$ is a convex combination of 
$(\epsilon,d)$-biased sources. That is, we allow a few large Fourier coefficients.  Then we use the machinery developed in Section~\ref{sec:low-bias} to extract randomness from $D$. On the other hand, if $p$ is large enough, then $D$ has no large nontrivial Fourier coefficients; it is $\epsilon$-biased. In this case, the XOR lemma is sufficient, as argued in \cite{dvir-gabizon-wigderson, dvir-varieties}.

To apply Bombieri's estimate to a high-dimensional affine variety $V$, we follow \cite{dvir-gabizon-wigderson, dvir-varieties} and decompose $V$ into a family of affine curves $C_i$ such that the polynomial $f$ that does not vanish identically on $V$ still does not vanish on most $C_i$. 

In \cite{dvir-gabizon-wigderson}, this is achieved using an argument based on the Jacobian criterion for algebraic independence, but it works only when the characteristic $p$ is large. Instead of using this argument, we use the decomposition of $(n,k,d)$ algebraic sources into irreducibly minimal $(n,k,d)$ algebraic sources proved in Section~\ref{sec:decomposition}, whose proof is based on the effective fiber dimension theorem (Theorem~\ref{thm:effective-FDT-general}) and works for any characteristic.

The last idea we introduce is the use of the effective Lang--Weil bound (Theorem~\ref{theorem:lang-weil}), which allows us to extract almost $\log q$ bits.
To explain the idea, consider an affine variety $V\subseteq\AA_{\FF_q}^n$ and write $V(\FF_q)$ as a disjoint union of $C_i(\FF_q)$ for a family of affine curves $C_i$ over $\FF_q$. Let $f$ be a low-degree polynomial and assume for simplicity that $f$ is non-constant on every $C_i$. Let $\chi$ be a nontrivial character of $\FF_q$.
The following win-win argument was used in \cite{dvir-gabizon-wigderson} to bound the bias $\delta:=\left|\Ex_{x\in V(\FF_q)}[\chi(f(x))]\right|$:
For a curve $C_i$, if $|C_i(\FF_q)|$ is small, say $|C_i(\FF_q)|\leq \Delta$ for some threshold $\Delta$, then its contribution to the bias $\delta$ is small assuming that $V$ has many rational points. On the other hand, if $C_i(\FF_q)>\Delta$, then Bombieri's estimate (Lemma~\ref{lem:bias-over-irred-curve}), together with the fact that
\[
 \left|\Ex_{x\in C_i(\FF_q)}[\chi(f(x))]\right|=\frac{\left|\sum_{x\in C_i(\FF_q)}[\chi(f(x))]\right|}{|C_i(\FF_q)|}
 \leq \frac{\left|\sum_{x\in C_i(\FF_q)}[\chi(f(x))]\right|}{\Delta},
\]
implies that $\left|\Ex_{x\in C_i(\FF_q)}[\chi(f(x))]\right|$ is small. Considering all curves $C_i$ shows that the bias is small. We note that no information about $|C_i(\FF_q)|$ was used in this win-win argument. For this reason, the choice of threshold $\Delta$ cannot be too large or too small, and the resulting extractors only extract a constant fraction of $\log q$ bits. To improve the output length, we observe that the effective Lang--Weil bound (Theorem~\ref{theorem:lang-weil}) together with Lemma~\ref{lem:not-absolute-irreducible} gives more information about $|C_i(\FF_q)|$. In particular, for an irreducible affine curve $C$, the number $|C(\FF_q)|$ is either close to $q$ or very small, depending on whether $C$ is absolutely irreducible.
Exploiting this fact yields an explicit construction of deterministic extractors that output almost $\log q$ bits.

\subsection{Open Problems}
\label{sec:open}

While improving the dependence on any of the parameters in our construction remains an open problem, in our opinion, the main challenge is reducing the field size.
In our construction for polynomial images of varieties (Theorem \ref{thm:intro:extractor-main}), we require field size $\poly(n,1/\epsilon,d)$. We stress that for certain varieties, $d$ can be exponential in $n$ (although it is by no means necessarily so).
Can we construct extractors for significantly smaller fields, perhaps even constant size?

As mentioned above, over very small fields, such as $\FF_2$, certain Ramsey-theoretic lower bounds imply that constructions such as ours that work for any min-entropy cannot exist. A key reason to study $\FF_2$ is that explicit extractors with certain parameters imply new circuit lower bounds.

In our construction of new affine extractors (Theorem \ref{thm:affine-extractor-intro}), we obtain a field size that is slightly super-polynomial in $n$. It is a very appealing open problem to reduce the field size to a polynomial in $n$.

A related problem is reducing the degree of our deterministic rank extractor. In Section \ref{sec:deterministic-rank-extractor}, we construct a deterministic rank extractor for varieties whose degree is $\poly(n,d)$ for degree $d$ varieties. Reducing the degree, perhaps to depend only on $d$, would help lower the field size requirement for the extractor for polynomial images of varieties to depend only on the degree.

We end with two general questions.  Can our constructions or techniques help in designing extractors for larger and more general classes of sources, either algebraic or complexity-theoretic? Do our constructions have any complexity-theoretic implications, such as lower bounds for certain models of computation?

\section{Notations and Preliminaries}
Let $\NN=\{0,1,\dots\}$, $\NN^+=\{1,2,\dots,\}$, and $[n]=\{1,2,\dots,n\}$ for $n\in\NN$.
Write $\ZZ_n$ for the cyclic group $\{0,1,\dots,n-1\}$ with addition modulo $n$.

The cardinality of a set $S$ is denoted by $|S|$.
We also use $|c|$ to denote the absolute value of a number $c\in\CC$.
Denote by $\log x$ the base 2 logarithm of $x$, and by $\ln x$ the natural logarithm of $x$.
For sets $A$ and $B$, denote by $A\setminus B$ the set difference $\{x\in A: x\not\in B\}$.
The restriction of a map $f: A\to B$ to a subset $A'\subseteq A$ is denoted by $f|_{A'}$, which is a map from $A'$ to $B$.

A \emph{formal Laurent series} over a field $\FF$ has the form 
\[
h(T)=c_{i} T^{i}+ c_{i+1} T^{i+1}+\dots
\]
where $i\in \ZZ$ and $c_j\in \FF$ for $j\geq i$.  Denote by $\ord(f)$ the least degree of the terms that appear in $f$, i.e., $f=c_0 T^{\ord(f)}+c_1 T^{\ord(f)+1}+\cdots$ where $c_0\neq 0$. If $f=0$, then let $\ord(f)=+\infty$.
The set of formal Laurent series over $\FF$ is a field, denoted by $\FF((T))$.  We say $f\in\FF((T))$ is a \emph{formal power series} over $\FF$ if $\ord(f)\geq 0$. The set of formal power series over $\FF$ is a subring of $\FF((T))$, denoted by $\FF[[T]]$.

We write $x\sim D$ if $x$ is sampled from a distribution $D$.
The \emph{support} of a distribution $D$ over a finite set $\Omega$ is $\supp(D):=\{a\in \Omega: \Pr[D=a]\neq 0\}$.
For an event $A$ that occurs with a nonzero probability under a distribution $D$, write $D|_A$ for the distribution of $D$ conditioned on $A$.
The product distribution of two distributions $D, D'$ is denoted by $D\times D'$.
The \emph{statistical distance} between two distributions $D, D'$ over a finite set $\Omega$ is defined to be
\[
\Delta(D,D'):=\max_{A\subseteq \Omega} |\Pr[D\in A]-\Pr[D'\in A]|.
\]
Two distributions $D$ and $D'$ are \emph{$\epsilon$-close} if their statistical distance is at most $\epsilon$, and we write $D=_{\epsilon} D'$ for this statement.

The uniform distribution over a finite set $S$ is denoted by $U_S$.
For $n\in\NN$, denote by $U_n$ the uniform distribution over $\{0,1\}^n$.

The \emph{min-entropy} of a distribution $D$ over a finite set $\Omega$ is
\[
H_{\mathrm{min}}(D):=-\log(\max_{a\in \Omega}\Pr[D=a]).
\]
We say $D$ is a \emph{$k$-source} if $H_{\mathrm{min}}(D)\geq k$.

Let $\Omega$ and $B$ be finite sets, and let $\mathcal{D}$ be a class of distributions over $\Omega$.
A function $\Ext: \Omega\to B$ is said to be a \emph{(deterministic) $\epsilon$-extractor} for $\mathcal{D}$ if $\Ext(D)=_\epsilon U_B$ for all $D\in\mathcal{D}$.
A function $\Ext: \Omega\times\{0,1\}^\ell\to B$ is said to be a \emph{seeded $\epsilon$-extractor} for $\mathcal{D}$ if $\Ext(D\times U_\ell)=_\epsilon U_B$ for all $D\in\mathcal{D}$, where $\ell\in\NN$ is called the \emph{seed length} of $\Ext$.

\paragraph{Facts about distributions.}

The following lemmas are standard and can be found in, e.g., \cite[Section~2]{Sha08}.

\begin{lemma}
Let $f: A\to B$ be a map between finite sets.
Let $D$ and $D'$ be distributions over $A$. If $D=_\epsilon D'$, then $f(D)=_\epsilon f(D')$.
\end{lemma}

\begin{lemma}\label{lem:approx-2}
Let $D=(D_1,D_2)$ be a joint distribution over a finite product set $A\times B$.
Suppose $T$ is a subset of $A$ such that $\Pr[D_1\in T]\geq 1-\epsilon_1$ and $D_2|_{D_1=a}=_{\epsilon_2} D_2'$ for some distribution $D_2'$ over $B$ and all $a\in T\cap \supp(D_1)$.
Then $D=_{\epsilon_1+\epsilon_2} D_1\times D_2'$.
\end{lemma}

\section{Sources with Low Bias and Their Extractors}\label{sec:low-bias}

In this section, we consider several natural extensions of $\varepsilon$-biased sources which are useful for our extractor constructions. We then show how to extract randomness from such sources.

\subsection{$(\epsilon,e)$-Biased Sources}

Let $A$ be a finite abelian group and let $\widehat{A}$ denote the dual group of $A$, that is, the group of characters over $A$. A distribution $D$ over $A$ is \emph{$\epsilon$-biased} if $|\Ex[\chi(D)]|\leq \epsilon$ for all nontrivial characters $\chi\in\widehat{A}$. This is a standard definition, introduced in \cite{naor-naor}, which has been immensely useful in the construction of extractors and in the theory of pseudorandomness in general.

We now introduce two natural generalizations. We say $D$ is \emph{$(\epsilon,e)$-biased} if $|\Ex[\chi(D)]|\leq \epsilon$ for all but at most $e$ characters $\chi\in\widehat{A}$. 
And we say $D$ is \emph{strongly $(\epsilon,e)$-biased} if the set of $\chi\in\widehat{A}$ satisfying $|\Ex[\chi(D)]|> \epsilon$ is contained in an abelian subgroup of $A$ of size at most $e$. The usefulness of the latter definition will be clear shortly.

Note that if $A$ is a vector space over a finite field of characteristic $p$, then an $(\epsilon,e)$-biased distribution is also strongly $(\epsilon,p^e)$-biased.

We have the following easy observation, which shows that strongly $(\epsilon,e)$-biased sources generalize affine sources (of low codimension).

\begin{lemma}
An affine source of codimension at most $k$ in $\FF_p^n$ is strongly $(0, p^k)$-biased.
\end{lemma}

The following proposition states that $(\epsilon, e)$-biased sources are (close to) distributions with high min-entropy. 

\begin{proposition}\label{prop:high-entropy}
Let $D$ be an $(\epsilon,e)$-biased distribution over $A$. Then $D$ is $\epsilon'$-close to a $k$-source if $\epsilon'\geq 2^k(\epsilon^2 + |A|^{-1}e)$.
In particular, this holds for  $k=\min\{2\log(1/\epsilon), \log |A|-\log e\}-\log (2/\epsilon')$.
\end{proposition}
\begin{proof}
View $D$ as a function $D: A\to \RR$ such that $D(x)=\Pr[D=x]$ for $x\in A$.
For $\chi\in\widehat{A}$, we have $\widehat{D}(\chi)=\Ex_{x\in A}[D(x)\overline{\chi(x)}]=|A|^{-1}\left(\sum_{x\in A}  \Pr[D=x] \cdot \overline{\chi(x)}\right)=|A|^{-1}\overline{\Ex[\chi(D)]}$.
So $|\widehat{D}(\chi)|\leq |A|^{-1}$ for all $\chi\in\widehat{A}$, and $|\widehat{D}(\chi)|\leq |A|^{-1}\epsilon$ for all but at most $e$ characters $\chi\in\widehat{A}$.
By Parseval's identity, 
\[
\sum_{x\in A} |D(x)|^2=
|A|\cdot\left(\Ex_{x\in A} \left[|D(x)|^2\right]\right)=|A|\cdot\left(\sum_{\chi\in\widehat{A}} |\widehat{D}(\chi)|^2\right)\leq \epsilon^2 + |A|^{-1}e.
\]
Let $A' \subseteq A$ be the set of $x\in A$ such that $D(x)>2^{-k}$.
Then 
\[
\epsilon^2 + |A|^{-1}e\geq \sum_{x\in A} |D(x)|^2\geq \sum_{x\in A'} |D(x)|^2\geq \left(\sum_{x\in A'} D(x) \right)2^{-k}.
\]
So the total probability mass contributed by $x\in A'$ is bounded by $\epsilon'=2^k(\epsilon^2 + |A|^{-1}e)$.
This implies that $D$ is $\epsilon'$-close to a $k$-source. 
\end{proof}

Next, suppose that $A$ and $B$ are finite groups. We wish to bound the bias of conditional distributions over $A$ (or $B$), assuming bounds on the bias of a distribution over $A \times B$. We begin with the following technical calculation.

\begin{lemma}\label{lem:bias-of-conditional}
Let $A$ and $B$ be finite abelian groups. Identifying $\widehat{A}\times\widehat{B}$ with $\widehat{A\times B}$ so that $(\chi, \theta)(x,y)=\chi(x)\theta(y)$ for $(x,y)\in A\times B$ and $(\chi,\theta)\in \widehat{A}\times\widehat{B}$.
Let $D=(D_1,D_2)$ be a joint distribution over $A\times B$.
For $x\in \supp(D_1)$ and $\theta\in\widehat{B}$, we have
\[
\Ex[\theta(D_2|_{D_1=x})]
=\sum_{\chi\in \widehat{A}}\Pr[D_1=x]^{-1}\cdot |A|^{-1}\cdot \overline{\chi(x)}\cdot \Ex[(\chi, \theta)(D)].
\]
\end{lemma}

\begin{proof}
Define $\delta_x: A\to \{0,1\}$ to be the indicator function such that $\delta_x(z)=1$ if and only if $z=x$.
Then
\begin{align*}
\Ex[\theta(D_2|_{D_1=x})]
&=\Pr[D_1=x]^{-1}\cdot \Ex[\delta_x(D_1)\theta(D_2)]\\
&=\Pr[D_1=x]^{-1}\cdot \Ex\left[\left(\sum_{\chi\in\widehat{A}}\widehat{\delta_x}(\chi)\chi(D_1)\right)\theta(D_2)\right]\\
&=\sum_{\chi\in\widehat{A}} \Pr[D_1=x]^{-1} \cdot \widehat{\delta_x}(\chi)\cdot \Ex[\chi(D_1)\theta(D_2)]\\
&=\sum_{\chi\in \widehat{A}}\Pr[D_1=x]^{-1}\cdot |A|^{-1}\cdot \overline{\chi(x)}\cdot \Ex[(\chi, \theta)(D)]
\end{align*}
where the last equality uses the fact $\widehat{\delta_x}(\chi)=\Ex_{z\in A}[\delta_x(z)\overline{\chi(z)}]=|A|^{-1}\cdot \overline{\chi(x)}$.
\end{proof}

As a corollary, we bound the bias of the marginal distribution $D_2$ conditioned on any value of $D_1$.

\begin{corollary}\label{cor:conditional-has-low-bias}
Use the notations in Lemma~\ref{lem:bias-of-conditional}.
Let $\epsilon, \epsilon'>0$.
Assume that every character $\chi\in\widehat{A\times B}\cong \widehat{A}\times\widehat{B}$ satisfying $\Ex[\chi(D)]>\epsilon$ is contained in the subgroup $\widehat{A}\times\{1\}$.  
Then with probability at least $1-\epsilon'$ over $x\sim D_1$, the conditional distribution $D_2|_{D_1=x}$ is $|A|\epsilon /\epsilon'$-biased.
\end{corollary}

\begin{proof}
Let $T$ be the set of $x\in A$ satisfying $\Pr[D_1=x]\leq |A|^{-1}\epsilon'$.
Then $\Pr[D_1\in T]=\sum_{x\in T}\Pr[D_1=x]\leq \epsilon'$.
So it suffices to show that $|\Ex[\theta(D_2|_{D_1=x})]|\leq |A|\epsilon /\epsilon'$ holds  for every $x\in A\setminus T$ and every nontrivial character $\theta$ of $B$.

Consider $x\in A\setminus T$ and a nontrivial character $\theta$ of $B$. As $x\not\in T$, we have $\Pr[D_1=x]> |A|^{-1}\epsilon'$.
By Lemma~\ref{lem:bias-of-conditional}, 
\begin{align*}
|\Ex[\theta(D_2|_{D_1=x})]|&=\left|\sum_{\chi\in \widehat{A}}\Pr[D_1=x]^{-1}\cdot |A|^{-1}\cdot \overline{\chi(x)}\cdot \Ex[(\chi, \theta)(D)]\right|\\
&\leq \left|\sum_{\chi\in \widehat{A}}\overline{\chi(x)}\cdot \Ex[(\chi, \theta)(D)]\right| / \epsilon'\\
&\leq \sum_{\chi\in \widehat{A}} |\Ex[(\chi, \theta)(D)]|/\epsilon'.
\end{align*}
Note that for $\chi\in\widehat{A}$, $(\chi,\theta)$ is not in the subgroup $\widehat{A}\times\{1\}$ since $\theta\in\widehat{B}$ is nontrivial.
So $|\Ex[(\chi, \theta)(D)]|\leq \epsilon$ by assumption.
It follows that $|\Ex[\theta(D_2|_{D_1=x})]|\leq |A|\epsilon /\epsilon'$, as desired.
\end{proof}

\subsection{Extraction via the XOR Lemma and Rank-Metric Codes}
\label{sec:xor-lemma}

We need the following form of Vazirani's XOR lemma, taken from \cite{rao-2007-bourgain-ext}.
\begin{lemma}[XOR lemma]\label{lem:xor}
Every $\epsilon$-biased distribution over a finite abelian group $A$ is $\epsilon |A|^{1/2}$-close to the uniform distribution over $A$.
\end{lemma}

\paragraph{Large characteristic.}
Over fields of large characteristic, we use the mod-$M$ function in our constructions as an extractor for low-bias distributions, in a similar manner to \cite{dvir-gabizon-wigderson} and \cite{dvir-varieties}. We follow the treatment in \cite{rao-2007-bourgain-ext}.

For a finite abelian group $A$ and a function $h : A \to \mathbb{C}$, define $\Lone{h} = \sum_{x \in A} |h(x)|$ and $\Linf{h} = \max_{x \in A} |h(x)|$.

\begin{lemma}[{\cite[Lemma~4.3]{rao-2007-bourgain-ext}}]\label{lem:fourier-to-dist}
Let $A$ and $B$ be finite abelian groups. Let $D$ be an $\epsilon$-biased distribution over $A$. 
Suppose $f: A \to B$ is a map such that for every character $\psi$ of $B$, we have that
$\Lone{\widehat{\psi\circ f}}\leq\tau$.
Then $f(D)$ is $\epsilon'$-close to $f(U_A)$, where $\epsilon'=\tau\epsilon |B|^{1/2}$.
\end{lemma}

\begin{lemma}[{\cite[Lemma~4.4]{rao-2007-bourgain-ext}}]\label{lem:l1-mod-m}
Let $f: \ZZ_N\to \ZZ_M$ be the map sending $a\bmod N$ to $a\bmod M$ for $a\in\{0,1,\dots,N-1\}$. Let $\psi$ be a character of $\ZZ_M$. Then $\Lone{\widehat{\psi\circ f}}\leq c\log N$, where $c$ is an absolute constant.
\end{lemma}

When $p$ is large but $\FF_q$ is possibly non-prime, we  simply apply the mod-$M$ function to the last $\FF_p$-coordinate of $\FF_q$ and use the following corollary of Lemma~\ref{lem:l1-mod-m}.

\begin{corollary} \label{cor:l1-mod-m-general}
Let $f: \ZZ_N^t\to \ZZ_N^{t-1}\times \ZZ_M$ be the map that sends $(a_1,\dots,a_{t-1},a\bmod N)$ to $(a_1,\dots,a_{t-1},a\bmod M)$ for $(a_1,\dots,a_{t-1},a)\in \ZZ_N^{t-1} \times \{0,1,\dots,N-1\}$. Let $\psi$ be a character of $\ZZ_N^{t-1}\times \ZZ_M$. Then $\Lone{\widehat{\psi\circ f}}\leq c\log N$, where $c$ is an absolute constant.
\end{corollary}

Combining Lemma~\ref{lem:fourier-to-dist} and Corollary~\ref{cor:l1-mod-m-general} gives the following lemma, which allows us to extract randomness from $\epsilon$-biased sources over $\FF_q$.

\begin{lemma}\label{lem:mod-m-extractor}
Let $f: \ZZ_N^t\to \ZZ_N^{t-1}\times \ZZ_M$ be the map in Corollary~\ref{cor:l1-mod-m-general}.
Then for every $\epsilon$-biased distribution $D$ over $\ZZ_N^t$, $f(D)$ is $\epsilon'$-close to the uniform distribution over $\ZZ_N^{t-1}\times \ZZ_M$, where $\epsilon'=\epsilon\cdot (N^{t-1}M)^{1/2}\cdot c\log N+M/N$ and $c$ is an absolute constant.
\end{lemma}

\begin{proof}
By Lemma~\ref{lem:fourier-to-dist} and Corollary~\ref{cor:l1-mod-m-general}, $f(D)$ is $\epsilon\cdot (N^{t-1}M)^{1/2}\cdot c\log N$-close to $f(U)$.
For each $b\in \ZZ_M$, the number of $a\in\{0,1,\dots, N-1\}$ satisfying $a\bmod M=b$ is either $\lfloor N/M\rfloor$ or $\lceil N/M\rceil$, and its difference from $N/M$ is bounded by one. So $f(U)$ is $M/N$-close to the uniform distribution over $\ZZ_N^{t-1}\times \ZZ_M$, and the lemma follows.
\end{proof}

\paragraph{Small characteristic.}
The XOR lemma requires the distribution to be $\epsilon$-biased. 
However, when the characteristic is small, we need to deal with the more general class of $(\epsilon, e)$-biased distributions, where $e$ is small.
The following lemma states that we can extract randomness from such sources using a function $f$, provided that the $L_1$ and $L_\infty$ norms of the Fourier transforms of certain functions $\psi\circ f$ are reasonably bounded.

\begin{lemma}\label{lem:error-bound}
Let $A$ and $B$ be finite abelian groups.
Let $D$ be an $(\epsilon,e)$-biased distribution over $A$.  
Let $f: A\to B$ be a map such that for every nontrivial character $\psi\in \widehat{B}$, $\Lone{\psi\circ f} \le a_1$ and $\Linf{\psi \circ f} \le a_{\infty}$.
Then $f(D)$ is $\epsilon'$-close to the uniform distribution over $B$, where $\epsilon'=(a_1\epsilon+a_\infty e)|B|^{1/2}$.
\end{lemma}

\begin{proof}
Let $\psi$ be a nontrivial character of $B$.
By the XOR lemma, we just need to prove that 
\[
|\Ex[\psi(f(D))]|\leq a_1\epsilon+a_\infty e.
\]

Let $E$ be the set of $\chi\in\widehat{A}$ such that $|\Ex[\chi(D)]|>\epsilon$.
Then $|E|\leq e$.
Writing $\psi\circ f=\sum_{\chi\in \widehat{B}} \widehat{\psi\circ f}(\chi)\cdot \chi$, we have
\begin{align*}
|\Ex[\psi(f(D))]|&=|\Ex[(\psi\circ f)(D)]|\\
&\leq \sum_{\chi\in \widehat{A}} |\widehat{\psi\circ f}(\chi)|  \cdot |\Ex[ \chi(D))]|\\
&=\sum_{\chi\in \widehat{A}\setminus E} |\widehat{\psi\circ f}(\chi)|  \cdot |\Ex[ \chi(D))]|
+\sum_{\chi\in E} |\widehat{\psi\circ f}(\chi)|  \cdot |\Ex[ \chi(D))]|\\
&\leq \left(\sum_{\chi\in \widehat{A}\setminus E} |\widehat{\psi\circ f}(\chi)|\right) \cdot \epsilon
+ |E|\cdot \max_{\chi\in E} |\widehat{\psi\circ f}(\chi)| \\
&\leq a_1\epsilon+a_\infty e
\end{align*}
as desired.
\end{proof}

We turn to constructing functions $f$ with the properties required by Lemma \ref{lem:error-bound}.
As a warm-up, we first give a construction based on the inner product function and error-correcting codes. Later, we give an improved construction based on \emph{rank-metric codes}.

\paragraph{Construction of $f$ based on the inner product function and error-correcting codes.}
We construct $f$ by adapting the inner product function $IP(x_1,\dots,x_{2r})=\sum_{i=1}^r x_{2i-1}x_{2i}$.
Suppose $r, n>0$ are integers such that $2r\leq n$.
Let $\pi:\FF_p^n\to\FF_p^{2r}$ be a projection over $\FF_p$.
Let $C\subseteq \FF_p^r$ be a linear code over $\FF_p$ of dimension $t$ and relative distance $\delta$ with a generating matrix $G=(c_{i,j})\in \FF_p^{t\times r}$ (i.e., $C$ is the row space of $G$). 
For $i\in [t]$, define $f_i: \FF_p^n\to \FF_p$ by
\[
f_i(x)=\sum_{j=1}^r c_{i,j} y_{2j-1} y_{2j}, ~\text{where}~ (y_1,\dots,y_{2r})=\pi(x)\in\FF_p^{2r}.
\]
Let $f=(f_1,\dots,f_t):\FF_p^n\to \FF_p^t$.

\begin{lemma}
Let $f$ be as above.
For every nontrivial character $\psi\in \widehat{\FF_p^t}$, we have $\Lone{\widehat{\psi\circ f}}\leq p^r$ and $\Linf{\widehat{\psi\circ f}}\leq p^{-\delta r}$.
\end{lemma}
\begin{proof}
Note that $\psi\circ f$ may be viewed as a function in $\pi(x)\in\FF_p^{2r}$ and hence can be written as a linear combination of the characters of $\FF_p^{2r}$. 
On the other hand, the projection $\pi:\FF_p^n\to\FF_p^{2r}$ induces an injective group homomorphism $\iota:\widehat{\FF_p^{2r}}\hookrightarrow\widehat{\FF_p^n}$ via $\chi\mapsto \chi\circ\pi$.
This means the support of $\widehat{\psi\circ f}$ is contained in the subgroup $\iota(\widehat{\FF_p^{2r}})\subseteq\widehat{\FF_p^n}$ of size $p^{2r}$.
By the Cauchy--Schwarz inequality,
\[
 \Lone{\widehat{\psi\circ f}}=\sum_{\chi\in\iota(\widehat{\FF_p^{2r}})} |\widehat{\psi\circ f}(\chi)|
 \leq \left(\sum_{\chi\in\iota(\widehat{\FF_p^{2r}})} |\widehat{\psi\circ f}(\chi)|^2\right)^{1/2}\cdot p^r=\Ltwo{\widehat{\psi\circ f}}\cdot p^r=p^r.
\]
where the last equality uses the fact that $\psi\circ f$ takes values in the unit circle and hence $\Ltwo{\widehat{\psi\circ f}}=1$ by Parseval's identity.

We now prove the second claim, i.e., $|\widehat{\psi\circ f}(\chi)|\leq p^{-\delta r}$ for $\chi\in \widehat{\FF_p^n}$.
We may also assume $\chi\in \iota(\widehat{\FF_p^{2r}})$ since the support of $\widehat{\psi\circ f}$ is contained in $\iota(\widehat{\FF_p^{2r}})$.

Fix a nontrivial character $\sigma$ of $\FF_p$.
Then $\psi(x_1,\dots,x_t)=\sigma(\sum_{i=1}^t u_i x_i)$ for some nonzero vector $u=(u_1,\dots,u_t)\in\FF_p^t$.
Let $(v_1,\dots,v_r)=u G$, which is a nonzero codeword in $C$. 
Note that by definition,
\[
(\psi\circ f)(x)=\sigma\left(\sum_{i=1}^t u_i f_i(x)\right)=\sigma\left(\sum_{i=1}^r v_i y_{2i-1} y_{2i}\right), ~\text{where}~ (y_1,\dots,y_{2r})=\pi(x)\in\FF_p^{2r}.
\]

As $\chi\in \iota(\widehat{\FF_p^{2r}})$,
we have $\chi(x)=\sigma(\sum_{i=1}^{2r} w_i y_i)$ for some $(w_1,\dots,w_{2r})\in\FF_p^{2r}$, where $(y_1,\dots,y_{2r})=\pi(x)$.
Then 
\begin{align*}
|\widehat{\psi\circ f}(\chi)|&=
\left|\Ex_{x\in\FF_p^n}\left[(\psi\circ f)(x)\overline{\chi(x)}\right]\right|\\
&=\left|\Ex_{(y_1,\dots,y_{2r})\in\FF_p^{2r}}\left[\sigma\left(\sum_{i=1}^r v_i y_{2i-1} y_{2i}\right)\cdot\overline{\sigma\left(\sum_{i=1}^{2r} w_i y_i\right)}\right]\right|\\
&=\left|\Ex_{(y_1,\dots,y_{2r})\in\FF_p^{2r}}\left[\sigma\left(\sum_{i=1}^r (v_i y_{2i-1} y_{2i}-w_{2i-1} y_{2i-1} - w_{2i} y_{2i})\right)\right]\right|\\
&=\prod_{i=1}^r \left|\Ex_{y,y'\in\FF_p}\left[\sigma\left(P_i(y,y')\right)\right]\right|
\end{align*}
where $P_i(y,y'):=v_i y y'-w_{2i-1} y - w_{2i} y'=(v_iy-w_{2i})y'-w_{2i-1}y$.
Note that $\left|\Ex_{y,y'\in\FF_p}\left[\sigma\left(P_i(y,y')\right)\right]\right|\leq 1/p$ whenever $v_i\neq 0$. (This holds since $\Ex_{y'\in\FF_p}\left[\sigma\left(P_i(y,y')\right)\right]$ is zero when $v_i\neq 0$ and $y$ is assigned a value in $\FF_p$ different from $w_{2i}/v_i$, in which case $P_i(y,y')$ is a degree-$1$ polynomial in $y'$.)
As $C$ is a linear code of relative distance $\delta$, there are at least $\delta r$ indices $i\in [r]$ for which $v_i\neq 0$. It follows that $|\widehat{\psi\circ f}(\chi)|\leq p^{-\delta r}$.
\end{proof}

By choosing $C\subseteq\FF_p^r$ to be an explicit asymptotically good linear code over $\FF_p$ (e.g., an expander code), we obtain the following result.

\begin{corollary}
There exist absolute constants $c, c'>0$ such that the following holds.
For integers $r,n>0$ such that $2r\leq n$, there exists an explicit function $\FF_p^n\to\FF_p^t$, where $t\geq c r$, such that $\Lone{\widehat{\psi\circ f}}\leq p^{r}$
and 
$L_\infty(\widehat{\psi\circ f})\leq p^{-c' r}$
for every nontrivial character $\psi\in\widehat{\FF_p^t}$.
\end{corollary}

\paragraph{Construction of $f$ based on rank-metric codes.}

The bilinear maps used in the above construction may be viewed as diagonal matrices with many nonzeros on the diagonal.
We observe that the analysis only requires the matrices to have a high rank, and they do not have to be diagonal or even square matrices. 
This leads to the following improved construction of $f$ based on \emph{rank-metric codes}, which are subspaces of matrices such that every non-zero matrix in the subspace has a high rank.
Here we use an optimal construction of rank-metric codes discovered independently by Delsarte \cite{Del78} and Gabidulin \cite{Gab85}.

\begin{definition}[Delsarte--Gabidulin codes]
Let $k\leq r\leq s$ be positive integers. 
Fix $g_1,\dots,g_r\in \FF_{p^s}$ that are linearly independent over $\FF_p$.
Also fix $\tau:\FF_{p^s}\to \FF_p^s$ that is an isomorphism of vector spaces over $\FF_p$. 
For $u=(u_1,\dots,u_k)\in \FF_{p^s}^k$, let $f_u:=\sum_{i=1}^k u_i X^{p^{i-1}}\in\FF_{p^s}[X]$ and $M_u:=(\tau(f_u(g_1)),\dots, \tau(f_u(g_r)))\in\FF_{p}^{s\times r}$, where each $\tau(f_u(g_i))\in\FF_p^s$ is viewed as a column vector.
\end{definition}

\begin{lemma}\label{lem:gabidulin-bound}
For any nonzero $u\in \FF_{p^s}^k$, the matrix $M_u$ has rank at least $r-k+1$.
\end{lemma}

\begin{proof}
Let $Z\subseteq\FF_{p^s}$ be the set of roots of $f_u$ in $\FF_{p^s}$. We have $|Z|\leq \deg(f_u)\leq p^{k-1}$.
Note that for a column vector $y=(y_1,\dots,y_r)\in\FF_p^r$, we have $M_u y=\sum_{i=1}^r y_i \tau(f_u(g_i))=\tau(f_u(\sum_{i=1}^r y_i g_i))$ since $\tau$ is $\FF_p$-linear and $f_u$ is a linearized polynomial.
So $M_u y=0$ iff $\sum_{i=1}^r y_i g_i\in Z$. As $g_1,\dots,g_r$ are linearly independent over $\FF_p$, the map $y\mapsto \sum_{i=1}^r y_i g_i$ is injective. So there are at most $|Z|\leq p^{k-1}$ choices of $y$ such that $M_u y=0$. We conclude that the right kernel of $M_u$ has size at most $p^{k-1}$ and hence has dimension at most $k-1$. In other words, the rank of $M_u$ is at least $r-k+1$.
\end{proof}

The map $u\mapsto M_u$ is an $\FF_p$-linear injective map by Lemma~\ref{lem:gabidulin-bound}. 
So  $\{M_u: u\in \FF_{p^s}^k\}$ is a linear space of dimension $ks$ over $\FF_p$, and any nonzero matrix in this space has rank at least $r-k+1$ by Lemma~\ref{lem:gabidulin-bound}. We record this in the following corollary.

\begin{corollary}\label{cor:explicit-matrices}
Let $k\leq r\leq s$ be positive integers.
Let $t\leq ks$.
There exist explicit matrices $M_1,\dots,M_t\in\FF_p^{s\times r}$ such that for any nonzero $c=(c_1,\dots,c_t)\in\FF_p^t$, the matrix $\sum_{i=1}^t c_i M_i$ has rank at least $r-k+1$.
\end{corollary}

\begin{construction}\label{construction-from-rank-codes}
Suppose $r,n>0$ are integers such that $2r\leq n$. Let $s=n-r\geq r$.
Let $k\in [r]$ and $t\in [ks]$. And let $M_1,\dots,M_t$ be as in Corollary~\ref{cor:explicit-matrices}.
Identify $\FF_p^n$ with $\FF_p^s\times\FF_p^r$.
Then for $i\in [t]$, define $f_i: \FF_p^n\cong \FF_p^s\times\FF_p^r\to \FF_p$ to be the $\FF_p$-bilinear map sending $(x,y)\in\FF_p^s\times\FF_p^r$ to $x^T M_i y$, where $x$ and $y$ are viewed as column vectors.
Finally, let $f=(f_1,\dots,f_t):\FF_p^n\to \FF_p^t$.
\end{construction}

\begin{lemma}\label{lem:norm-bounds}
Let $f$ be as above.
For every nontrivial character $\psi\in \widehat{\FF_p^t}$, we have $\Lone{\widehat{\psi\circ f}}\leq p^r$ and $\Linf{\widehat{\psi\circ f}}\leq p^{-(r-k+1)}$.
\end{lemma}

\begin{proof}
Fix a nontrivial character $\sigma$ of $\FF_p$.
Identify $\FF_p^n$ with $\FF_p^s\times\FF_p^r$.
Then $\psi\circ f$ has the form $(x,y)\mapsto \sigma(x^T M y)$ where $M\in\FF^{s\times r}$ is a nontrivial linear combination of the matrices $M_1,\dots,M_t$. Let $R$ be the rank of $M$. Then $r-k+1\leq R\leq \min\{s,r\}=r$ by Corollary~\ref{cor:explicit-matrices}.

Let $U, V\subseteq\FF_p^s$ such that $U$ is the left kernel of $M$ and $V$ is a complement of $U$. 
We have $\codim_{\FF_p} U=\dim_{\FF_p} V=R$.
Identify $\FF_p^s$ with $U\times V$.
Identify $\widehat{\FF_p^n}$ with $\widehat{U}\times \widehat{V}\times\widehat{\FF_p^r}$ so that $(\chi_1,\chi_2,\chi_3)$ sends $(u,v,y)\in U\times V\times\FF_p^r\cong \FF_p^n$ to $\chi_1(u)\chi_2(v)\chi_3(y)$.
Consider $\chi=(\chi_1,\chi_2,\chi_3)\in \widehat{\FF_p^n}$.
We have
\begin{equation}\label{eq:fourier_coeff}
\begin{aligned}
\widehat{\psi\circ f}(\chi)
&=\Ex_{x\in\FF_p^n}\left[(\psi\circ f)(x)\overline{\chi(x)}\right]\\
&=\Ex_{(u,v,y)\in U\times V\times\FF_p^r}\left[\sigma((u+v)^TMy)\overline{\chi_1(u)\chi_2(v)\chi_3(y)}\right]\\
&=\Ex_{(u,v,y)\in U\times V\times\FF_p^r}\left[\sigma(v^TMy)\overline{\chi_1(u)\chi_2(v)\chi_3(y)}\right]\\
&=\Ex_{(v,y)\in  V\times\FF_p^r}\left[\sigma(v^TMy)\overline{\chi_2(v)\chi_3(y)}\right]\Ex_{u\in U}\left[\overline{\chi_1(u)}\right]
\end{aligned}
\end{equation}
Note $\Ex_{u\in U}\left[\overline{\chi_1(u)}\right]=0$ whenever $\chi_1\neq 1$. So the support of $\widehat{\psi\circ f}$ is contained in the subgroup $\{1\}\times \widehat{V}\times\widehat{\FF_p^r}$ of size $p^{R+r}\leq p^{2r}$. Also note that $\psi\circ f$ takes values in the unit circle of $\CC$ and hence $\sum_{\chi\in\widehat{\FF_p^n}} |\widehat{\psi\circ f}(\chi)|^2=\Ex_{x\in\FF_p^n}\left[(\psi\circ f)(x)^2\right] =1$ by Parseval's identity.
Then by the Cauchy--Schwarz inequality,
\[
 \Lone{\widehat{\psi\circ f}}=\sum_{\chi\in\{1\}\times \widehat{V}\times\widehat{\FF_p^r}} |\widehat{\psi\circ f}(\chi)|
 \leq \left(\sum_{\chi\in\widehat{\FF_p^n}} |\widehat{\psi\circ f}(\chi)|^2\right)^{1/2}\cdot p^r=p^r.
\]
This proves the first claim.

We now prove the second claim, i.e., $|\widehat{\psi\circ f}(\chi)|\leq p^{-(r-k+1)}$ for $\chi\in \widehat{\FF_p^n}$.
We may assume $\chi=(\chi_1,\chi_2,\chi_3)\in \{1\}\times \widehat{V}\times\widehat{\FF_p^r}$ since otherwise $\widehat{\psi\circ f}(\chi)=0$.
Choose $w\in V$ such that $\chi_2$ sends $v\in V$ to $\sigma(v^T w)$.
Then
\begin{align*}
|\widehat{\psi\circ f}(\chi)|
&=\left|\Ex_{(v,y)\in  V\times\FF_p^r}\left[\sigma(v^TMy)\overline{\chi_2(v)\chi_3(y)}\right]\right|\\
&=\left|\Ex_{(v,y)\in  V\times\FF_p^r}\left[\sigma(v^T(My-w))\overline{\chi_3(y)}\right]\right|\\
&=\left|\Ex_{y\in \FF_p^r}\left[\Ex_{v\in V}\left[\sigma(v^T(My-w))\right]\cdot \overline{\chi_3(y)} \right]\right|\\
&\leq \Pr_{y\in\FF_p^r}[My=w].
\end{align*}
The first equality holds by \eqref{eq:fourier_coeff} and the assumption $\chi_1=1$. The last inequality above holds since $\Ex_{v\in V}\left[\sigma(v^T(My-w))\right]=0$ when $My\neq w$.
Finally, note $\Pr_{y\in\FF_p^r}[My=w]$ is either zero or $p^{-R}$, depending on whether $My=w$ has a solution. In either case, we have $|\widehat{\psi\circ f}(\chi)|\leq p^{-R}\leq p^{-(r-k+1)}$.
\end{proof}

Assuming that $n$ is sufficiently large and $\epsilon$ is not too much larger than $p^{-n/2}$, the construction above gives explicit deterministic extractors that extract a constant fraction of min-entropy from $(\epsilon,e)$-biased sources over $\FF_p^n$, as stated by the following theorem.

\begin{theorem}\label{thm:extractor-constant-fraction}
Let
$n,t,e\in\NN^+$ and $d,\epsilon'>0$ such that $n\geq c$ and $t\log p\leq c^{-1}n\log p-2 \log(de/\epsilon')$, where $c>0$ is a large enough absolute constant. Let $\epsilon=dp^{-n/2}$.
Then there exists an explicit deterministic $\epsilon'$-extractor $f: \FF_p^n\to \FF_p^t$ for $(\epsilon, e)$-biased sources distributed over $\FF_p^n$.
\end{theorem}

\begin{proof}
Choose $r=n/4$, $s=n-r$, and
$k=r/2$. 
(For ease of readability, we omit floor and ceiling functions.)
Note $t\leq ks$ assuming $c$ is large enough.
Let $f:\FF_p^n\to \FF_p^t$ be as in Construction~\ref{construction-from-rank-codes}.
Then $\Lone{\widehat{\psi\circ f}}\leq p^r$ and $\Linf{\widehat{\psi\circ f}}\leq p^{-r/2}$ by Lemma~\ref{lem:norm-bounds}.
By Lemma~\ref{lem:error-bound}, for any $(\epsilon, e)$-biased distribution $D$ over $\FF_p^n$, $f(D)$ is $\epsilon''$-close to the uniform distribution over $\FF_p^t$, where $\epsilon''=(p^r d p^{-n/2}+p^{-r/2}e)\cdot p^{t/2}$.
As $c$ is large enough, the conditions in the theorem imply $p^r d p^{-n/2}\cdot p^{t/2}\leq \epsilon'/2$ and $p^{-r/2}e\cdot p^{t/2}\leq \epsilon'/2$. So $\epsilon''\leq \epsilon'$ and hence $f(D)$ is $\epsilon'$-close to the uniform distribution.
\end{proof}

We also obtain an explicit construction of deterministic extractors that extract most min-entropy from very dense affine sources, and more generally, $(0,e)$-biased sources. It is interesting to compare and constant the theorem below with our construction in Section \ref{sec:affine-extractors}, which requires a large field size but works for arbitrary min-entropy, and is also based on very different ideas.

\begin{theorem}\label{thm:dense-affine-extractor}
Let $n,t,e\in\NN^+$ and $\epsilon\in (0,1)$ such that $t\leq n-3-2\log_p(e/\epsilon)$.
Then there exists an explicit deterministic $\epsilon$-extractor $f: \FF_p^n\to\FF_p^t$ for $(0, e)$-biased sources, and in particular, for affine sources of codimension at most $\log_p e$.
\end{theorem}

\begin{proof}
Note that $n\geq t+3\geq 4$.
Choose $r=\lfloor n/2\rfloor\geq (n-1)/2$, $s=n-r$, and
$k=2\leq r$. 
Note $t\leq ks$ since $t\leq n$ and $s\geq n/2$.
Let $f:\FF_p^n\to \FF_p^t$ be as in Construction~\ref{construction-from-rank-codes}.
Then $\Linf{\widehat{\psi\circ f}}\leq p^{-(r-k+1)}$ by Lemma~\ref{lem:norm-bounds}.
By Lemma~\ref{lem:error-bound}, for any $(0, e)$-biased distribution $D$ over $\FF_p^n$, $f(D)$ is $\epsilon'$-close to the uniform distribution over $\FF_p^t$, where $\epsilon'=p^{-(r-k+1)}\cdot e\cdot p^{t/2}$.
Finally, note that $\epsilon'\leq \epsilon$ by the choice of $t$.
\end{proof}

\paragraph{Extracting most min-entropy from strongly $(\epsilon,e)$-biased sources.}

We now construct deterministic extractors that extract most min-entropy from strongly $(\epsilon,e)$-biased sources.

The following lemma allows us to reduce to the case of affine sources.
 
\begin{lemma}\label{lem:reduction-to-affine}
Let $\pi: \FF_p^n\to \FF_p^m$ be a surjective $\FF_p$-linear map.
Let $D$ be a strongly $(\epsilon, e)$-biased distribution over $\FF_p^n$. Then $\pi(D)$ is $\epsilon'$-close to a convex combination of affine sources of codimension at most $\log_p e$ in $\FF_p^m$, where $\epsilon'=2(p^{m/2}\cdot e\cdot \epsilon)^{1/2}$.
\end{lemma}

\begin{proof}
Let $A$ be a subspace of $\FF_p^n$ of size at most $e$ and $B$ be a complement of $A$ such that by identifying $\FF_p^n$ with $A\times B$ and $\widehat{\FF_p^n}$ with $\widehat{A}\times\widehat{B}$, every nontrivial character $\chi$ of $\widehat{\FF_p^n}$ satisfying $|\Ex[\chi(D)]|>\epsilon$ is in the subgroup $\widehat{A}\times\{1\}$.
This is possible as $D$ is strongly $(\epsilon,e)$-biased.
Identify $A$ and $B$ with the subgroups $A\times\{0\}$ and $\{0\}\times B$ of $A\times B$ respectively.
Then the codimension of $B$ in $\FF_p^n$ is at most $\log_p e$.
Let $D_1$ and $D_2$ be the marginal distributions of $D$ over $A$ and $B$ respectively.

Let $\epsilon_0>0$, whose value is determined later.
By Corollary~\ref{cor:conditional-has-low-bias}, for $x$ sampled from $D_1$, with probability at least $1-\epsilon_0$, the conditional distribution $D_2|_{D_1=x}$ is $|A|\epsilon /\epsilon_0$-biased. Fix $x\in\supp(D_1)$ such that $D_2|_{D_1=x}$ is $|A|\epsilon /\epsilon_0$-biased.
We claim that $\pi(D|_{D_1=x})$ is $\epsilon_1$-close to an affine source of codimension at most $\log_p e$, where $\epsilon_1=p^{m/2}|A|\epsilon/\epsilon_0$.
Note that $\pi(D|_{D_1=x})=\pi(x)+\pi(D_2|_{D_1=x})$. 
So to prove the claim, we just need to show that $\pi(D_2|_{D_1=x})$ is $\epsilon_1$-close to an affine source of codimension at most $\log_p e$.

Let $H=\pi(B)$.
Then $\pi|_{B}: B\to H$ is a surjective linear map.
So the map $\chi\mapsto \chi\circ \pi$ sends a nontrivial character of $H$ to a nontrivial character of $B$.
As $D_2|_{D_1=x}$ is an $|A|\epsilon/\epsilon_0$-biased distribution over $B$, we see that $\pi(D_2|_{D_1=x})$ is an $|A|\epsilon/\epsilon_0$-biased distribution over $H$.
By the XOR lemma (Lemma~\ref{lem:xor}), $\pi(D_2|_{D_1=x})$ is $|H|^{1/2}|A|\epsilon/\epsilon_0$-close to the uniform distribution over $H$.
Note
\begin{align*}
\dim H=\dim B - \dim(B\cap \ker \pi)&\geq \dim B-\dim (\ker \pi)\\
&= \dim B - (n-m)
=m-\codim B.
\end{align*}
So the codimension of $H$ in $\FF_p^m$ is at most $\codim B\leq \log_p e$.
Therefore, the distribution  $\pi(D_2|_{D_1=x})$ is $|H|^{1/2}|A|\epsilon/\epsilon_0$-close to an affine source of codimension at most $\log_p e$. This proves the above claim as $|H|^{1/2}|A|\epsilon/\epsilon_0\leq p^{m/2}|A|\epsilon/\epsilon_0=\epsilon_1$.

We have shown that for $x$ sampled from $D_1$, with probability at least $1-\epsilon_0$, the distribution $D|_{D_1=x}$ is $\epsilon_1$-close to an affine source of codimension at most $\log_p e$.
It follows that $D$ is $(\epsilon_0+\epsilon_1)$-close to a convex combination of affine sources of codimension at most $\log_p e$, where $\epsilon_1=p^{m/2}|A|\epsilon/\epsilon_0$. Finally, choose $\epsilon_0=(p^{m/2}|A|\epsilon)^{1/2}$ so that 
\[
\epsilon_0+\epsilon_1=2\epsilon_0=2(p^{m/2}\cdot |A|\cdot \epsilon)^{1/2}\leq 2(p^{m/2}\cdot e\cdot \epsilon)^{1/2}. \qedhere
\]
\end{proof}

\begin{theorem}\label{thm:extractor-for-strongly-biased}
Let $n,t,e$ be positive integers and $\epsilon, \epsilon'\in (0,1)$.
Let $n'=\min\{\lfloor 2\log_p(1/\epsilon)-2\log_p(16e/\epsilon'^2)\rfloor, n\}$.
Suppose $t\leq n'-3-2\log_p(2e/\epsilon')$.
Then there exists an explicit  $\epsilon'$-extractor $\Ext: \FF_p^n\to\FF_p^t$ for strongly $(\epsilon, e)$-biased sources.
\end{theorem}

\begin{proof}
We construct the extractor $\Ext$ as follows.
\begin{itemize}
    \item Let $\pi: \FF_p^n\to\FF_p^{n'}$ be an explicit linear surjective map.
    \item Let $f:\FF_p^{n'}\to\FF_p^t$ be an explicit deterministic $\epsilon'/2$-extractor for affine sources of codimension at most $\log_p e$.
    \item Finally, let $\Ext=f\circ \pi$.
\end{itemize}
Here the existence of $f$ is guaranteed by Theorem~\ref{thm:dense-affine-extractor} as $t\leq n'-3-2\log_p(2e/\epsilon')$. 
We claim that $\Ext:\FF_p^n\to \FF_p^t$ is a deterministic $\epsilon'$-extractor for strongly $(\epsilon, e)$-biased sources.

Consider an arbitrary strongly $(\epsilon, e)$-biased source $D$.
By the choice of $n'$, we have $2(p^{n'/2}\cdot e\cdot \epsilon)^{1/2}\leq \epsilon'/2$.
So by Lemma~\ref{lem:reduction-to-affine}, $\pi(D)$ is $\epsilon'/2$-close to a convex combination of affine sources of codimension at most $\log_p e$.
Then by the affine extractor property of $f$, we know $\pi(D)$ is $\epsilon'$-close to the uniform distribution. This proves the claim.
\end{proof}

\begin{remark}
By Proposition~\ref{prop:high-entropy}, an $(\epsilon, e)$-biased source over $\FF_p^n$ is $\epsilon'$-close to a source of min-entropy at least $\min\{2\log(1/\epsilon), n\log p-\log e\}-\log(2/\epsilon')$.
The output bit-length $t\log p$ in Theorem~\ref{thm:extractor-for-strongly-biased} matches this bound up to an additive term $O(\log e+\log(1/\epsilon')+\log p)$.
\end{remark}

\begin{remark}
As $(\epsilon, e)$-biased sources over $\FF_p^n$ are strongly $(\epsilon, p^e)$-biased, Theorem~\ref{thm:extractor-for-strongly-biased} also gives deterministic extractors for $(\epsilon, e)$-biased sources at the cost of replacing $e$ with $p^e$ in the theorem statement.
It is interesting to ask for deterministic extractors, explicit or not, that extract most min-entropy from $(\epsilon, e)$-biased sources (up to the lower bound in  Proposition~\ref{prop:high-entropy})  without this exponential blow-up. We note that the extractors in Theorem~\ref{thm:extractor-constant-fraction} could extract a constant fraction of the min-entropy from $(\epsilon, e)$-biased sources assuming that $\epsilon$ is small enough.
\end{remark}

\section{Preliminaries on Algebraic Geometry}
\label{sec:prelim-AG}

In this section, we introduce preliminaries and notations on algebraic geometry. One can also refer to a standard text, e.g., \cite{Sha13, Vak17}.
 
 \subsection{Terminology}

All rings in this paper are commutative rings with unity. A proper ideal $I$ of a ring $R$ is \emph{prime} if $ab\in I$ implies $a\in I$ or $b\in I$ for $a,b\in R$. 
This is equivalent to the condition that $R/I$ is an integral domain.
An ideal $I$ of  $R$ is \emph{radical} if $a^m\in I$ implies $a\in I$ for every $a\in R$ and $m\in\NN^+$.

\paragraph{Affine varieties.}
Let $\FF$ be a field and let $\overline{\FF}$ be its algebraic closure. 
For $n\in\NN$, the \emph{affine $n$-space $\AA^n_\FF$  over $\FF$} is 
 the set $\overline{\FF}^n$ equipped with the {\em Zariski topology}, defined as follows. A subset $U\subseteq \AA_\FF^n$ is \emph{closed} if it is the set of common zeros of a set of polynomials in $\FF[X_1,\dots, X_n]$. And $U\subseteq \AA_\FF^n$ is \emph{open} if its complement is closed.
 A closed subset $V\subseteq \AA^n_\FF$ is said to be an \emph{affine variety over $\FF$}, and its elements are called \emph{points} of $V$.\footnote{When $\FF$ is not algebraically closed, affine varieties (and affine spaces) are often defined in a different way such that each point of an affine variety is not a single point in $\overline{\FF}^n$, but an orbit consisting of all the conjugates of a point in $\overline{\FF}^n$ under the natural action of the automorphism group of $\overline{\FF}$ over $\FF$. Our definition instead follows that in \cite{dvir-varieties}, which suffices for us and is more elementary. One can switch between these two definitions by splitting the orbits into points in $\overline{\FF}^n$ and vice versa. One notable difference between the two definitions is that, in our definition, points are not necessarily closed in the Zariski topology over $\FF$ when the field $\FF$ is not algebraically closed, while they are always closed in the other definition.
  }
  We say $V$ is \emph{defined by} a set $S\subseteq\FF[X_1,\dots, X_n]$ if it is the set of common zeros of the polynomials in $S$.

We often write  $\AA^n$ instead of $\AA^n_\FF$ when $\FF$ is algebraically closed and clear from the context.

A point $a\in \AA^n_\FF$  is said to be a \emph{rational point} if the coordinates of $a$ are all in $\FF$. 
For an affine variety $V$ over $\FF$, denote by $V(\FF)$ the set of rational points in $V$. Each rational point is closed (as a set containing the point itself) in the Zariski topology.

For an affine variety $V\subseteq \AA^n_\FF$ over $\FF$ defined by $S\subseteq\FF[X_1,\dots,X_n]$, and an extension field $\KK$ of $\FF$,
denote by $V_{\KK}$ the affine variety $V'\subseteq \AA^n_{\KK}$ over $\KK$ defined by $S$.

\begin{remark*}
Note that under our definitions, $V_{\overline{\FF}}$ is the same as $V$ as subsets of $\overline{\FF}^n$.
\end{remark*}

For two affine varieties $V, V'\subseteq \AA^n_\FF$ over a field $\FF$, we say $V$ is an \emph{affine subvariety} of $V'$ if $V\subseteq V'$.
We say  $V$ is the affine subvariety of $V'$ \emph{defined by} $S\subseteq\FF[X_1,\dots,X_n]$ if it is the set of common zeros in $V'$ of the polynomials in $S$. 

\paragraph{Irreducibility and absolute irreducibility.}

An affine variety $V$ over a field $\FF$  is \emph{irreducible} if it is nonempty and cannot be written
as the union of finitely many proper affine subvarieties over $\FF$. Otherwise, we say $V$ is \emph{reducible}.
A subvariety $V_0$ of an affine variety $V$ is an \emph{irreducible component} of $V$ if $V_0$ is irreducible and maximal (with respect to set inclusion) for this property.
Every affine variety can be uniquely represented as the union of finitely many irreducible components.

An affine variety $V$ over $\FF$  is \emph{absolutely irreducible}  if $V_{\overline{\FF}}$ is irreducible. 
By definition, the Zariski topology over $\overline{\FF}$ is finer than that over $\FF$, i.e., if a set is closed over $\FF$, then it is also closed over $\overline{\FF}$. So absolute irreducibility implies irreducibility.

Let $V\subseteq\AA^n_\FF$ be an irreducible variety over $\FF$. The automorphism group $\mathrm{Aut}(\overline{\FF}/\FF)$ of $\overline{\FF}$ over $\FF$ acts naturally on $V_{\overline{\FF}}$ such that $\tau\in \mathrm{Aut}(\overline{\FF}/\FF)$ sends a point $a=(a_1,\dots,a_n)\in V_{\overline{\FF}}$ to $(\tau(a_1),\dots,\tau(a_n))$. Every $\tau\in \mathrm{Aut}(\overline{\FF}/\FF)$ permutes the irreducible components of $V_{\overline{\FF}}$.
When $\FF=\FF_q$, the Frobenius automorphism $\sigma: x\mapsto x^q$ of $\overline{\FF}_q$ over $\FF_q$ generates a cyclic group that acts transitively on the set of irreducible components of $V_{\overline{\FF}}$. In particular, if $V$ is not absolutely irreducible, then $\sigma$ sends each irreducible component of $V_{\overline{\FF}}$ to a different irreducible component. 

\paragraph{The ideal-variety correspondence.}
For an ideal $I$ of $\FF[X_1,\dots,X_n]$, denote by $V(I)$ the affine subvariety of $\AA^n_\FF$ over $\FF$ defined by $I$. 
Define $V(f_1,\dots, f_k)=V(\langle f_1,\dots, f_k\rangle)$ for $f_1,\dots,f_k\in \FF[X_1,\dots,X_n]$.
For an affine variety $V\subseteq \AA^n$ over $\FF$, denote by $I(V)$ the ideal of $\FF[X_1,\dots,X_n]$ consisting of all the polynomials vanishing on $V$. The ideal $I(V)$, and in fact every ideal of $\FF[X_1,\dots, X_n]$, has a finite generating set by \emph{Hilbert's Basis Theorem}.

 The map $V\mapsto I(V)$ is an inclusion-reversing one-to-one correspondence between the affine subvarieties of $\AA^n_\FF$ over $\FF$ and the radical ideals of $\FF[X_1,\dots,X_n]$, with the inverse map $I\mapsto V(I)$. 
An affine variety $V$ is irreducible iff $I(V)$ is a prime ideal.

For an affine variety $V\subseteq \AA^n_\FF$ over $\FF$, define 
\[
\FF[V]:=\FF[X_1,\dots,X_n]/I(V),
\] 
which is called the \emph{coordinate ring} of $V$. 
When $V$ is irreducible, the ideal $I(V)$ is prime and $\FF[V]$ is an integral domain.
In this case, define the  \emph{function field} of $V$, denoted by $\FF(V)$, to be  the field of fractions of $\FF[V]$.

\paragraph{Dimension.}
The \emph{dimension} of an irreducible affine variety $V$ over a field $\FF$ is defined to be the largest integer $m$ such that there exists a chain of irreducible affine subvarieties $\emptyset\subsetneq V_0\subsetneq V_1\subsetneq\cdots\subsetneq V_m=V$ over $\FF$. More generally, the dimension of a nonempty affine variety $V$, denoted by $\dim V$, is the maximal dimension of its irreducible components. We have $\dim V=\dim V_{\overline{\FF}}$ for a nonempty affine variety $V$ over $\FF$.

The dimension of an irreducible affine variety $V$ over $\FF$ also equals the \emph{transcendence degree} of $\FF(V)/\FF$, i.e., the largest cardinality of an algebraically independent subset of $\FF(V)$ over $\FF$.  In particular, $\dim \AA^n_\FF=n$ as its function field $\FF(X_1,\dots,X_n)$ has transcendence degree $n$ over $\FF$.

A nonempty affine variety is \emph{equidimensional} if its irreducible components have the same dimension. Equidimensional affine varieties of dimension one are also called \emph{(affine) curves}.

The following fact is a geometric version of \emph{Krull's principal ideal theorem}. See, e.g., \cite[\S 12.3]{Vak17}.

\begin{lemma}\label{lem:krull-PIT}
Let $V\subseteq\AA^n_\FF$ be an irreducible affine variety over a field $\FF$. Let $f\in\FF[X_1,\dots,X_n]$ such that $f$ does not vanish identically on $V$. Then either $V\cap V(f)=\emptyset$ or $V\cap V(f)$ is equidimensional of dimension $\dim V-1$.
\end{lemma}

\paragraph{Degree.} Let $\FF$ be an algebraically closed field. For an irreducible affine variety $V\subseteq \AA^n_\FF$ over $\FF$, the \emph{degree} of $V$ in $\AA^n_\FF$, denoted by $\deg V$,  is the number of intersections of $V$ with an affine subspace of $\AA^n_\FF$ of codimension $\dim V$ in general position.  Following \cite{HS80, Hei83},  for an affine variety $V\subseteq \AA^n_\FF$ over $\FF$ with irreducible components $V_1,\dots, V_k$, we let $\deg V:=\sum_{i=1}^k \deg V_i$.
For an affine variety $V$ over an arbitrary field $\FF$, let $\deg V:=\deg V_{\overline{\FF}}$. 

For a nonzero polynomial $f\in \FF[X_1,\dots,X_n]$, we have $\deg (V(f))\leq \deg f $ (and the equality holds if $f$ is squarefree).

The following version of \emph{B\'ezout's inequality} is very useful for us.

\begin{lemma}[{B\'ezout's inequality \cite{HS80, Hei83}}]\label{lem:bezout}
Let $V, V'\subseteq\AA^n_\FF$ be affine varieties over a field $\FF$. Then 
\[
\deg(V\cap V')\leq \deg V\cdot \deg V'.
\]
\end{lemma}

\paragraph{Morphisms between affine varieties over a field $\FF$.}
Let $\FF$ be a field. Let $f_1,\dots,f_m\in\FF[X_1,\dots,X_n]$, which define a map $f: \AA^n_{\FF}\to \AA^m_{\FF}$ sending $a\in \AA^n_{\FF}$ to $f(a)=(f_1(a),\dots,f_m(a))\in  \AA^m_{\FF}$.
Now suppose  $V\subseteq \AA^n_\FF$ and $V'\subseteq \AA^m_\FF$ are affine varieties over $\FF$ such that $f(V)\subseteq V'$. Then the map $f$ restricts to a map $\varphi: V\to V'$.
It is associated with an $\FF$-algebra homomorphism 
\[
\varphi^\sharp: \FF[V']=\FF[Y_1,\dots,Y_m]/I(V')\to \FF[V]=\FF[X_1,\dots,X_n]/I(V),
\]
which sends $Y_i+I(V')$ to $f_i(X_1,\dots,X_n)+I(V)$ for $i\in [m]$.
Note that $\varphi^\sharp$ simply sends a function on $V'$ to its composition with $\varphi$, which is a function on $V$.

We say the pair $(\varphi, \varphi^\sharp)$ is a \emph{morphism}  from $V$ to $V'$ (over $\FF$) and it is \emph{defined by $f_1,\dots,f_m$}.
For simplicity, we usually suppress $\varphi^\sharp$ and denote the morphism by $\varphi$ when there is no confusion.

A morphism between affine spaces is also called a \emph{polynomial map} in this paper.
Let $\varphi: \AA^n_\FF\to\AA^m_\FF$ be a polynomial map.
If $V\subseteq \AA^m_\FF$ is an affine variety defined by a set $S$ of polynomials, then $\varphi^{-1}(V)\subseteq \AA^n_\FF$ is the affine variety defined by $\varphi^\sharp(S)$.
So the preimage of a closed (resp.\ open) set under a polynomial map  is closed (resp.\ open).

\paragraph{Zariski closure and dominant morphisms.}

For a set $S\subseteq \AA^n_\FF$, the (\emph{Zariski}-)closure of $S$, denoted by $\overline{S}$, is the smallest closed set containing $S$, i.e., the intersection of all affine subvarieties of $\AA^n_\FF$ that contain $S$. We say $S$ is \emph{dense} in $V$ if $\overline{S}=V$.

Suppose $V\subseteq  \AA^n_\FF$ is an affine variety over $\FF$ and $\varphi: V \to \AA^m_\FF$ is a morphism defined by polynomials $f_1,\dots,f_m\in \FF[X_1,\dots,X_n]$.
The image $\varphi(V)$ is not necessarily a closed set in $\AA^n_\FF$. To understand the closure $\overline{\varphi(V)}$, note that a polynomial $P\in \FF[Y_1,\dots,Y_m]$ vanishes on $\varphi(V)$ (or equivalently, on $\overline{\varphi(V)}$) iff the composition of $P$ with $\varphi$ vanishes on $V$, i.e., $\varphi^\sharp(P)=0$. So the ideal of $\FF[Y_1,\dots,Y_n]$ defining $\overline{\varphi(V)}$ is precisely the kernel of $\varphi^\sharp: \FF[Y_1,\dots,Y_n]\to \FF[V]$.
Let $r_i=X_i+I(V)\in \FF[X_1,\dots,X_n]/I(V)=\FF[V]$ for $i\in [n]$.
Then the coordinate ring $\FF[\overline{\varphi(V)}]=\FF[Y_1,\dots,Y_m]/I\left(\overline{\varphi(V)}\right)$ of $\overline{\varphi(V)}$ may be identified with $\FF[f_1(r_1,\dots,r_n),\dots,f_m(r_1,\dots, r_n)]\subseteq\FF[V]$ via $Y_i+I\left(\overline{\varphi(V)}\right)\mapsto f_i(r_1,\dots,r_n)$.

We say a morphism $\varphi: V\to V'$ between affine varieties is \emph{dominant} if $\overline{\varphi(V)}=V'$.
If $\varphi: V\to V'$ is a dominant morphism between affine varieties and $V$ is irreducible, then $V'$ is also irreducible.

We also need the following lemma.

\begin{lemma}\label{lem:dominant-projection}
Let $V\subseteq\AA^n_\FF$ be an irreducible affine variety over a field $\FF$.
Let $f_1,\dots,f_s,g_1,\dots,g_t\in\FF[X_1,\dots,X_n]$.
The polynomials $f_1,\dots,f_s$ (resp. $f_1,\dots,f_s,g_1,\dots,g_t$) define a polynomial map $\pi_1:\AA^n_\FF\to \AA^{s}_\FF$ (resp. $\pi_2:\AA^n_\FF\to\AA^{s+t}_\FF$).
Then there exist distinct $i_1,\dots,i_k\in [t]$, where $k=\dim \overline{\pi_2(V)}-\dim \overline{\pi_1(V)}$, such that the polynomial map $\pi:\AA^n_\FF\to\AA^{s+k}_\FF$ defined by $f_1,\dots,f_s,g_{i_1},\dots,g_{i_k}$ satisfies $\dim \overline{\pi(V)}=\dim\overline{\pi_2(V)} = \dim\overline{\pi_1(V)} + k$.
\end{lemma}

\begin{proof}
Let $r_i=X_i+I(V)\in \FF[X_1,\dots,X_n]/I(V)=\FF[V]$ for $i\in [n]$.
Let $\bar{f}_i=f(r_1,\dots,r_n)$ for $i\in [s]$ and $\bar{g}_i=g(r_1,\dots,r_n)$ for $i\in [t]$.
The function fields of $\overline{\pi_1(V)}$ and $\overline{\pi_2(V)}$ are $\KK_1:=\FF(\bar{f}_1,\dots,\bar{f}_s)$ and $\KK_2:=\FF(\bar{f}_1,\dots,\bar{f}_s,\bar{g}_1,\dots,\bar{g}_t)$ respectively. So $\bar{g}_1,\dots,\bar{g}_t$ is a generating set of $\KK_2/\KK_1$.
The lemma then follows from the fact in field theory that a generating set of a field extension $\KK_2/\KK_1$ always contains a subset that is a transcendence basis of $\KK_2/\KK_1$ (see \cite[Chapter~VIII, Theorem~1.1]{Lan02}).
\end{proof}

\paragraph{Fibers of morphisms.}

Let $\varphi: V\to V'$ be a morphism between affine varieties over $\FF$.
For $b\in V'(\FF)$, the preimage $\varphi^{-1}(b)=\{a\in V: \varphi(a)=b\}$ is an affine subvariety of $V$ over $\FF$. We call $\varphi^{-1}(b)$ the \emph{fiber of $\varphi$ over $b$}, or the \emph{fiber of $V$ over $b$} if $\varphi$ is clear from the context.

The following theorem, known as the \emph{fiber dimension theorem}, relates the dimension of a general fiber of a dominant morphism $\varphi: V\to V'$ between  irreducible affine varieties $V, V'$ to the dimension of $V$ and that of $V'$.

 \begin{theorem}[Fiber dimension theorem]\label{thm_dimfiber}
Suppose $\varphi: V\to V'$ is a dominant morphism between irreducible affine
varieties over an algebraically closed field $\FF$.
Then for every $b\in \varphi(V)$ and every irreducible component $Z$ of $\varphi^{-1}(b)$, it holds that
\[
\dim Z\geq \dim V-\dim V'.
\]
Moreover, there exists $U\subseteq \varphi(V)$ such that $U$ is a dense open subset of $V'$ and $\dim \varphi^{-1}(b)=\dim V-\dim V'$ holds for all $b\in U$.
\end{theorem}

See, e.g., \cite[\S I.6.3, Theorem~7]{Sha13} for a proof. 
We remark that the above version of the fiber dimension theorem can be generalized in several ways, but this version suffices for us.

We also need the notion of generic fibers. Let $\varphi: V\to V'$ be a dominant morphism defined by  $f_1,\dots,f_m\in\FF[X_1,\dots,X_n]$ between irreducible affine varieties $V\subseteq\AA^n_\FF$ and $V'\subseteq\AA^m_\FF$ over $\FF$. 
Then $\FF[V']=\FF[Y_1,\dots,Y_m]/I(V')$ is identified with  a subring of $\FF[V]$ under $\varphi^\sharp$, so that $Y_i+I(V')$ is identified with $\bar{f}_i:=f_i+I(V)$.
The \emph{generic fiber} $V_\varphi$ of $\varphi$ is then the affine subvariety of $V_{\FF(V')}$ defined by  $f_1-\bar{f}_1,\dots,f_m-\bar{f}_m\in \FF(V')[X_1,\dots,X_n]$. Its known from transcendence theory that $V_\varphi$ is irreducible of dimension $\dim V-\dim V'$. Thus, the second claim in Theorem~\ref{thm_dimfiber} states that the dimension of the irreducible components of a general fiber equals that of the generic fiber.

\paragraph{Finite morphisms.}

We say a morphism $\varphi: V\to V'$ between affine varieties $V$ and $V'$ over $\FF$ is a \emph{finite} morphism if $\FF[V]$ is finitely generated as a module over its subring $\varphi^\sharp(\FF[V'])$.
The image of a closed set under a finite morphism is closed. In particular, a finite morphism is surjective if it is dominant.
Fibers of finite morphisms are finite sets.

 \subsection{Further Results}
 
 We list some further results in algebraic geometry, which are used in later sections.
 
 \paragraph{Estimates for the number of rational points over $\FF_q$.}

We start with the following elementary upper bound on the number of rational points of an affine variety over a finite field $\FF_q$, which can be derived from  B\'ezout's inequality \cite{HS80, CM06}.

\begin{lemma}[{\cite[Proposition~2.3]{HS80}}]\label{lem:elementary-bound}
Let $V\subseteq\AA^n_{\FF}$ be an affine variety of dimension $k$ and degree $d$ over a field $\FF$. Let $S\subseteq\FF$ be a finite set.
Then $|V\cap S^n|\leq d |S|^{k}$.
In particular, if $\FF=\FF_q$, then $|V(\FF_q)|\leq d q^{k}$.
\end{lemma}

If $V$ is absolutely irreducible and $q$ is sufficiently large,  one can do better than Lemma~\ref{lem:elementary-bound} and show that $|V(\FF_q)|$ is close to $q^{\dim V}$ using the \emph{Lang--Weil bound}. We need the following effective version of this bound.

\begin{restatable}[Effective Lang--Weil bound]{theorem}{langweil}
\label{theorem:lang-weil}
Let $V\subseteq\AA^n_{\FF_q}$ be an absolutely irreducible affine variety over $\FF_q$ of dimension $k$ and degree $d$. Then
\[
|V(\FF_q)-q^k|<(d-1)(d-2)q^{k-1/2}+5d^{13/3}q^{k-1}.
\]
In particular, we have $|V(\FF_q)|\geq q^k/2$ if $q\geq 20d^5$.
\end{restatable}

Theorem~\ref{theorem:lang-weil} was proved by Cafure and Matera as \cite[Theorem~7.1]{CM06} with an extra condition that $q > 2(k + 1)d^2$. However, a more careful analysis shows that this condition can be removed. This was confirmed to us in \cite{CM-communication}.  See Appendix \ref{sec:noether-linear-map} for more details.
The fact that $q$ does not need to depend on $k$ is crucial to making our required field size independent of $\dim V$, where $V$ is an affine variety that defines an $(n,k,d)$ algebraic source over $\FF_q$.

\paragraph{Bombieri's estimate for exponential sums.} We also need Bombieri's estimate for exponential sums over rational points of curves over $\FF_q$.

\begin{theorem}[{\cite[Theorem~6]{bombieri-1966-exponential-sums}}]\label{thm:bombieri}
Let $C\subseteq\AA^n_{\FF_q}$ be an affine curve of degree $d_1$ over a finite field $\FF_q$ of characteristic $p$. Let $\sigma:\FF_p\to\CC^\times$ be the character $x\mapsto e^{2\pi i x/p}$ of $\FF_p$.
Suppose $f\in \FF_q[X_1,\dots,X_n]$ is a polynomial of degree $d_2$ such that for any $g\in \overline{\FF}_q[X_1,\dots,X_n]$ and any irreducible component $C_0$ of $C$, the function $f-(g^p-g)$ does not vanish identically on $C_0$.
Then
\[
\left|\sum_{x\in C(\FF_q)} (\sigma\circ\Tr\circ f)(x)\right|\leq (d_1^2+2d_1 d_2-3d_1)q^{1/2}+d_1^2.
\]
where $\Tr$ denotes the trace map from $\FF_q$ to $\FF_p$.
\end{theorem}

\paragraph{Noether normalization.}

The classical Noether normalization lemma proved by Noether \cite{Noe26} states that an affine variety $V$ of dimension $k$ over an infinite field $\FF$ admits a finite morphism $\varphi: V\to \AA^k_\FF$.
Moreover, $\varphi$ may be chosen to be a linear map.
We give the following quantitative version of this result, which states that the coefficients that specify the linear map can be chosen from a finite subset $S\subseteq\FF$ provided that $S$ is large enough. 

\begin{restatable}[Noether normalization]{lemma}{finitelinearmap}
\label{lem:finite-linear-map}
Let $V\subseteq \AA^n_\FF$ be an affine variety of dimension $k$ and degree $d$ over a field $\FF$.
Suppose $S$ is a finite subset of $\FF$ of size greater than $d$.
Then there exists a polynomial map $\varphi: \AA^n_{\FF}\to\AA^k_{\FF}$ defined by linear polynomials $\ell_i=\sum_{j=1}^n c_{i,j} X_i\in\FF[X_1,\dots,X_n]$ with coefficients $c_{i,1},\dots,c_{i,n}\in S$ for $i=1,\dots,k$
such that $\varphi|_V: V\to \AA^k_{\FF}$ is a finite morphism.
\end{restatable}

For convenience, we also prove the following lemma, which guarantees the existence of linear polynomials achieving simultaneous Noether normalization for two affine varieties.

\begin{restatable}{lemma}{finitelinearmapvariant}
\label{lem:finite-linear-map-variant}
Let $\KK_1$ and $\KK_2$ be extension fields of a field $\FF$. 
For $i=1,2$, let $V_i\subseteq \AA^n_{\KK_i}$ be an affine variety of dimension $k_i$ and degree $d_i$ over $\KK_i$.
Suppose $S$ is a finite subset of $\FF$ of size greater than $d_1+d_2$.
Then there exist linear polynomials $\ell_1,\dots,\ell_{\max\{k_1,k_2\}}\in\FF[X_1,\dots,X_n]$ with coefficients in $S$ such that
the morphism $V_i\to  \AA^{k_i}_{\KK_i}$ defined by $\ell_1,\dots,\ell_{k_i}$ is finite for $i=1,2$.
\end{restatable}

See Appendix~\ref{sec:noether-linear-map} for the proofs of Lemma~\ref{lem:finite-linear-map} and Lemma~\ref{lem:finite-linear-map-variant}.

\paragraph{Effective fiber dimension theorem.}

We also need an effective version of the fiber dimension theorem. To suit our needs, we first formulate the theorem in the following general form.
Recall that for $h_1,\dots,h_s\in \FF[X_1,\dots,X_n]$, we denote by $\mathcal{L}_{h_1,\dots,h_s, \FF}$ the linear span of $h_1,\dots, h_s$ and $1$ over $\FF$.

\begin{restatable}[Effective fiber dimension theorem -- general form]{theorem}{effectiveFDTgeneral}
\label{thm:effective-FDT-general}
Let $V\subseteq \AA^n$ be an irreducible affine variety of dimension $k$ over an algebraically closed field $\FF$.
Let $h_1,\dots,h_s\in \FF[X_1,\dots,X_n]$
with $\deg h_1\geq \dots \geq \deg h_s$.
Let $f_1,\dots,f_m\in \mathcal{L}_{h_1,\dots,h_s,\FF}$, which define a polynomial map $f: \AA^n\to\AA^m$. 
Let $k'=\dim \overline{f(V)}$.

Let $j_1,\dots,j_{k'}\in [m]$ such that the morphism $f': V\to \AA^{k'}$ defined by $f_{j_1},\dots, f_{j_{k'}}$ is dominant, which exist by Lemma~\ref{lem:dominant-projection}.
Let $V_{f'}\subseteq \AA^n_{\FF(Y_1,\dots,Y_{k'})}$ be the generic fiber of $f'$ (see the definition after Theorem~\ref{thm_dimfiber}). Finally, let $\ell_1,\dots,\ell_{k}\in\FF[X_1,\dots,X_n]$ be linear polynomials such that 
both the morphism $\pi: V\to \AA^k$ defined by $\ell_1,\dots,\ell_{k}$ 
and the morphism $\tau: V_{f'}\to\AA^{k-k'}_{\FF(Y_1,\dots,Y_{k'})}$ defined by $\ell_1,\dots,\ell_{k-k'}$ are finite.

Let $t\in\{0,\dots,k-k'\}$.
Then there exists a polynomial $P\in\FF[X_1,\dots,X_n]$ of degree at most $k'\cdot \deg V\cdot  \prod_{i=1}^{k'} \deg h_i$ that does not vanish identically on $V$ such that the following holds:
Let $\varphi: \AA^n\to\AA^{t+m}$ be the polynomial map defined by $\ell_{1},\dots,\ell_{t},f_1,\dots,f_m$.
Then for every $a\in V$ satisfying $P(a)\neq 0$, the fiber $ \varphi|_V^{-1}(\varphi(a))$ is equidimensional of dimension $k-k'-t$. 
\end{restatable}

As a corollary, we have the following effective fiber dimension theorem, stated in a more standard form.

\begin{corollary}[Effective fiber dimension theorem -- standard form]
\label{cor:effective-FDT}
Let $V\subseteq \AA^n$ be an irreducible affine variety over an algebraically closed field $\FF$.
Let $h_1,\dots,h_s\in \FF[X_1,\dots,X_n]$
with $\deg h_1\geq \dots \geq \deg h_s$.
Let $f_1,\dots,f_m\in \mathcal{L}_{h_1,\dots,h_s,\FF}$, which define a polynomial map $f: \AA^n\to\AA^m$.
Finally, let $W=\overline{f(V)}\subseteq\AA^m$.
Then there exists a polynomial $P\in\FF[X_1,\dots,X_n]$ of degree at most $\dim W\cdot \deg V\cdot  \prod_{i=1}^{\dim W} \deg h_i$ that does not vanish identically on $V$ such that for every $a\in V$ satisfying $P(a)\neq 0$, the fiber $f|_V^{-1}(f(a))$ is equidimensional of dimension $\dim V-\dim W$. 
\end{corollary}

\begin{proof}
Use the notations in Theorem~\ref{thm:effective-FDT-general}. Note that the linear polynomials $\ell_1,\dots,\ell_{k}$ satisfying the conditions in Theorem~\ref{thm:effective-FDT-general} exist by Noether normalization (see Lemma~
\ref{lem:finite-linear-map-variant}). Now apply Theorem~\ref{thm:effective-FDT-general} with $t=0$.
\end{proof}

Theorem~\ref{thm:effective-FDT-general} is proved in Appendix~\ref{sec:proof-subsec-1}. 

\paragraph{Degree bound for the images of affine varieties.}

Finally, we need the following degree bound for the images of affine varieties (or more precisely, their closures) under polynomial maps.

\begin{restatable}{lemma}{degpolyimage}
\label{lem:deg-poly-image}
Let $V\subseteq \AA^n_\FF$ be an affine variety over a field $\FF$.
Let $h_1,\dots,h_s\in \FF[X_1,\dots,X_n]$
with $\deg h_1\geq \dots \geq \deg h_s$.
Let $f_1,\dots,f_m\in \mathcal{L}_{h_1,\dots,h_s,\FF}$, which define a polynomial map $f: \AA^n_\FF\to\AA^m_\FF$.
Finally, let $W=\overline{f(V)}\subseteq \AA^m_\FF$.
Then 
\[
\deg W\leq \deg V\cdot \prod_{i=1}^{\dim W} \deg h_i.
\]
\end{restatable}

We prove Lemma~\ref{lem:deg-poly-image} in Appendix~\ref{sec:misc}.
It generalizes a bound in \cite[\S 8.5]{BCS97}, which states that $\deg W\leq d^{\dim W}$ if $V=\AA^n_\FF$ and $\deg f_i\leq d\in \NN^+$ for $i\in [m]$.

\section{Linear Seeded Rank Extractors for Varieties}\label{sec:seeded-rank-extractor}
 
In this section, we consider the problem of constructing \emph{seeded rank extractors for varieties} that are linear: i.e., a set of \emph{linear} maps such that for every variety $V$ most of the maps in the set preserve the dimension of $V$.
We show that these objects are simply linear \emph{seeded rank extractors for subspaces}, a well-known linear algebraic pseudorandom object for which explicit constructions were given in \cite{gabizon-raz, FS12, For14}.

The proof is based on the notion of \emph{tangent spaces} of varieties, which are linear subspaces that are local first-order approximations of varieties.
Intuitively, for an affine variety $V$, as we look at smaller and smaller neighborhoods of a \emph{nonsingular point} $a$ of $V$, the tangent space $T_a V$ would become a better and better approximation of $V$. Thus, one should expect that a linear map that preserves the dimension of $T_a V$, which is a subspace, also preserves the dimension of $V$. While it is not entirely obvious what ``smaller and smaller neighborhoods'' mean in the Zariski topology, we will see that the claim is indeed true and follows from general facts in algebraic geometry.

Fix $\FF$ to be an algebraically closed field throughout this section. We first formally define seeded rank extractors for varieties and subspaces.

\begin{definition}[Seeded rank extractors]\label{defn:seeded-rank-extractor}
Let $\varphi_1,\dots,\varphi_\ell:\AA^n\to \AA^m$ be polynomial maps, where $n\geq m$. We say $(\varphi_i)_{i\in [\ell]}$ is an \emph{$(n, m, k, \epsilon)$ seeded rank extractor for varieties} (resp.\ subspaces)  if for every affine variety (resp.\ linear subspace) $V\subseteq\AA^n$ over $\FF$ of dimension at least $k$, all but at most $\epsilon$-fraction of $\varphi_i$ satisfy $\dim \overline{\varphi_i(V)}=m$ (or equivalently, $\varphi_i|_V: V\to \AA^m$ is dominant). 
We call $\log \ell$ the  \emph{seed length} of the seeded rank extractor.

In addition, we say $(\varphi_i)_{i\in [\ell]}$ is \emph{linear} if each $\varphi_i$ is a linear map, i.e., defined by linear polynomials.
\end{definition}

The optimal choice of $k$ is $k=m$, in which case the seeded rank extractor is ``lossless.''
Explicit linear $(n,m,k,\epsilon)$ seeded rank extractors for subspaces with seed length $O(\log n+\log(1/\epsilon))$ and $k=m$ was first constructed by Gabizon and Raz \cite{gabizon-raz}. We use an improved construction given in \cite{FS12, For14}.

\begin{lemma}[{\cite{FS12, For14}}]\label{lem_condenser}
Let $n\in \NN^+$ and $m\in [n]$.
Let $\omega\in \FF^\times$ such that the multiplicative order of $\omega$ is at least $n$. Let $s_1,\dots,s_\ell$ be distinct elements in $\FF^\times$.
For $i\in [\ell]$, let $\varphi_i: \AA^n\to\AA^m$ be the linear map defined by the $m\times n$ matrix $((\omega^{j'-1} s_i)^{j-1})_{j'\in [m], j\in [n]}$.
In other words, $\varphi_i$ is given by
\[
\varphi_i:  (a_1,\dots,a_n)\mapsto \left(\sum_{j=1}^n s_i^{j-1} a_j,  \sum_{j=1}^n (\omega s_i)^{j-1} a_j, \dots, \sum_{j=1}^n (\omega^{m-1} s_i)^{j-1} a_j \right).
\]
Then $(\varphi_i)_{i\in [\ell]}$ is a linear $(n,m,m,\epsilon)$ seeded rank extractor for subspaces, where $\epsilon=m(n-m)/\ell$.
\end{lemma}

The main result of this section is the following theorem.

\begin{theorem}\label{thm:linear-seeded-equivalence}
An $(n, m, k, \epsilon)$ linear seeded rank extractor for subspaces is also an $(n, m, k, \epsilon)$ linear seeded rank extractor for varieties. 
\end{theorem}

\begin{corollary}\label{cor:linear-seeded}
The construction $(\varphi_i)_{i\in [\ell]}$ in Lemma~\ref{lem_condenser}  is a linear $(n,m,m,\epsilon)$ seeded rank extractor for varieties, where $\epsilon=m(n-m)/\ell$.
\end{corollary}

\subsection{Tangent Spaces and the Jacobian Criterion for Smoothness}\label{subsec:Jacobian}

The proof of Theorem~\ref{thm:linear-seeded-equivalence} uses the notion of \emph{tangent spaces}. 

\begin{definition}[Tangent space]
Let $V\subseteq \AA^n$ be an affine variety over $\FF$.
For a point $a\in V$, 
the \emph{tangent space} $T_a V$ of $V$ at $a$ is the linear subspace of $\AA^n$ defined by the linear equations
\[
\sum_{i=1}^n \frac{\partial f}{\partial X_i}(a)\cdot X_i=0,  \qquad f=f(X_1,\dots,X_n)\in I(V).
\]
\end{definition}

To explain the intuition, note that at a point $a=(a_1,\dots,a_n)\in V$, the inhomogeneous linear equations $\sum_{i=1}^n \frac{\partial f}{\partial X_i}(a)\cdot (X_i - a_i)=0$ with $f\in I(V)$ define an affine subspace of $\AA^n$ passing through $a$, which can be seen as a first-order approximation of $V$ locally at $a$.
The tangent space $T_a V$ is defined to be the linear subspace that is a translate of this affine subspace.

The dimension of a tangent space is bounded from below by the dimension of the variety, as stated by the following lemma.

\begin{lemma}[{\cite[\S II.1]{Sha13}}]\label{lem:tangent_lower_bound}
$\dim T_a V \geq \dim V$ for every affine variety $V$ over $\FF$ and $a\in V$.
\end{lemma}

The notion of \emph{smoothness} can be defined in terms of whether the equality in Lemma~\ref{lem:tangent_lower_bound} is attained. For simplicity, we define it only for irreducible affine varieties, which suffices for us.

\begin{definition}[Smoothness]
Let $V$ be an irreducible affine variety over $\FF$ and let $a\in V$.
If $\dim T_a V=\dim V$, 
then we say $V$ is \emph{smooth} or \emph{nonsingular} at $a$, and $a$ is a \emph{nonsingular point} of $V$.
Otherwise,   we say  $V$ is \emph{non-smooth} or \emph{singular} at $a$, and $a$ is a \emph{singular point} of $V$.

Denote by $V_{\mathrm{sing}}$ the subset of singular points of $V$, called the \emph{singular locus} of $V$.
\end{definition}

\paragraph{Jacobian criterion for smoothness.}

The singular locus $V_\mathrm{sing}$ of $V$ can be determined via the \emph{Jacobian criterion}, which we explain now. 

\begin{definition}[Jacobian matrix]
Let $f_1,\dots,f_m\in\FF[X_1,\dots,X_n]$. The associated \emph{Jacobian matrix} $J_{\mathbf{f}}$ is an $m\times n$ matrix over the ring $\FF[X_1,\dots,X_n]$, defined by 
\[
J_{\mathbf{f}}=\left(\frac{\partial f_i}{\partial X_j}\right)_{i\in [m], j\in [n]}.
\]
\end{definition}

\begin{lemma}\label{lem:tangent_to_Jacobian}
Let $V\subseteq\AA^n$ be an affine variety over $\FF$, and
let $f_1,\dots,f_m\in\FF[X_1,\dots,X_n]$ such that $\{f_1,\dots,f_m\}$ is a generating set of $I(V)$. 
Then for $a=(a_1,\dots,a_n)\in V$, the tangent space $T_a V$ is the right nullspace of $J_{\mathbf{f}}(a)$.
\end{lemma}

\begin{proof}
By definition, we want to show the following statement:   $\sum_{j=1}^n \frac{\partial f_i}{\partial X_j}(a)\cdot a_j=0$  for all $i\in [m]$
iff $\sum_{j=1}^n \frac{\partial f}{\partial X_j}(a)\cdot a_j=0$ for all $f\in I(V)$. The ``if'' part is immediate as $f_1,\dots,f_m\in I(V)$.

To see the ``only if'' part, consider $f\in I(V)$. Then $f=\sum_{i=1}^m g_i f_i$ for some $g_1,\dots,g_m\in\FF[X_1,\dots,X_n]$.
So for $j\in [n]$,
\[
\frac{\partial f}{\partial X_j}(a)=\sum_{i=1}^m\left(\frac{\partial g_i}{\partial X_j}(a)f_i(a)+g_i(a)\frac{\partial f_i}{\partial X_j}(a)\right)
=\sum_{i=1}^m g_i(a)\frac{\partial f_i}{\partial X_j}(a)
\]
where the second equality holds as $f_i$ vanishes at $a\in V$ for $i\in [m]$.
It follows that if $\sum_{j=1}^n \frac{\partial f_i}{\partial X_j}(a)\cdot a_j=0$ for $i\in [m]$, then $\sum_{j=1}^n \frac{\partial f}{\partial X_j}(a)\cdot a_j=0$.
\end{proof}

\begin{corollary}[Jacobian criterion for smoothness]
Let $V\subseteq\AA^n$ be an irreducible affine variety of dimension $k$  over $\FF$, and
let $f_1,\dots,f_m\in\FF[X_1,\dots,X_n]$ such that $\{f_1,\dots,f_m\}$ is a generating set of $I(V)$. 
Then  the singular locus of $V$ is given by 
\[
V_\mathrm{sing}=\{a\in V: \rank J_{\mathbf{f}}(a)<n-k\},
\]
which is an affine subvariety of $V$ defined by the set of all $(n-k)\times (n-k)$ minors of $J_{\mathbf{f}}$.
\end{corollary}

\begin{proof}
By Lemma~\ref{lem:tangent_to_Jacobian}, the condition that $\rank J_{\mathbf{f}}(a)<n-k$ is equivalent to $\dim T_a V>k$. This is further equivalent to $\dim T_a V\neq k$ by Lemma~\ref{lem:tangent_lower_bound}. So the claim holds by definition.
\end{proof}

\begin{remark*}
The above Jacobian criterion is related to but different from the \emph{Jacobian criterion for algebraic independence}. The latter states that the transcendence degree of $\FF(f_1,\dots,f_m)/\FF$ equals the rank of the associated Jacobian matrix $J_{\mathbf{f}}$, and in particular, $f_1,\dots,f_m$ are algebraically independent iff $\rank J_{\mathbf{f}}=m$. However, this statement requires the characteristic of $\FF$ to be zero or large, or more generally, a certain separability condition to hold. On the other hand, the Jacobian criterion for smoothness that we use holds without extra conditions. In particular, it works in \emph{any} characteristic.
\end{remark*}

We also need the fact that varieties are almost-everywhere-nonsingular.

\begin{lemma}[{\cite[\S II.1]{Sha13}}]\label{lem_nonsing}
The set of nonsingular points of an irreducible affine variety $V$ over $\FF$ is a dense open subset of $V$. That is, $V_{\mathrm{sing}}$ is a proper subvariety of $V$.
\end{lemma}

\subsection{Proof of Theorem~\ref{thm:linear-seeded-equivalence}}\label{subsec:proof-Jacobian}

The proof of Theorem~\ref{thm:linear-seeded-equivalence} is based on the following lemma.

\begin{lemma}\label{lem_linseeded}
Let $V\subseteq\AA^n$ be an irreducible affine variety over $\FF$ and let $a\in V$. Let $f_1,\dots,f_m\in \FF[X_1,\dots,X_n]$, which defines a polynomial map $\varphi: \AA^n\to \AA^m$. 
Let $W$ be the right nullspace of $J_{\mathbf{f}}(a)=\left(\frac{\partial f_i}{\partial X_j}(a)\right)_{i\in [m], j\in [n]}$.
Suppose $\dim (W\cap T_a V)=\dim V - m$. Then $\dim \overline{\varphi(V)}=m$.
\end{lemma}

\begin{proof}
Let $b=(b_1,\dots,b_m)=\varphi(a)\in \AA^m$, so that $a\in\varphi^{-1}(b)$. As $V\cap \varphi^{-1}(b)$ is a subvariety of both $V$ and $\varphi^{-1}(b)$, we have $T_a(V\cap \varphi^{-1}(b))\subseteq (T_a V)\cap (T_a \varphi^{-1}(b))$.

Note that the fiber $\varphi^{-1}(b)\subseteq \AA^n$ is defined by the polynomials $\hat{f}_1,\dots,\hat{f}_m$ where $\hat{f}_i:=f_i-b_i$ for $i\in [m]$.
So $\hat{f}_i\in I(\varphi^{-1}(b))$ for $i\in [m]$. 
Pick a finite generating set $\{g_1,\dots, g_s\}$ of $I(\varphi^{-1}(b))$ such that $s\geq m$ and  $g_i=\hat{f}_i$ for $i\in [m]$.  
By Lemma~\ref{lem:tangent_to_Jacobian},  the tangent space $T_a \varphi^{-1}(b)$ is the right nullspace of $J_{\mathbf{g}}(a)=\left(\frac{\partial g_i}{\partial X_j}(a)\right)_{i\in [s], j\in [n]}$.
As $\frac{\partial \hat{f}_i}{\partial X_j}=\frac{\partial f_i}{\partial X_j}$ for $i\in [m]$ and $j\in [n]$, we have  $J_{\mathbf{f}}(a)=\left(\frac{\partial \hat{f_i}}{\partial X_j}(a)\right)_{i\in [m], j\in [n]}$. As $g_i=\hat{f}_i$ for $i\in [m]$,  the matrix $J_{\mathbf{f}}(a)$ is the upper $m\times n$ submatrix of  $J_{\mathbf{g}}(a)$.
It follows that  $T_a \varphi^{-1}(b)\subseteq W$.
So $T_a(V\cap \varphi^{-1}(b))\subseteq W\cap T_a V$.
Therefore, 
\[
\dim (V\cap \varphi^{-1}(b))\leq \dim T_a(V\cap \varphi^{-1}(b))\leq \dim (W\cap T_a V)=\dim V - m
\]
where the first inequality holds by Lemma~\ref{lem:tangent_lower_bound}.
On the other hand, by the fiber dimension theorem (Theorem~\ref{thm_dimfiber}),
\[
\dim (V\cap \varphi^{-1}(b))\geq \dim V-\dim  \overline{\varphi(V)} \geq \dim V-m.
\] 
This forces $\dim (V\cap \varphi^{-1}(b))=\dim V-m$ and $\dim \overline{\varphi(V)}=m$.
\end{proof}

We are now ready to prove Theorem~\ref{thm:linear-seeded-equivalence}.

 \begin{proof}[Proof of Theorem~\ref{thm:linear-seeded-equivalence}] 
Let $V\subseteq\AA^n$ be an affine variety over $\FF$ of dimension at least $k$, and let $V_0$ be an irreducible component of $V$ such that $\dim V_0=\dim V$.
Let $a$ be a nonsingular point of $V_0$, which exists by Lemma~\ref{lem_nonsing}.
Then $\dim T_a V_0=\dim V_0\geq k$. 

We claim that for a linear map $\varphi: \AA^n\to\AA^m$,  if $\dim \varphi(T_a V_0)=m$, then $\dim \overline{\varphi(V_0)}=m$ and hence $\dim \overline{\varphi(V)}=m$. Note that this claim implies Theorem~\ref{thm:linear-seeded-equivalence}. This follows by choosing $\varphi$ to be each of the linear maps in an $(n, m, k, \epsilon)$ linear seeded rank extractor and noting that $T_a V_0$ is a linear subspace of $\AA^n$ of dimension at least $k$.

So it remains to prove the above claim.
Assume $\dim \varphi(T_a V_0)=m$.
Suppose $\varphi$ is defined by  linear polynomials $f_1,\dots, f_m\in \FF[X_1,\dots,X_n]$.
Let $W$ be the kernel of $\varphi$. Then 
\[
\dim (W\cap T_a V_0)=\dim T_a V_0 - \dim \varphi(T_a V_0) = \dim V_0 -m.
\] 
As $\varphi$ is linear, $W$ equals the right nullspace of the  matrix $J_{\mathbf{f}}(a)=\left(\frac{\partial f_i}{\partial X_j}(a)\right)_{i\in [m], j\in [n]}$. So  $\dim \overline{\varphi(V_0)}=m$ by Lemma~\ref{lem_linseeded}. This proves the claim and Theorem~\ref{thm:linear-seeded-equivalence} follows.
\end{proof}


\section{Deterministic Rank Extractors for Varieties}\label{sec:deterministic-rank-extractor}

Let $\FF$ be an algebraically closed field.
In this section, we consider the problem of constructing explicit \emph{deterministic (lossless) rank extractors/condensers for varieties}. These are polynomial maps $\AA^n\to\AA^m$ that preserve the dimension of low-degree affine varieties $V\subseteq\AA^n$ over $\FF$ but reduce the dimension of the ambient space.

Dvir, Gabizon and Wigderson \cite{dvir-gabizon-wigderson} constructed explicit deterministic rank extractors for \emph{polynomial sources}. These objects can also be viewed as deterministic rank extractors for varieties that are the closures of the images of polynomial maps.
A key technique used in their analysis is the  \emph{Jacobian criterion for algebraic independence}, which requires the characteristic of $\FF$ to be zero or large. 

To solve the problem for general varieties, one natural approach  is generalizing the Jacobian criterion for algebraic independence.
A key step in the proof of \cite{dvir-gabizon-wigderson} is showing that a certain polynomial associated with the Jacobian matrix is nonzero.
Thus, it is natural for us to show that a similar polynomial does not vanish completely on affine varieties and that this is sufficient for constructing deterministic rank extractors for varieties.

While this idea can be made rigorous, the problem is that proving the nonvanishing of a polynomial on an affine variety appears to be challenging. We need to show that not only is the polynomial nonzero, but it remains nonzero modulo the ideal defining the variety. It is not clear to us how to prove such a result due to the generality of the variety.

\paragraph{The DKL construction.}
Instead of using a Jacobian-based construction, we take a different approach. Namely, we show that the explicit construction of \emph{variety evasive sets} by Dvir, Koll\'ar, and Lovett \cite{dvir-kollar-lovett-2014} can be used to construct deterministic rank extractors for varieties. 
Variety evasive sets are large finite subsets of $\AA^n$ that have small intersections with varieties of low degree and low dimension.
While they do not give deterministic rank extractors for varieties in general, we show that the construction of variety evasive sets in \cite{dvir-kollar-lovett-2014} does give such a construction.

More specifically, Dvir, Kollar and Lovett \cite{dvir-kollar-lovett-2014} construct explicit variety evasive sets by constructing an explicit polynomial map $\varphi:\AA^n\to\AA^m$ defined by polynomials $f_1,\dots,f_m\in \FF[X_1,\dots, X_n]$ such that the intersection of $\varphi^{-1}(\mathbf{0})=V(f_1,\dots,f_m)$ with any low-degree variety of dimension at most $m$ is finite, where $\mathbf{0}$ denotes the origin of $\AA^n$. We observe that this remains true if  $\varphi^{-1}(\mathbf{0})$ is replaced by  $\varphi^{-1}(b)$ for any $b\in \AA^m$. In other words, for any low-degree variety $V$ of dimension at most $m$, the polynomial map $\varphi$ restricts to a morphism $\varphi|_V: V\to \AA^m$ whose fibers are all finite sets. In the terminology of algebraic geometry, this means $\varphi|_V$ is a \emph{quasi-finite morphism}. By the fiber dimension theorem (Theorem~\ref{thm_dimfiber}), we then have $\dim \overline{\varphi(V)}=\dim(V)$.

In this section, we construct explicit deterministic rank extractors and rank condensers for varieties by adapting the analysis in \cite{dvir-kollar-lovett-2014}. We also formulate the construction in a way that highlights the connection with linear error-correcting codes. In particular, a linear MDS code yields a deterministic rank extractor for varieties in the sense that the coefficients of the polynomials that define the rank extractor are specified by a parity-check matrix of the code. 

In Section~\ref{sec_rank_variety} and Appendix~\ref{sec:NNL}, we will show that the polynomial map $\varphi$ has the stronger property that $\varphi|_V$ is a \emph{finite morphism}, not just quasi-finite, and this gives explicit Noether normalization lemmas for affine varieties and affine algebras.

\subsection{The Explicit Construction}

We first define deterministic rank extractors and rank condensers for varieties.

\begin{definition}[Deterministic rank extractors/condensers for varieties]\label{defi:rank-extractor}
Let $n\in\NN^+$ and $m\in [n]$.
A polynomial map $\varphi:\AA^n\to \AA^m$ is an \emph{$(n,m,k,d)$ deterministic (lossless) rank condenser} if $\dim\barr{\varphi(V)} = \dim V $ for every affine variety $V \subseteq \AA^n$ over $\FF$ of dimension at most $k$ and degree at most $d$. 
When $k=m$, we also say $\varphi$ is an \emph{$(n,m,d)$ deterministic (lossless) rank extractor}.
\end{definition}

\paragraph{$k$-regular matrices.}
Let $n\in\NN^+$ and $m,k\in [n]$.
We say a matrix $M\in\FF^{m\times n}$ is \emph{$k$-regular} if any $k$ distinct columns of $M$ are linearly independent. 
(The same definition was given in \cite{dvir-kollar-lovett-2014} but for only for the special case where $k=m$.)

The following lemma gives a coding-theoretic characterization of $k$-regularity. Its proof is straightforward.

\begin{lemma} \label{lem:coding-characterization}
Let $\KK$ be a subfield of $\FF$ and let $M\in \KK^{m\times n}\subseteq\FF^{m\times n}$, where $n\in\NN^+$ and $m,k\in [n]$. The following statements hold.
\begin{itemize}
\item  $M$ is $k$-regular iff there does not exist a nonzero vector $u\in \KK^n$ of Hamming weight at most $k$ such that $Mu=0$.
\item Suppose $k=m$.  Then $M$ is $k$-regular iff it is an \emph{MDS matrix}, i.e., every maximal minor of $M$ is nonzero.
\end{itemize}
\end{lemma}

In particular, assuming $\KK$ is a finite field, the matrix $M$ is $k$-regular iff  the linear code $C=\{u\in\KK^n: Mu=0\}$ over $\KK$ defined by the parity check matrix $M$ has minimum distance at least $k+1$. And if $k=m$, then $M$ is $k$-regular iff $C$ is a linear MDS code of minimum distance $k+1$, i.e., it is a linear code of dimension $n-k$ and minimum distance $k+1$.\footnote{We define the minimum distance of the zero code $\{0\}$ to be $n+1$, so that the statement also holds for $k=n$.}

\paragraph{The construction.}

We now present the explicit construction of deterministic rank extractors and condensers for varieties. It is based on the explicit construction of variety evasive sets in \cite{dvir-kollar-lovett-2014}.

Let $n, d\in \NN^+$ and $m,k\in [n]$. Let $ d_1,\dots, d_n$ be $n$ pairwise coprime integers greater than $d$.\footnote{While \cite{dvir-kollar-lovett-2014} assumes $d_1>\dots>d_n$, this assumption does not really matter.}
Let $M=(c_{i,j})_{i\in [m], j\in [n]}\in\FF^{m\times n}$ be a $k$-regular matrix.
Let $\varphi=\varphi(M): \AA^n\to \AA^m$ be the polynomial map
\[
\varphi: (a_1,\dots,a_n)\mapsto \left(\sum_{j=1}^n c_{1,j} a_j^{d_j}, \dots, \sum_{j=1}^n c_{m,j} a_j^{d_j}\right).
\]

We remark that, curiously, the construction above is very similar to the construction of an affine extractor in Section \ref{sec:affine-extractors}, although their purposes and the techniques used to analyze them are substantially different.

The following theorem and its corollaries are the main results of this section.

\begin{theorem}\label{thm:finite-fiber}
For every $b\in \AA^m$ and every affine variety $V\subseteq\AA^n$ over $\FF$ of dimension at most $k$ and degree at most $d$, the fiber $(\varphi|_V)^{-1}(b)=\varphi^{-1}(b)\cap V$ is a finite set.  
\end{theorem}

\begin{corollary}
$\varphi$ is an $(n, m, k,d)$ deterministic rank condenser for varieties.
In particular, if $m=k$, then $\varphi$ is an $(n, m, d)$ deterministic rank extractor for varieties.
\end{corollary}

\begin{proof}[Proof (assuming Theorem \ref{thm:finite-fiber})]
Let $V\subseteq \AA^n$ be an affine variety over $\FF$ of dimension at most $k$ and degree at most $d$. Let $V_0$ be an irreducible component of $V$.
Then $\dim V_0\leq k$ and $\deg V_0\leq d$.
It suffices to show $\dim \barr{\varphi(V_0)}= \dim V_0$.
This follows from Theorem~\ref{thm:finite-fiber} and the fiber dimension theorem (Theorem~\ref{thm_dimfiber}).
\end{proof}

\paragraph{Choosing $d_1,\dots,d_n$.}

We still need to argue that $d_1,\dots,d_n$ and $M$ can be computed efficiently.
One can choose $d_1,\dots, d_n$ to be $n$ distinct primes greater than $d$.
The resulting deterministic time complexity of computing these integers is $\poly(n,d)$. The polynomial dependence on $d$ is due to the fact that there is no known deterministic $N^{o(1)}$-time algorithm for finding primes greater than an integer $N>0$.

To improve the time complexity, we may compute $d_i$ in the following alternative way. Compute the smallest $n$ distinct primes $p_1,\dots,p_n$, which have order $O(n\log n)$.
For $i\in [n]$, let $d_i$ be the smallest power of $p_i$ such that $d_i>d$, so that $d_i=O(p_i d)=O(nd\log n)$. Then $d_1,\dots,d_n$ can be computed 
in time $\poly(n, \log d)$.

\paragraph{Choosing the matrix $M$.}

We need to choose a $k$-regular matrix $M$.
For the problem of constructing an $(n,m,d)$ deterministic rank extractor for varieties (i.e., $k=m$), we need to choose $M\in \FF^{n\times m}$ to be an MDS matrix by Lemma~\ref{lem:coding-characterization}.
This can be achieved by choosing $M$ to be an Vandermonde matrix $(\omega_j^{i-1})_{i\in [m], j\in [n]}$ with distinct $\omega_1,\dots,\omega_n\in \FF$.

Suppose $\FF$ has a finite subfield $\FF_q$. Then using a Vandermonde matrix, we need $q\geq n$ to have $M\in \FF_q^{m\times n}$. 
The condition $q\geq n$  can be relaxed to $q\geq n-1$ in general, and to $q\geq n-2$ in some special cases as explicit MDS matrices $M\in\FF_q^{m\times n}$ are known in these cases \cite{MS77}.

In the case where $k=m\in\{1,n-1,n\}$, we do not need any lower bound on $q$ as all-one vectors and identity matrices are always MDS matrices, and so is the $(n-1)\times n$ matrix $M=(c_{i,j})_{i\in [n-1], j\in [n]}$ defined by
\[
c_{i,j}=\begin{cases}
1 & i=j,\\
-1 & j=n,\\
0 & \text{otherwise}.
\end{cases}
\]

So we have the following corollary.

\begin{corollary}\label{cor:MDS-extractor-for-curves}
For $m\in\{1,n-1,n\}$, there exists an explicit construction of an $(n,m,d)$ deterministic rank extractor for varieties that is defined by polynomials $f_1,\dots,f_m\in\FF[X_1,\dots, X_n]$ satisfying the following:
\begin{itemize}
\item All the coefficients of $f_1,\dots,f_m$ are in $\{0,1,-1\}$, and hence are in every subfield of $\FF$.
\item $\deg f_1,\dots, \deg f_m=O((n+d)\log (n+d))$. And the sparse representations of $f_1,\dots,f_m$ can be computed in time $\poly(n, d)$. The time complexity can be improved to $\poly(n, \log d)$ at the cost of increasing the degrees of $f_1,\dots,f_m$ to $O(nd\log n)$.
\end{itemize}
\end{corollary}

A similar statement holds for general $m\in [n]$ and the coefficients of $f_1,\dots,f_m$ can be chosen in a finite field $\FF_q$, assuming $\FF_q$ is a subfield of $\FF$ and $q\geq n-1$. The time complexity would also depend polynomially on $\log q$.

The above explicit $(n,m,d)$ deterministic extractor for varieties will be used in the proof of Theorem~\ref{thm:intro:extractor-main}, but only in the case where $m=1$. 
Previously, Dvir \cite[Theorem~3.1]{dvir-varieties} gave an explicit construction of an $(n, 1,d)$ deterministic rank extractor  for varieties, where the polynomial defining the rank extractor is recursively constructed and has degree $\poly(d^n)$. Corollary~\ref{cor:MDS-extractor-for-curves} improves the degree of the polynomial to $\widetilde{O}(n+d)$ or $\widetilde{O}(nd)$.

\subsection{Proof of Theorem~\ref{thm:finite-fiber}}

Theorem~\ref{thm:finite-fiber} can be proved via a simple adaptation of the proof in \cite{dvir-kollar-lovett-2014}.
For the sake of completeness, we present the proof below. 

First, we need the following two lemmas from \cite{dvir-kollar-lovett-2014}.

\begin{lemma}[\cite{dvir-kollar-lovett-2014}, Lemma 3.1]\label{lem:existence-of-Laurent-series}
Let $V\subseteq\AA^n$ be an affine variety over $\FF$ of dimension at least one.
Then there exist Laurent power series $h_1(T),\dots,h_n(T)\in\FF((T))$
such that 
\begin{enumerate}
    \item at least one $h_i(T)$ has a pole, i.e., $h_i(T)\not\in\FF[[T]]$, and
    \item $P(h_1(T),\dots,h_n(T))=0$ for all $P\in I(V)$.
\end{enumerate}
\end{lemma}

\begin{lemma}[\cite{dvir-kollar-lovett-2014}, Lemma 3.5]\label{lem:polynomial-from-projection}
Let $V\subseteq\AA^n$ be an affine variety over $\FF$ of dimension $k<n$ and degree $d$.
For every $J=\{i_1,\dots,i_{k+1}\}\subseteq [n]$ of size $k+1$, there exists a nonzero polynomial $g\in I(V)\cap \FF[X_{i_1},\dots,X_{i_{k+1}}]$ of degree at most $d$.
\end{lemma}

We adapt the proof of \cite{dvir-kollar-lovett-2014} to prove the following lemma.

\begin{lemma}\label{lem:three-conditions}
Let $M=(c_{i,j})_{i\in [m], j\in [n]}\in\FF^{m\times n}$ be a $k$-regular matrix, and let $ d_1, \dots, d_n$ be pairwise coprime integers greater than $d$.
Let $V\subseteq\AA^n$ be an affine variety over $\FF$ of dimension at most $k$ and degree at most $d$.
Then, there do not exist $h_1(T), \dots , h_n(T) \in \FF((T))$ that simultaneously satisfy the following conditions:
\begin{enumerate}
\item At least one $h_i(T)$ has a pole, i.e., $h_i(T)\not\in\FF[[T]]$.
\item $P(h_1(T),\dots,h_n(T))=0$ for all $P\in I(V)$. 
\item $\sum_{j=1}^n c_{i,j} h_j(T)^{d_j}\in \FF[[T]]$ for all $i\in [m]$.
 \end{enumerate}
\end{lemma}

\begin{proof}
By replacing $k$ with $k'=\dim V$, we may assume the dimension of $V$ is exactly $k$.
Assume to the contrary that there exist $h_1(T), \dots, h_n(T) \in \FF((T))$ satisfying the three conditions.
Let $R$ be the greatest integer such that the term $T^{-R}$ appears in the Laurent series $h_j(T)^{d_j}$ for some $j\in [n]$. By the first condition, at least one of the $h_j(T)$ has a pole, so we know $R > 0$. 

Let $J$ be the set of $j\in [n]$ for which the term $T^{-R}$ appears in the Laurent series $h_j(T)^{d_j}$. Then $J\neq \emptyset$.
We claim $|J|\geq k+1$.
To see this, 
let $u_j\in\FF$ be the coefficient of the term $T^{-R}$ in $h_j(T)^{d_j}$ for $j\in [n]$. 
Let $u=(u_1,\dots,u_n)\in\FF^n$.
Then the support of $u$ is precisely $J$.
By the third condition, for $i\in [m]$,
\begin{equation}\label{eq:condition-3}
\sum_{j=1}^n c_{i,j} h_j(T)^{d_j}\in \FF[[T]].
\end{equation}
The coefficient of $T^{-R}$ in the LHS of \eqref{eq:condition-3} is $\sum_{j=1}^n c_{i,j} u_j$, which equals zero by \eqref{eq:condition-3} and the fact that $R>0$.
So we get the equation
\[
M u=0.
\]
By the $k$-regularity of $M$ and Lemma~\ref{lem:coding-characterization}, the Hamming weight of $u$ is at least $k+1$, i.e., $|J|\geq k+1$. This proves the claim. Also note that this implies $k<n$ as $|J|\leq n$.

From here, the rest of the proof is identical to that of \cite[Theorem 2.1]{dvir-kollar-lovett-2014}. 
We present the proof for the sake of completeness.
For $j \in J$,  let $r_j$ be the maximal integer such that $T^{-r_j}$ appears in $h_j$, and it follows that $r_j = R/d_j$. 

Let $\{j_1, \dots, j_{k+1}\}$ be a subset of $J$ of size $k+1$. By Lemma~\ref{lem:polynomial-from-projection},
there exists a nonzero polynomial $g(X_{j_1},\dots,X_{j_{k+1}})\in I(V)\cap \FF[X_{j_1},\dots,X_{j_{k+1}}]$ of degree at most $d$.
By the second condition, we have
\begin{equation}\label{eq:vanishing-on-Laurent-series}
g(h_{j_1}(T),\dots, h_{j_{k+1}}(T))=0.
\end{equation}
Next, we observe that for every monomial $Q=X_{j_1}^{\gamma_1}\cdots X_{j_{k+1}}^{\gamma_{k+1}}$, the term $T^{-\sum_{i=1}^{k+1} \gamma_i r_{j_i}}$ appears in $Q(h_{j_1}(T),\dots, h_{j_{k+1}}(T))$.
Choose a monomial $X_{j_1}^{\alpha_1}\cdots X_{j_{k+1}}^{\alpha_{k+1}}$ that appears in $g$ such that $\sum_{i=1}^{k+1} \alpha_i r_{j_i}$ is maximized. Such a monomial exists as $g\neq 0$. Then $g$ must contain a different monomial $X_{j_1}^{\beta_1}\cdots X_{j_{k+1}}^{\beta_{k+1}}$ such that
\begin{equation}\label{eq:two-terms}
\sum_{i=1}^{k+1} \alpha_i r_{j_i}=\sum_{i=1}^{k+1} \beta_i r_{j_i}.
\end{equation}
Otherwise, the term $T^{-\sum_{i=1}^{k+1} \alpha_i r_{j_i}}$ would appear in $g(h_{j_1}(T),\dots, h_{j_{k+1}}(T))$, which contradicts \eqref{eq:vanishing-on-Laurent-series}.

Let $D = \prod\limits_{i=1}^{k+1} d_{j_i}$.
Plugging $r_{j_i}=R/d_{j_i}$ into \eqref{eq:two-terms} and then multiplying both sides of \eqref{eq:two-terms} by $D/R$, we get
\begin{equation}\label{eq:two-terms-2}
\sum_{i=1}^{k+1} \alpha_i D/d_{j_i}=\sum_{i=1}^{k+1} \beta_i D/d_{j_i}.
\end{equation}
Consider arbitrary $i\in [k+1]$.
Taking \eqref{eq:two-terms-2} modulo $d_{j_{i}}$, we get
\[
\alpha_{i} D/d_{j_{i}}\equiv\beta_{i} D/d_{j_{i}} \pmod{d_{j_{i}}}.
\]
As $D/d_{j_{i}}$ is coprime to $d_{j_{i}}$, we may cancel it from both sides, which gives
\[
\alpha_{i} \equiv\beta_{i} \pmod{d_{j_{i}}}.
\]
As $0\leq \alpha_{i}, \beta_{i}\leq \deg(g)\leq d< d_{j_{i}}$, we have $\alpha_{i}=\beta_{i}$. As $i\in [k+1]$ is arbitrary, we have $(\alpha_1,\dots,\alpha_{k+1})=(\beta_1,\dots,\beta_{k+1})$, contradicting the assumption that $X_{j_1}^{\alpha_1}\cdots X_{j_{k+1}}^{\alpha_{k+1}}\neq X_{j_1}^{\beta_1}\cdots X_{j_{k+1}}^{\beta_{k+1}}$.
\end{proof}

\begin{remark*} 
Lemma~\ref{lem:three-conditions} was implicitly proved in \cite{dvir-kollar-lovett-2014} except that the third condition was replaced by the stronger statement $\sum_{j=1}^n c_{i,j} h_j(T)^{d_j}=0$.
Our observation is that the proof still works if this condition is relaxed to $\sum_{j=1}^n c_{i,j} h_j(T)^{d_j}\in \FF$, or even $\sum_{j=1}^n c_{i,j} h_j(T)^{d_j}\in \FF[[T]]$. (As can be seen below, the former relaxation suffices for proving Theorem~\ref{thm:finite-fiber}, but we will need the latter in Section~\ref{sec_rank_variety} and Appendix~\ref{sec:NNL} when we prove that $\varphi|_V$ is a finite morphism.)
\end{remark*}

Now we are ready to prove Theorem~\ref{thm:finite-fiber}.

\begin{proof}[Proof of Theorem~\ref{thm:finite-fiber}]
Assume to the contrary that Theorem~\ref{thm:finite-fiber} does not hold. Then there exist an affine variety $V\subseteq\AA^n$ over $\FF$ of dimension at most $k$ and degree at most $d$ and $b=(b_1,\dots,b_m)\in\AA^m$ such that $\varphi^{-1}(b)\cap V$ is not finite, i.e., $\dim (\varphi^{-1}(b)\cap V)\geq 1$.
Applying Lemma~\ref{lem:existence-of-Laurent-series} to $\varphi^{-1}(b)\cap V$, we see that there exist Laurent power series $h_1(T),\dots,h_n(T)\in\FF((T))$ such that
\begin{enumerate}
    \item at least one $h_i(T)$ has a pole, and
    \item $P(h_1(T),\dots,h_n(T))=0$ for all $P\in I(\varphi^{-1}(b)\cap V)$.
\end{enumerate}
As $I(V)\subseteq I(\varphi^{-1}(b)\cap V)$, the second item implies $P(h_1(T),\dots,h_n(T))=0$ for all $P\in I(V)$.
In addition, for $i\in [m]$, we have $\left(\sum_{j=1}^n c_{i,j} X_j^{d_j}\right)-b_i\in I(\varphi^{-1}(b))\subseteq I(\varphi^{-1}(b)\cap V)$ by the definition of $\varphi$. So the second item also implies
\[
\sum_{j=1}^n c_{i,j} h_j(T)^{d_j}=b_i\in \FF\subseteq \FF[[T]] \qquad \text{for}~i\in [m]. 
\]
But then $h_1(T),\dots,h_n(T)$ satisfy the three conditions in Lemma~\ref{lem:three-conditions}, and $M=(c_{i,j})_{i\in [m], j\in [n]}$ is $k$-regular. This contradicts Lemma~\ref{lem:three-conditions}.
\end{proof}

\section{Decomposition and Min-Entropy Estimation of $(n,k,d)$ Algebraic Sources}\label{sec:decomposition}

In this section, we prove that every $(n,k,d)$ algebraic source can be (approximately) decomposed into a convex combination of irreducible, or even irreducibly minimal $(n,k,d)$ sources. 
In particular, this reduces the problem of constructing deterministic extractors for general $(n,k,d)$ algebraic sources to that for irreducibly minimal $(n,k,d)$ algebraic sources. We will use this reduction in Section~\ref{sec:extracting-seed}.

In addition, we show that every $(n,k,d)$ algebraic source $D$ over $\FF_q$ is close to a distribution with min-entropy about $k\log q$, and that this estimation is tight up to an additive term of order $O(\log d)$ assuming that $k$ is maximized, i.e., that $D$ is not an $(n,k+1,d)$ algebraic source over $\FF_q$.

\subsection{Decomposition of $(n,k,d)$ Algebraic Sources}

First, we prove some useful lemmas.

\begin{lemma}\label{lem:bound-intersection}
Let $V\subseteq\AA^n_\FF$ be an affine variety of dimension $k$ over a field $\FF$. Let $V_1,\dots, V_s$ be the irreducible components of $V$.
Suppose $S$ is a finite subset of $\FF$.
For $i\in [s]$, let $B_i$ be the subset of $V_i\cap S^n$ consisting of the points that are in the intersection of at least two irreducible components of $V$. Then $\sum_{i=1}^s |B_i|\leq (\deg V)^2 |S|^{k-1}$.
\end{lemma}
\begin{proof}
For $i\in [s]$, we have $|B_i|\leq \deg V_i\cdot \deg V\cdot q^{k-1}$ by B\'ezout's inequality (Lemma~\ref{lem:bezout}) and Lemma~\ref{lem:elementary-bound}.
So $\sum_{i=1}^s |B_i|\leq \sum_{i=1}^s\deg V_i\cdot \deg V\cdot q^{k-1}=(\deg V)^2\cdot q^{k-1}$.
\end{proof}

\begin{lemma}\label{lem:not-absolute-irreducible}
Let $V\subseteq\AA^n_{\FF_q}$ be an affine variety of dimension $k$ over $\FF_q$ such that no irreducible component of $V$ is absolutely irreducible. Then $|V(\FF_q)|\leq (\deg V)^2 q^{k-1}$.
\end{lemma}

\begin{proof}
Naturally identify $V(\FF_q)$ with $V_{\overline{\FF}_q}\cap \FF_q^n$.
The Frobenius automorphism over $\FF_q$ sends every irreducible component of $V_{\overline{\FF}_q}$ to a different irreducible component as these irreducible components are not absolutely irreducible. But it fixes every point in $V(\FF_q)$ since these points are rational. So every point in $V(\FF_q)$ is in the intersection of at least two irreducible components of $V_{\overline{\FF}_q}$.
Applying Lemma~\ref{lem:bound-intersection} with $\FF=\overline{\FF}_q$ and $S=\FF_q$, we see $|V(\FF_q)|\leq (\deg V_{\overline{\FF}_q})^2 q^{k-1}=(\deg V)^2 q^{k-1}$.
\end{proof}

\begin{lemma}\label{lem:dist-distance}
Suppose $S$ is a finite set and $B$ is a proper subset of $S$.
Let $D$ and $D'$ be the uniform distributions over $S$ and $S\setminus B$ respectively. Then $D$ and $D'$ are $\frac{|B|}{|S|}$-close.
\end{lemma}

\begin{proof}
Consider the following process: First sample $x\sim D$. If $x\not\in B$, then output $x$. Otherwise sample $y\sim D'$ and output $y$. The distribution of the output is exactly $D'$. So the statistical distance between $D$ and $D'$ is at most $\Pr_{x\sim D}[x\in B]\leq \frac{|B|}{|S|}$.
\end{proof}

The next lemma states that every $(n,k,d)$ algebraic source $D$ can be approximately decomposed into a convex combination of irreducible $(n,k,d)$ algebraic sources, and such a decomposition preserves the minimality (i.e., the property that $\dim V=k$). The idea behind the proof is quite natural: we start by decomposing $V$ as a union of irreducible components. We then observe that the contribution of components that are not absolutely irreducible or have dimensions strictly lower than $\dim V$ is small (using Lemma \ref{lem:not-absolute-irreducible}).
In addition, the remaining irreducible components are approximately disjoint, namely, the intersection of any two distinct irreducible components is small (by Lemma \ref{lem:bound-intersection}).
This implies that these irreducible components approximately define a convex combination of irreducible algebraic sources that is close to $D$.

\begin{lemma}[Decomposition into irreducible algebraic sources]\label{lem:decompose-irreducible}
Suppose $q\geq \max\{20d^5, 2d^2/\epsilon\}$, where $\epsilon\in (0,1)$.
Then every $(n,k,d)$ algebraic source $D$ over $\FF_q$ is $\epsilon$-close to a convex combination of irreducible $(n,k,d)$ algebraic sources $D_i$ over $\FF_q$. Moreover, if $D$ is a minimal $(n,k,d)$ algebraic source over $\FF_q$, then each $D_i$ can be chosen to be an irreducibly minimal $(n,k,d)$ algebraic source over $\FF_q$.
\end{lemma}

\begin{proof}
Let $D$ be an $(n,k,d)$ algebraic source.
Let $V\subseteq \AA^r_{\FF_q}$ and $f:\AA^r_{\FF_q}\to\AA^n_{\FF_q}$ be as in Definition~\ref{defn:algebraic-source} so that $D=f(U_{V(\FF_q)})$.

Let $V_*$ be the union of the irreducible components of $V$ that are absolutely irreducible and have dimension $\dim V$, and let $V_*^\mathrm{c}$ be the union of the remaining irreducible components of $V$. Then $V_*\neq \emptyset$ by the first condition in Definition~\ref{defn:algebraic-source}.
By the effective Lang--Weil bound (Theorem~\ref{theorem:lang-weil}) and the assumption that $q\geq \max\{20d^5, 2d^2/\epsilon\}$, we have
\begin{equation}\label{eq:bound1}
|V(\FF_q)|\geq |V_*(\FF_q)|\geq q^{\dim V}/2\geq d^2 q^{\dim V-1}/\epsilon.
\end{equation}

Consider an irreducible component $V_0$ of $V_*^\mathrm{c}$. Either $\dim V_0<\dim V$ holds or $V_0$ is not absolutely irreducible.
In the former case, we have $|V_0(\FF_q)|\leq \deg V_0\cdot q^{\dim V-1}$ by Lemma~\ref{lem:elementary-bound}.
And in the latter case, we have $|V_0(\FF_q)|\leq (\deg V_0)^2\cdot q^{\dim V-1}$ by Lemma~\ref{lem:not-absolute-irreducible}.
Summing over all $V_0$, we conclude that
\begin{equation}\label{eq:bound2}
|V_*^\mathrm{c}(\FF_q)|\leq (\deg V_*^\mathrm{c})^2\cdot q^{\dim V-1}.
\end{equation}
Let $U_{V_*(\FF_q)}$ be the uniform distribution over $V_*(\FF_q)$.
By Lemma~\ref{lem:dist-distance}, the distributions $U_{V(\FF_q)}$ and $U_{V_*(\FF_q)}$ are $\epsilon'$-close, where 
\begin{equation}\label{eq:bound3}
\epsilon':=\frac{|V(\FF_q)\setminus V_*(\FF_q)|}{|V(\FF_q)|}
\leq \frac{|V_*^\mathrm{c}(\FF_q)|}{|V(\FF_q)|}
\stackrel{\eqref{eq:bound1},\,\eqref{eq:bound2}}{\leq} \frac{(\deg V_*^\mathrm{c})^2}{d^2}\cdot\epsilon.
\end{equation}

Let $V_1,\dots, V_s$ be the irreducible components of $V_*$.
For $i\in [s]$, let $U_i$ be the uniform distribution over $V_i(\FF_q)$. Let $N=\sum_{i=1}^s |V_i(\FF_q)|\geq |V_*(\FF_q)|$.
Define the distribution $U'$ to be the convex combination 
\[
U':=\sum_{i=1}^s \frac{|V_i(\FF_q)|}{N} U_i.
\]
Let $B$ be the set of points in $V_*(\FF_q)$ that are in the intersection of at least two irreducible components of $V_*$, and let $B_i=B\cap V_i(\FF_q)$ for $i\in [s]$.
Note that all the points in $V_*(\FF_q)\setminus B$ have the same probability $\frac{1}{N}$ in the distribution $U'$.

Let $p_B=\Pr_{x\sim U'}[x\in B]=\frac{\sum_{i=1}^s|B_i|}{N}$.
Consider the following process: Sample $x\sim U'$. If $x\not\in B$, then output $x$. Otherwise, output $x'\in V_*(\FF_q)$ such that each $y\in V_*(\FF_q)$ is output with probability $p_y$, where
\[
p_y=\begin{cases}
p_B^{-1}\Pr_{x\sim U_{V_*(\FF_q)}}[x=y] & \text{if}~y\in B, \\
p_B^{-1}\left(\Pr_{x\sim U_{V_*(\FF_q)}}[x=y]-\frac{1}{N}\right)=p_B^{-1}\left(\frac{1}{|V_*(\FF_q)|}-\frac{1}{N}\right) & \text{if}~y\not\in B.
\end{cases}
\]
It is easy to verify that the probabilities $p_y$ do define a distribution over $V_*(\FF_q)$.
Moreover, they are chosen in the way that the output distribution of the above process is precisely $U_{V_*(\FF_q)}$.
It follows that $U_{V_*(\FF_q)}$ and $U'$ are $p_B$-close.
By Lemma~\ref{lem:bound-intersection}, we have $\sum_{i=1}^s |B_i|\leq (\deg V_*)^2 q^{\dim V-1}$. So 
\begin{equation}\label{eq:bound4}
p_B=\frac{\sum_{i=1}^s|B_i|}{N}\leq \frac{(\deg V_*)^2 q^{\dim V-1}}{N}\leq \frac{(\deg V_*)^2 q^{\dim V-1}}{|V_*(\FF_q)|}\stackrel{\eqref{eq:bound1}}{\leq} \frac{(\deg V_*)^2}{d^2}\cdot\epsilon.
\end{equation}

As  $\deg V_*+\deg V_*^\mathrm{c}=\deg V\leq d$, we have 
\[
\epsilon'+p_B\stackrel{\eqref{eq:bound3},\,\eqref{eq:bound4}}{\leq}
\frac{(\deg V_*^\mathrm{c})^2}{d^2}\cdot\epsilon+\frac{(\deg V_*)^2}{d^2}\cdot\epsilon
\leq \epsilon.
\]
So $U_{V(\FF_q)}$ and $U'$ are $\epsilon$-close. It follows that $D=f(U_{V(\FF_q)})$ and $f(U')$ are $\epsilon$-close. Recall that $U'$ is a convex combination of $U_1,\dots, U_s$, where each $U_i$ is the uniform distribution over $V_i(\FF_q)$. And by definition, $f(U_i)$ is an irreducible $(n,k,d)$ algebraic sources over $\FF_q$ for $i\in [s]$.
It follows that $D$ is $\epsilon$-close to a convex combination of the irreducible $(n,k,d)$ algebraic sources $f(U_1),\dots, f(U_s)$ over $\FF_q$.

Finally, if $D$ is a minimal $(n,k,d)$ algebraic source over $\FF_q$, then by definition, the affine variety $V$ may be chosen such that $\dim V=k$. Then we also have $\dim V_i=k$ for $i\in [s]$ in the above proof. In this case, each $f(U_i)$ is an irreducibly minimal $(n,k,d)$ algebraic source over $\FF_q$ by definition.
So $D$ is $\epsilon$-close to a convex combination of irreducibly minimal $(n,k,d)$ algebraic sources over $\FF_q$.
\end{proof}

Next, we further decompose an irreducible $(n,k,d)$ algebraic source into a convex combination of irreducibly minimal $(n,k,d)$ algebraic sources. Our main tool is the effective fiber dimension theorem (Theorem~\ref{thm:effective-FDT-general}). Using this theorem and the results of Section \ref{sec:prelim-AG}, we intersect the variety $V$ with various translates of a carefully chosen linear subspace. There are some bad events that could happen for some of these intersections. For example, the intersection may have the ``wrong'' dimension,
or the resulting variety might have the ``correct" dimension $k$ but none of the irreducible components of dimension $k$ are absolutely irreducible.
Using the effective fiber dimension theorem, we are able to show that these bad events correspond to small portions of the variety $V$, and then we again obtain a natural way to decompose the remaining part as a convex combination of irreducibly minimal $(n,k,d)$ sources.

We start with the following application of the fiber dimension theorem.

\begin{lemma}\label{lem:bounding-bad-events}
Let $V\subseteq\AA^r_{\FF_q}$ be an irreducible affine variety over $\FF_q$. 
Let $\varphi=(\varphi_1,\varphi_2):V\to \AA^{n}_{\FF_q}$ be a dominant morphism defined by $f_1,\dots,f_{n}$, where $\varphi_1$ and $\varphi_2$ are defined by the first $n_1$ and the last $n_2$ polynomials respectively and $n_1+n_2=n$.
Let $U$ be the subset of $a\in V(\FF_q)$ such that $\varphi^{-1}(\varphi(a))$ is equidimensional of dimension $\dim V-n_1-n_2$.
Then we have:
\begin{enumerate}
\item For $a\in U$, the fiber $V_{\varphi_1(a)}:=\varphi_1^{-1}(\varphi_1(a))$ is equidimensional of dimension $\dim V-n_1$.
\item For $a\in U$ and each irreducible component $Z$ of $V_{\varphi_1(a)}$, we have $\dim \overline{\varphi_2(Z)}\leq n_2$. Moreover, the equality is attained if $a\in Z$.
\end{enumerate}
\end{lemma}

\begin{proof}
Let $a\in U$.
By the fiber dimension theorem (Theorem~\ref{thm_dimfiber}), every irreducible component of $V_{\varphi_1(a)}$ is at least $\dim V-\dim\overline{\varphi_1(V)}\geq \dim V-n_1$.
Also note that $\varphi^{-1}(\varphi(a))=\varphi_2|_{V_{\varphi_1(a)}}^{-1}(\varphi_2(a))$. Again by the fiber dimension theorem, we have
\[
\dim \varphi^{-1}(\varphi(a))\geq \dim V_{\varphi_1(a)} - \dim \overline{\varphi_2(V_{\varphi_1(a)})}\geq \dim V_{\varphi_1(a)}-n_2.
\]
We know $\dim \varphi^{-1}(\varphi(a))=\dim V-n_1-n_2$ by assumption.
So $\dim V_{\varphi_1(a)}\leq (\dim V-n_1-n_2)+n_2=\dim V - n_1$.
It follows that $V_{\varphi_1(a)}$ is equidimensional of dimension $\dim V-n_1$, which proves the first claim.

Let $Z$ be an irreducible component of $V_{\varphi_1(a)}$. We already know that 
\[
\dim \overline{\varphi_2(Z)}\leq \dim\overline{\varphi_2(V_{\varphi_1(a)})}\leq n_2.
\]
Now assume $a\in Z$. Let $W$ be the irreducible component of $\varphi^{-1}(\varphi(a))$ that contains $a$.
We have $\dim W=\dim V-n_1-n_2$ by assumption and $\dim Z=\dim V-n_1$ by the first claim.
Note that $W$ is an irreducible component of $\varphi_2|_{Z}^{-1}(\varphi_2(a))$.
So by the fiber dimension theorem,
\[
\dim V-n_1-n_2=\dim W\geq  \dim Z - \dim \overline{\varphi_2(Z)} = (\dim V-n_1)-\dim \overline{\varphi_2(Z)}
\]
which implies that $\dim \overline{\varphi_2(Z)}\geq n_2$. So $\dim \overline{\varphi_2(Z)}=n_2$.
\end{proof}

Consider the setup in the previous lemma and further assume that $V$ is absolutely irreducible. The following lemma roughly asserts that for most values of $b$ in the image of $\varphi_1$, the fiber $V_b$ satisfies the property that all but at most an $\epsilon$ fraction of its points come from irreducible components $Z$ such that $\dim Z=\dim V-n_1$, $Z$ is absolutely irreducible, and $\dim \overline{\varphi_2(Z)}=n_2$. In other words, the set of ``bad'' points in $V_b$ that belong to other irreducible components is negligible.

\begin{lemma}\label{lem:restricting-to-fiber}
Let $V\subseteq\AA^r_{\FF_q}$ be an absolutely irreducible affine variety over $\FF_q$. Let $\varphi=(\varphi_1,\varphi_2):V\to \AA^{n}_{\FF_q}$ be a dominant morphism defined by $f_1,\dots,f_{n}$, where $\varphi_1$ and $\varphi_2$ are defined by the first $n_1$ and the last $n_2$ polynomials respectively and $n_1+n_2=n$.
For $b\in\FF_q^{n_1}$, let $V_b=\varphi_1^{-1}(b)$, and let $V'_b$ be the union of the irreducible components $Z$ of $V_b$ such that $Z$ is absolutely irreducible of dimension $\dim V-n_1$ and $\dim \overline{\varphi_2(Z)}=n_2$.
Define
\[
\delta=\Pr_{a\sim U_{V(\FF_q)}}[\dim\varphi^{-1}(\varphi(a))\neq \dim V-n].
\]
Let $d\in\NN^+$ and $\epsilon=(2d^2/q+\delta)^{1/2}$.
Assume $q\geq 20d^5$, $\deg V\leq d$, and $\deg V_b\leq d$ for all $b\in \varphi_1(V(\FF_q))$.
Then with probability at least $1-\epsilon$ over $b\sim \varphi_1(U_{V(\FF_q)})$, it holds that 
\[
|V_b(\FF_q)\setminus V_b'(\FF_q)|\leq \epsilon\cdot |V_b(\FF_q)|.
\]
\end{lemma}

\begin{proof}
For $b\in\varphi_1(V(\FF_q))$, let $W_b$ be the union of the irreducible components of $V_b$ of dimension at most $\dim V-n_1$ that are not absolutely irreducible. By Lemma~\ref{lem:not-absolute-irreducible}, we have $|W_b(\FF_q)|\leq d^2 q^{\dim V-n_1-1}$. Let $B_0=\bigcup_{b\in \varphi_1(V(\FF_q))} W_b(\FF_q)$. Then 
\[
|B_0|\leq |\varphi_1(V(\FF_q))|\cdot  d^2 q^{\dim V-n_1-1}\leq   d^2 q^{\dim V-1}\leq (2d^2/q)\cdot |V(\FF_q)|.
\]
where the last inequality uses Theorem~\ref{theorem:lang-weil}.

Let $U=\{a\in V(\FF_q): \dim\varphi^{-1}(\varphi(a))=\dim V-n\}$ and $B=V(\FF_q)\setminus U$.
Then 
$|B|=\delta \cdot |V(\FF_q)|$.
Also let $U'=\bigcup_{b\in \varphi_1(V(\FF_q))} V'_b(\FF_q)$ and $B'=V(\FF_q)\setminus U'$.

By Lemma~\ref{lem:bounding-bad-events}, for every $a\in U$, the irreducible component $Z$ of $V_{\varphi_1(a)}$ containing $a$ satisfies that $\dim Z=\dim V-n_1$ and $\dim \overline{\varphi_2(Z)}=n_2$.
So either $Z$ is not absolutely irreducible (which implies that $a\in Z\subseteq W_{\varphi_1(a)}\subseteq B_0$), or $a\in Z\subseteq V'_{\varphi_1(a)}\subseteq U'$.
It follows that $U\subseteq B_0\cup U'$ and hence $B'\subseteq B_0\cup B$. Therefore,
\[
|B'|\leq |B_0|+|B|\leq (2d^2/q+\delta)\cdot |V(\FF_q)|=\epsilon^2\cdot|V(\FF_q)|.
\]
So we have 
\[
\Ex_{b\sim \varphi_1(U_{V(\FF_q)})}\left[\frac{|V_b(\FF_q)\setminus V_b'(\FF_q)|}{|V_b(\FF_q)|}\right]
=\sum_{b\in \varphi_1(V(\FF_q))}\frac{|V_b(\FF_q)|}{|V(\FF_q)|}\cdot\frac{|V_b(\FF_q)\setminus V_b'(\FF_q)|}{|V_b(\FF_q)|}
=\frac{|B'|}{|V(\FF_q)|}\leq \epsilon^2.
\]
By Markov's inequality, the probability that $|V_b(\FF_q)\setminus V_b'(\FF_q)|>\epsilon \cdot |V_b(\FF_q)|$ holds over $b\sim \varphi_1(U_{V(\FF_q)})$ is at most $\epsilon$.
\end{proof}

As a consequence of Lemma~\ref{lem:restricting-to-fiber}, we also prove the following lemma, which will be used in Section~\ref{sec:full-extractor}.

\begin{lemma}\label{lem:conditional-dist}
Let $V\subseteq\AA^r_{\FF_q}$ be an absolutely irreducible affine variety over $\FF_q$. Let $\varphi=(\varphi_1,\varphi_2):V\to \AA^{n}_{\FF_q}$ be a dominant morphism defined by $f_1,\dots,f_{n}$ as in Lemma~\ref{lem:restricting-to-fiber}.
Let $d\in\NN^+$ and $\epsilon\in (0,1)$ such that $\epsilon^2\geq 2(n+1)d^2/q$.
Assume $q\geq 20d^5$.
Also assume that $f_1,\dots,f_{n}\in \mathcal{L}_{h_1,\dots,h_s,\FF_q}$ for some $h_1,\dots,h_s\in \FF_q[X_1,\dots,X_r]$ with $\deg h_1\geq \dots \geq \deg h_s$ such that 
\[
\deg V\cdot \prod_{i=1}^{n} \deg h_i \leq d.
\]
Let $D=(D_1,D_2)=\varphi(U_{V(\FF_q)})$ where $D_i=\varphi_i(U_{V(\FF_q)})$ for $i=1,2$.
Then with probability at least $1-\epsilon$ over $b\sim D_1$,  
the distribution $D_2|_{D_1=b}$ is $\epsilon$-close to an $(n_2, n_2, d)$ algebraic source over $\FF_q$.
\end{lemma}

\begin{proof}
Note that $\deg V\leq d$ as the fact that $\varphi$ is dominant implies that $\deg h_i\geq 1$ for $i\in [s]$. For $b\in\FF_q^{n_1}$, let $V_b$ and $V'_b$ be as in Lemma~\ref{lem:restricting-to-fiber}, i.e., 
$V_b=\varphi_1^{-1}(b)$ and $V'_b$ is the union of the irreducible components $Z$ of $V_b$ such that $Z$ is absolutely irreducible of dimension $\dim V-n_1$ and $\dim \overline{\varphi_2(Z)}=n_2$.

 Consider $a\in V(\FF_q)$ and let $b=\varphi_1(a)$.
 Note that $V_b=V\cap V(f_1-f_1(a),\dots,f_{n_1}-f_{n_1}(a))$.
 Recall that the polynomial $f_1,\dots,f_n\in \mathcal{L}_{h_1,\dots,h_s,\FF_q}$ are linear combinations of $h_1,\dots,h_s$ and $1$ over $\FF_q$.
By Gaussian elimination, we can find integers  $1\leq j_1< \dots<  j_t\leq n$, where $0\leq t\leq n_1$, and polynomials $g_1,\dots,g_t\in\FF_q[X_1,\dots,X_r]$ such that $V(f_1-f_1(a),\dots,f_{n_1}-f_{n_1}(a))=V(g_1,\dots,g_t)$ 
and each $g_i$ can be written as a linear combination 
\[
g_i=c_{i, j_i} h_{j_i}+ c_{i,j_i+1} h_{j_i+1}+\dots+c_{i,s} h_{s} +c_i
\]
with $c_{i,j}, c_i\in\FF_q$ and $c_{i,j_i}\neq 0$.
B\'ezout's inequality (Lemma~\ref{lem:bezout}) then gives
\begin{equation}\label{eq:degree-of-fiber}
\deg V_b \leq \deg V\cdot \prod_{i=1}^t \deg g_i = \deg V\cdot \prod_{i=1}^t \deg h_{j_i} \leq d.
\end{equation}

Let $\{\widehat{j}_1,\dots,\widehat{j}_{s-t}\}=[s]\setminus\{j_1,\dots,j_t\}$, where $\widehat{j}_1<\dots < \widehat{j}_{s-t}$.
As $g_1,\dots,g_t$ vanish identically on $V_b$, adding to each $f_i$ a multiple of $g_j$ for $j\in [t]$ does not change $f_i|_{V_b}$.
In particular, for $i\in [n]$, 
we can eliminate the dependence of $f_i$ on $h_{j_1},\dots,h_{j_t}$ and find $\widetilde{f}_i\in \mathcal{L}_{h_{\widehat{j}_1},\dots,h_{\widehat{j}_{s-t}},\FF_q}$ such that $\widetilde{f}_i|_{V_b}=f_i|_{V_b}$.
Then the morphism $\varphi_2|_{V_b}: V_b\to\AA^{n_2}_{\FF_q}$ is defined by the polynomials $\widetilde{f}_{n_1+1},\dots,\widetilde{f}_{n}$.
And
\begin{equation}\label{eq:degree-product}
\deg V'_b\cdot \prod_{i=1}^{n_2} \deg h_{\widehat{j}_i} \leq 
\deg V_b\cdot \prod_{i=1}^{n_2} \deg h_{\widehat{j}_i} 
\stackrel{\eqref{eq:degree-of-fiber}}{\leq}
\deg V\cdot \prod_{i=1}^t \deg h_{j_i} \cdot \prod_{i=1}^{n_2} \deg h_{\widehat{j}_i}\leq \deg V\cdot \prod_{i=1}^n \deg h_i\leq d.
\end{equation}

Let $\delta=\Pr_{a\sim U_{V(\FF_q)}}[\dim\varphi^{-1}(\varphi(a))\neq \dim V-n]$.
By the effective fiber dimension theorem (Corollary~\ref{cor:effective-FDT}), there exists a polynomial $P\in\overline{\FF}_q[X_1,\dots,X_r]$ of degree at most $n\cdot \deg V\cdot  \prod_{i=1}^{n} \deg h_i\leq nd$ that does not vanish identically on $V_{\overline{\FF}_q}$ such that for every $a\in V_{\overline{\FF}_q}$ satisfying $P(a)\neq 0$, the fiber $\varphi^{-1}(\varphi(a))$ is equidimensional of dimension $\dim V-n$.
 Let $B$ be the set of $a\in V(\FF_q)$ such that $\dim \varphi^{-1}(\varphi(a))\neq \dim V-n$. Then $B\subseteq V_{\overline{\FF}_q}\cap V(P)\cap \FF_q^r$.
 By B\'ezout's inequality (Lemma~\ref{lem:bezout}) and Lemma~\ref{lem:elementary-bound}, we have
 \[
|B|\leq \deg V\cdot \deg P\cdot q^{\dim V-1}\leq nd^2 q^{\dim V-1}.
 \]
Let $\delta=\Pr_{a\sim U_{V(\FF_q)}}[\dim\varphi^{-1}(\varphi(a))\neq \dim V-n]$.
Then 
\[
\delta=\frac{|B|}{|V(\FF_q)|}\leq \frac{nd^2 q^{\dim V-1}}{q^{\dim V}/2}=2nd^2/q,
\]
where we use the fact $|V(\FF_q)|\geq q^{\dim V}/2$ that follows from Theorem~\ref{theorem:lang-weil}.
So $\epsilon\geq (2(n+1)d^2/q)^{1/2}\geq (2d^2/q+\delta)^{1/2}$.

By Lemma~\ref{lem:restricting-to-fiber}, with probability at least $1-\epsilon$ over $b\sim \varphi_1(U_{V(\FF_q)})$, it holds that 
$|V_b(\FF_q)\setminus V_b'(\FF_q)|\leq \epsilon\cdot |V_b(\FF_q)|$.
Fix $b$ such that this holds. Note that $D_2|_{D_1=b}=\varphi_2(U_{V_b(\FF_q)})$.
So it suffices to verify that $\varphi_2(U_{V_b(\FF_q)})$ is $\epsilon$-close to an $(n_2,n_2,d)$ algebraic source over $\FF_q$.

The set $V_b'(\FF_q)$ is nonempty as $|V_b'(\FF_q)|\geq (1-\epsilon)|V_b(\FF_q)|>0$.
By definition, every irreducible component $Z$ of $V'_b$ is absolutely irreducible of dimension $\dim V-n_1$ and satisfies $\dim \overline{\varphi_2(Z)}=n_2$. 
So the distribution $\varphi_2(U_{V'_b(\FF_q)})$ satisfies the first two conditions of $(n_2,n_2,d)$ algebraic sources in Definition~\ref{defn:algebraic-source}.
And the third condition also holds by \eqref{eq:degree-product}.
This shows that $\varphi_2(U_{V'_b(\FF_q)})$ is an $(n_2,n_2,d)$ algebraic source over $\FF_q$.

Finally, as $|V_b(\FF_q)\setminus V_b'(\FF_q)|\leq \epsilon\cdot |V_b(\FF_q)|$, the distributions $U_{V_b(\FF_q)}$ and $U_{V_b'(\FF_q)}$ are $\epsilon$-close by Lemma~\ref{lem:dist-distance}.
It follows that  
$D_2|_{D_1=b}=\varphi_2(U_{V_b(\FF_q)})$ is $\epsilon$-close to the $(n_2,n_2,d)$ algebraic source $\varphi_2(U_{V'_b(\FF_q)})$.
\end{proof}

We now use Lemma~\ref{lem:restricting-to-fiber} and the effective fiber dimension theorem to decompose irreducible $(n,k,d)$ algebraic sources, thus completing the proof of the main result of this section.

\begin{lemma}\label{lem:decompose-minimal}
Suppose $q\geq \max\{20d^5, 2(k+1)d^2/\epsilon^2\}$, where $\epsilon\in (0,1)$.
Then every irreducible $(n,k,d)$ algebraic source over $\FF_q$ is $3\epsilon$-close to a convex combination of irreducibly minimal $(n,k,d)$ algebraic sources over $\FF_q$. 
\end{lemma}

\begin{proof}
Let $D$ be an irreducible $(n,k,d)$ algebraic source over $\FF_q$.
Let $V\subseteq \AA^r_{\FF_q}$ and $f:\AA^r_{\FF_q}\to\AA^n_{\FF_q}$ be as in Definition~\ref{defn:algebraic-source}, where $V$ is irreducible (and hence absolutely irreducible)  and $D=f(U_{V(\FF_q)})$.
Let $k_V=\dim V\geq k$.
By the effective Lang--Weil bound (Theorem~\ref{theorem:lang-weil}), we have 
\[
|V(\FF_q)|\geq q^{k_V}/2.
\]

By Lemma~\ref{lem:dominant-projection}, there exist distinct $j_1,\dots,j_k\in [n]$ such that the morphism $\varphi_2:V\to \AA^k_{\FF_q}$
defined by $f_{j_1},\dots,f_{j_k}$ satisfies $\dim \overline{\varphi_2(V)}=k$. 
View $\varphi_2$ as a morphism over $\overline{\FF}_q$ and let $V_{\varphi_2}\subseteq \AA^r_{\overline{\FF}_q(Y_1,\dots,Y_k)}$ be its generic fiber. By working over the algebraically closure of $\overline{\FF}_q(Y_1,\dots,Y_k)$ and following the proof of \eqref{eq:degree-of-fiber}, we see that the degree of $V_{\varphi_2}$ is bounded by $d$.
As $q>2d$, applying Lemma~\ref{lem:finite-linear-map-variant} with $S=\FF_q$, we see that there exist linear polynomials $\ell_1,\dots,\ell_{k_V}\in\FF_q[X_1,\dots,X_r]$
such that the morphisms
$\pi: V\to \AA^{k_V}_{\overline{\FF}_q}$ defined by $\ell_1,\dots,\ell_{k_V}$ and
$\tau:V_{\varphi_2}\to \AA^{k_V-k}_{\overline{\FF}_q(Y_1,\dots,Y_k)}$ defined by $\ell_1,\dots,\ell_{k_V-k}$ are finite.

Applying the general form of the effective fiber dimension theorem (Theorem~\ref{thm:effective-FDT-general}) to $\varphi_2$, where we choose $t=k_V-k$, we see that there exists a polynomial $P\in\overline{\FF}_q[X_1,\dots, X_r]$ of degree at most $k\cdot \deg V\cdot  \prod_{i=1}^{k} \deg h_i\leq kd$ that does not vanish identically on $V_{\overline{\FF}_q}$ such that the following holds:
Let $\varphi: V_{\overline{\FF}_q}\to\AA^{k_V}_{\overline{\FF}_q}$ be the morphism defined by $\ell_{1},\dots,\ell_{k_V-k}, f_{j_1},\dots,f_{j_k}$.
Then for every $a\in V_{\overline{\FF}_q}$ satisfying $P(a)\neq 0$, it holds that $\dim \varphi^{-1}(\varphi(a))=0$.
Note that $\varphi$ is dominant by the finiteness of $\tau$.

Let $\varphi_1: V\to\AA^{k_V-k}_{\FF_q}$ be the morphism defined by $\ell_{1},\dots,\ell_{k_V-k}$. 
View $\varphi$ as a dominant morphism $V\to \AA^{k_V}_{\FF_q}$ over $\FF_q$.
Then $\varphi=(\varphi_1,\varphi_2)$.
Let $B$ be the set of $a\in V(\FF_q)=V_{\overline{\FF}_q}\cap \FF_q^r$ satisfying $\dim \varphi^{-1}(\varphi(a))\neq 0$.
Then $B\subseteq V_{\overline{\FF}_q}\cap V(P)\cap \FF_q^r$. 
The degree of $V_{\overline{\FF}_q}\cap V(P)$ is bounded by $\deg V\cdot \deg P\leq kd^2$ by B\'ezout's inequality (Lemma~\ref{lem:bezout}).
As $P$ does not vanish identically on $V$, either $\dim(V\cap V(P))=k_V-1$ or $V\cap V(P)=\emptyset$.
So by Lemma~\ref{lem:elementary-bound}, 
\[
|B|\leq |V_{\overline{\FF}_q}\cap V(P)\cap \FF_q^r|\leq kd^2 q^{k_V-1}.
\]
Let
$
\delta=\Pr_{a\sim U_{V(\FF_q)}}[\dim\varphi^{-1}(\varphi(a))\neq 0]
$.
Then 
\[
\delta=\frac{|B|}{|V(\FF_q)|}\leq \frac{kd^2 q^{k_V-1}}{q^{k_V}/2}=2kd^2/q.
\]

For $b\in\FF_q^{k_V-k}$, let $V_b=\varphi_1^{-1}(b)$, and let $V'_b$ be the union of the irreducible components $Z$ of $V_b$ such that $Z$ is absolutely irreducible of dimension $k$ and $\dim \overline{\varphi_2(Z)}=k$.
As $\varphi_1$ is a linear map, we have $\deg V_b\leq \deg V\leq d$ by B\'ezout's inequality.

Note that $(2d^2/q+\delta)^{1/2}\leq (2(k+1)d^2/q)^{1/2}\leq \epsilon$.
By Lemma~\ref{lem:restricting-to-fiber}, with probability at least $1-\epsilon$ over $b\sim \varphi_1(U_{V(\FF_q)})$, it holds that 
\begin{equation}\label{eq:bound-error}
|V_b(\FF_q)\setminus V_b'(\FF_q)|\leq \epsilon\cdot |V_b(\FF_q)|.
\end{equation}

Fix $b\in \varphi_1(V(\FF_q))$ for which \eqref{eq:bound-error} holds.
Then $U_{V_b'(\FF_q)}$ is $\epsilon$-close to $U_{V_b(\FF_q)}$ by Lemma~\ref{lem:dist-distance}.

The set $V_b'(\FF_q)$ is nonempty as $|V_b'(\FF_q)|\geq (1-\epsilon)|V_b(\FF_q)|>0$.
By definition, every irreducible component $Z$ of $V'_b$ is absolutely irreducible of dimension $k$ and satisfies $\dim \overline{\varphi_2(Z)}=k$. 
This also implies that $\dim \overline{f(Z)}\geq k$ for every irreducible component $Z$ of $V'_b$ as the output of $\varphi_2$ is part of that of $f|_V$.
So the distribution $f(U_{V'_b(\FF_q)})$ satisfies the first two conditions of $(n,k,d)$ algebraic sources in Definition~\ref{defn:algebraic-source}.
And the third condition also holds as $\deg V_b\leq \deg V$.
Finally, we know $\dim V'_b=k$.
It follows that $f(U_{V'_b(\FF_q)})$ is a minimal $(n,k,d)$ algebraic source over $\FF_q$. Therefore, $f(U_{V_b(\FF_q)})$ is $\epsilon$-close to a minimal $(n,k,d)$ algebraic source over $\FF_q$.

Let $D'=U_{V(\FF_q)}$ so that $D=f(D')$.
For each $b\in \varphi_1(V(\FF_q))$, the distribution  $D|_{\varphi_1(D')=b}$ is exactly $f(U_{V_b(\FF_q)})$. We have already shown that with probability at least $1-\epsilon$ over $b\sim \varphi_1(U_{V(\FF_q)})$, the distribution $D|_{\varphi_1(D')=b}=f(U_{V_b(\FF_q)})$ is $\epsilon$-close to the minimal $(n,k,d)$ algebraic source  $f(U_{V'_b(\FF_q)})$.
It follows that $D$ is $2\epsilon$-close to a convex combination of minimal $(n,k,d)$ algebraic source over $\FF_q$.

Finally, by Lemma~\ref{lem:decompose-irreducible},
every minimal $(n,k,d)$ algebraic source over $\FF_q$ is $\epsilon$-close to a convex combination of irreducibly minimal $(n,k,d)$ algebraic sources over $\FF_q$.
It follows that $D$ is $3\epsilon$-close to a convex combination of irreducibly minimal $(n,k,d)$ algebraic sources over $\FF_q$.
\end{proof}

Combining Lemma~\ref{lem:decompose-irreducible} and Lemma~\ref{lem:decompose-minimal} yields the following corollary.

\sloppy
\begin{corollary}[Decomposition into irreducibly minimal algebraic sources]\label{cor:decompose}
 Suppose $q\geq \max\{20d^5, 2(k+1)d^2/\epsilon^2\}$, where $\epsilon\in (0,1)$.
Then every $(n,k,d)$ algebraic source over $\FF_q$ is $4\epsilon$-close to a convex combination of irreducibly minimal $(n,k,d)$ algebraic sources over $\FF_q$. 
\end{corollary}
 
 \subsection{Estimating the Min-Entropy of $(n,k,d)$ Algebraic Sources}
 
 We first prove the following lower bound on the min-entropy of an $(n,k,d)$ algebraic source $D$ (or more precisely, a distribution $D'$ close to $D$). The proof uses the decomposition into irreducible $(n,k,d)$ algebraic sources (Lemma~\ref{lem:decompose-irreducible}).
 
 \begin{lemma}\label{lem:entropy-lower-bound}
 Suppose $q\geq \max\{20d^5, 2kd^2/\epsilon\}$, where $\epsilon\in (0,1/2]$.
Then every $(n,k,d)$ algebraic source over $\FF_q$ is $2\epsilon$-close to a $k'$-source over the set $\FF_q^n$, where $k'=k\log q-\log d-2$.
 \end{lemma}
 
 \begin{proof}
  Assume $k>0$ as otherwise the lemma holds trivially.
 Let $D$ be an $(n,k,d)$ algebraic source over $\FF_q$.
 By Lemma~\ref{lem:decompose-irreducible}, we know $D$ is $\epsilon$-close to an irreducible $(n,k,d)$ algebraic source $D'$ over $\FF_q$.
 Suppose $D'=f(U_{V(\FF_q)})$ where $V\subseteq\AA^r_{\FF_q}$ and $f: \AA^r_{\FF_q}\to \AA^n_{\FF_q}$ are as in Definition~\ref{defn:algebraic-source}, and $V$ is absolutely irreducible.
 So $f$ is defined by polynomials $f_1,\dots,f_n\in \mathcal{L}_{h_1,\dots,h_s,\FF_q}$, where $h_1,\dots,h_s\in \FF_q[X_1,\dots,X_r]$, $\deg h_1\geq \dots \geq \deg h_s$, and $\deg V\cdot \prod_{i=1}^{k} \deg h_i \leq d$. 
By the effective Lang--Weil bound (Theorem~\ref{theorem:lang-weil}), we have $|V(\FF_q)|\geq q^{\dim V}/2$.

By Lemma~\ref{lem:dominant-projection},
there exist distinct $i_1,\dots,i_k\in [n]$ such that the polynomial map $\psi:\AA^r_{\FF_q}\to \AA^k_{\FF_q}$ defined by $f_{i_1},\dots,f_{i_k}$ satisfies $\dim \overline{\psi(V)}=k$.
Applying the effective fiber dimension theorem (Corollary~\ref{cor:effective-FDT}) to $\psi$ (viewed as a morphism over $\overline{\FF}_q$),
we see that there exists a polynomial $P\in\overline{\FF}_q[X_1,\dots,X_r]$ of degree at most $k\cdot \deg V\cdot  \prod_{i=1}^{k} \deg h_i\leq kd$ that does not vanish identically on $V_{\overline{\FF}_q}$ such that for every $a\in V_{\overline{\FF}_q}$ satisfying $P(a)\neq 0$, it holds that $\dim \psi|_{V_{\overline{\FF}_q}}^{-1}(\psi(a))=\dim V-k$. 

Let $B=\{a\in V(\FF_q): P(a)=0\}=V_{\overline{\FF}_q}\cap V(P)\cap \FF_q^r$. It follows from Lemma~\ref{lem:elementary-bound} and B\'ezout's inequality (Lemma~\ref{lem:bezout}) that 
\[
|B|\leq \deg V\cdot \deg P\cdot q^{\dim V-1}\leq kd^2 q^{\dim V-1}.
\]
Let $U'$ be the uniform distribution over $V(\FF_q)\setminus B$.
By Lemma~\ref{lem:dist-distance}, the distributions $U_{V(\FF_q)}$ and $U'$ are $\epsilon'$-close, where
\[
\epsilon'=\frac{|B|}{|V(\FF_q)|}\leq \frac{kd^2 q^{\dim V-1}}{q^{\dim V}/2} = 2kd^2/q \leq  \epsilon.
\]
So $D'=f(U_{V(\FF_q)})$ and $f(U')$ are $\epsilon$-close.
It follows that $D$ and $f(U')$ are $2\epsilon$-close.

It remains to prove that $f(U')$ has min-entropy at least $k\log q-\log d-2$.
As $\psi(U')$ can be obtained from $f(U')$ by projecting to a subset of coordinates, it suffices to show that $\psi(U')$ has min-entropy at least $k\log q-\log d-2$.
Consider arbitrary $a\in V(\FF_q)\setminus B$ and let $b=\psi(a)$. We have $P(a)\neq 0$ and hence $\dim \psi|_V^{-1}(b)=\dim V-k$.
The fiber $\psi|_V^{-1}(b)$ is the subvariety of $V$ defined by the $k$ polynomials $f_{i_1}-f_{i_1}(a),\dots,f_{i_k}-f_{i_k}(a)\in \mathcal{L}_{h_1,\dots,h_s,\FF_q}$. By Gaussian elimination, we can construct polynomials $g_1,\dots,g_t$ from $f_{i_1}-f_{i_1}(a),\dots,f_{i_k}-f_{i_k}(a)$ such that $t\leq k$, $\psi|_V^{-1}(b)=V\cap V(g_1,\dots,g_{t})$, and $g_i\in \mathcal{L}_{h_i,\dots,h_s,\FF_q}$ for $i\in [t]$.
In particular, we have $\deg g_i\leq \deg h_i$ for $i\in [t]$. 
It follows from B\'ezout's inequality (Lemma~\ref{lem:bezout}) that
\[
\deg \psi|_V^{-1}(b)\leq \deg V\cdot \prod_{i=1}^t \deg g_i\leq \deg V\cdot \prod_{i=1}^k \deg h_i\leq d.
\]
Lemma~\ref{lem:elementary-bound} then gives $\left|\left(\psi|_V^{-1}(b)\right)(\FF_q)\right|\leq d q^{\dim V-k}$.
Therefore, 
\[
\Pr[\psi(U')=b]=\frac{\left|\left(\psi|_V^{-1}(b)\right)(\FF_q)\right|}{|V(\FF_q)\setminus B|}\leq \frac{\left|\left(\psi|_V^{-1}(b)\right)(\FF_q)\right|}{|V(\FF_q)|/2}\leq \frac{dq^{\dim V-k}}{q^{\dim V}/4}=4d/q^k.
\]
Every element in the support of $\psi(U')$ has the form $b=\psi(a)$ for some $a\in V(\FF_q)\setminus B$. So $\psi(U')$ has min-entropy at least $-\log(4d/q^k)=k\log q-\log d-2$, as desired.
\end{proof}

The next proposition complements Lemma~\ref{lem:entropy-lower-bound} and gives an upper bound on the min-entropy.

\begin{proposition}\label{prop:entropy-upper-bound}
Suppose $q\geq 20d^5$.
Let $D$ be an $(n,k,d)$ algebraic source over $\FF_q$ such that $k$ is maximal with respect to this condition, i.e., $D$ is not an $(n, k+1, d)$ algebraic source over $\FF_q$.
Then the statistical distance between $D$ and any $(k\log q+2\log d+2)$-source is at least $\frac{1}{4d}$.
Moreover, if $D$ is an irreducible $(n,k,d)$ algebraic source over $\FF_q$, then  the statistical distance between $D$ and any $(k\log q+\log d +1)$-source is at least $\frac{1}{2}$. 
\end{proposition}

\begin{proof}
Suppose $D=f(U_{V(\FF_q)})$ where $V\subseteq\AA^r_{\FF_q}$ and $f: \AA^r_{\FF_q}\to \AA^n_{\FF_q}$ are as in Definition~\ref{defn:algebraic-source}. As $D$ is not an $(n,k+1,d)$ algebraic source over $\FF_q$, we know $V$ has an irreducible component $V_0$ of dimension $\dim V$ that is absolutely irreducible such that the dimension of $\overline{f(V_0)}$ is exactly $k$. We have $|V_0(\FF_q)|\geq q^{\dim V}/2$ by the effective Lang--Weil bound (Theorem~\ref{theorem:lang-weil}) and the assumption that $q\geq 20d^5$. Also note that $|V(\FF_q)|\leq d q^{\dim V}$ by Lemma~\ref{lem:elementary-bound}.

Let $W=\overline{f(V_0)}$. Then $\deg W\leq d$ by Lemma~\ref{lem:deg-poly-image} and the third condition in Definition~\ref{defn:algebraic-source}.
So $|W(\FF_q)|\leq d q^k$ by Lemma~\ref{lem:elementary-bound}.

Let $D'$ be a $k'$-source over the set $\FF_q^n$, where $k'=k\log q+2\log d+2$. Then
\[
\Pr[D'\in W(\FF_q)]\leq |W(\FF_q)|\cdot 2^{-k'}\leq dq^k 2^{-k'}=\frac{1}{4d}.
\]
On the other hand, as $f(V_0(\FF_q))\subseteq W(\FF_q)$ and $D=f(U_{V(\FF_q)})$, we have
\[
\Pr[D\in W(\FF_q)] \geq \frac{|V_0(\FF_q)|}{|V(\FF_q)|}\geq \frac{q^{\dim V}/2}{dq^{\dim V}}= \frac{1}{2d}.
\]
So the statistical distance between $D$ and $D'$ is at least $\frac{1}{2d}-\frac{1}{4d}=\frac{1}{4d}$.

Now assume $D$ is an irreducible $(n,k,d)$ algebraic source over $\FF_q$. So $V$ may be chosen to be irreducible. Then $V_0=V$ and hence $\Pr[D\in W(\FF_q)]=1$. Let $k''=k\log q+\log d +1$ and let $D''$ be a $k''$-source over the set $\FF_q^n$.
Then $\Pr[D''\in W(\FF_q)]\leq |W(\FF_q)|2^{-k''}\leq dq^k 2^{-k''}=\frac{1}{2}$. So the statistical distance between $D$ and $D''$ is at least $1-\frac{1}{2}=\frac{1}{2}$.
\end{proof}


\section{Extracting a Short Seed}\label{sec:extracting-seed}

In this section, we consider the problem of constructing explicit deterministic extractors for $(n,k,d)$ algebraic sources over a finite field $\FF_q$ in the special case where $k=1$.

The main results of this section are explicit constructions of deterministic extractors that extract almost $\log q$ bits from $(1, 1, d)$ algebraic sources and, more generally, $(n,1,d)$ algebraic sources over $\FF_q$. They  are used as building blocks in the construction of the full-fledged  deterministic extractors that extract most min-entropy from $(n,k,d)$ algebraic sources.

Formally, we prove the following theorems.

\begin{theorem}[Extractor for $(1,1,d)$ algebraic sources]\label{thm:single-polynomial}
Let $d\in\NN^+$ and $\epsilon\in (0,1/2]$.
Suppose $q\geq c_0 d^{5}/\epsilon^2$, where $c_0>0$ is a large enough absolute constant.
Then there exists an explicit  $\epsilon$-extractor $\Ext:\FF_q\to\{0,1\}^m$ for $(1,1,d)$ algebraic sources over $\FF_q$ such that $m\geq \log q-2\log\log p-O(\log(d/\epsilon))$.
\end{theorem}

\begin{theorem}[Extractor for $(n,1,d)$ algebraic sources]\label{thm:single-polynomial-general}
Let $d\in\NN^+$ and $\epsilon\in (0,1/2]$.
Suppose $q\geq  (nd/\epsilon)^{c_0}$, where $c_0>0$ is a large enough absolute constant.
Then there exists an explicit  $\epsilon$-extractor $\Ext:\FF_q\to\{0,1\}^m$ for $(n,1,d)$ algebraic sources over $\FF_q$ such that $m\geq \log q-2\log\log p-O( \log(nd/\epsilon))$.
\end{theorem}

Theorem~\ref{thm:single-polynomial-general} is derived from Theorem~\ref{thm:single-polynomial}.
As in \cite{dvir-gabizon-wigderson, dvir-varieties}, the proof of Theorem~\ref{thm:single-polynomial} uses Bombieri's estimate for exponential sums (Theorem~\ref{thm:bombieri}). However, the argument in \cite{dvir-gabizon-wigderson, dvir-varieties} works only when the characterisitic $p$ is large. Moreover, it only yields an extractor that extracts $c\log q$ bits for some constant $c\leq 1/2$. We introduce new ideas that allow us to extract almost $\log q$ bits regardless of the characteristic $p$. 

\subsection{Exponential Sums over Finite Fields of Arbitrary Characteristic}

The purpose of this subsection is to prove the following estimate for exponential sums over curves, even over finite fields of small characteristics. Recall that Bombieri's estimate (Theorem~\ref{thm:bombieri}) is valid as long as the polynomial $f$ does not have the form $g^p - g$ on the curve. One way to deal with this difficulty is to require $p$ to be large. However, we would like to get meaningful results for arbitrary $p$, and we do this by paying the cost of excluding a small subgroup of characters from the estimate.

\begin{lemma}\label{lem:bias-over-irred-curve}
Let $C\subseteq \AA_{\FF_q}^n$ be an irreducible affine curve of degree $d_1$ over a finite field $\FF_q$ of characteristic $p$, and let $f\in\FF_q[X_1,\dots,X_n]$ be a polynomial of degree $d_2$ that is not constant on $C$.
Then the set of characters $\chi\in\widehat{\FF_q}$ for which 
\begin{equation}\label{eq:character-sum-bound}
\left|\sum_{x\in C(\FF_q)}\chi(f(x))\right|\leq (d_1^2+2d_1 d_2-3d_1)q^{1/2}+d_1^2
\end{equation}
fails to hold is contained in a subgroup of $\widehat{\FF_q}$ of size at most $d_1d_2$.
\end{lemma}

Before proving Lemma \ref{lem:bias-over-irred-curve}, we require some preliminary results.
Recall that for a formal Laurent series $f\in \FF((T))$, we denote by $\ord(f)$ the least degree of the terms that appear in $f$, i.e., $f=c_0 T^{\ord(f)}+c_1 T^{\ord(f)+1}+\cdots$ where $c_0\neq 0$. And $\ord(f)=+\infty$ if $f=0$.
For $i\in\ZZ$, denote by $\coeff_i(f)$ the coefficient of $T^i$ in $f$.

\begin{restatable}{lemma}{laurentexpansion}\label{lem:laurent-expansion}
Let $C_0\subseteq \AA^n$ be an irreducible affine curve of degree $d$ over an algebraically closed field $\FF$. Then there exists an $\FF$-linear field embedding $\tau:\FF(C_0)\hookrightarrow \FF((T))$ such that for any polynomial $f\in\FF[X_1,\dots,X_n]$ of degree $d$ that is not constant on $C_0$, the map $\tau$ sends $f$ to $\tilde{f}\in\FF((T))$ such that
\[
-\deg(C_0)\cdot d\leq \ord(\tilde{f})<0.
\]
\end{restatable}

We defer the proof of Lemma~\ref{lem:laurent-expansion} to Appendix~\ref{sec:misc}.

\begin{lemma}\label{lem:linearized-poly}
Let $\FF$ be a field of characteristic $p$.
Suppose $f,g\in\FF((T))$ such that $f=g^p-g$ and $\ord(f)<0$.
Let $t<0$ and $e\geq 0$ be integers such that $t$ is coprime to $p$ and $tp^{e+1}<\ord(f)\leq tp^e$.
Then
\[
\sum_{i=0}^e(\coeff_{tp^{i}}(f))^{p^{e-i}}=0.
\]
\end{lemma}

\begin{proof}
The lemma can be proved by expressing the coefficients of $f$ in terms of those of $g$. We give an alternative proof.
Note that by taking a field extension, we may assume $\FF$ is algebraically closed and hence a perfect field. 
Then the map $a\mapsto a^{1/p}$ is an automorphism of $\FF$.
We will instead prove 
\begin{equation}\label{eq:coeff-identity}
\sum_{i=0}^{+\infty}(\coeff_{tp^{i}}(f))^{1/p^i}=0.
\end{equation}
Note that the LHS of \eqref{eq:coeff-identity} is actually a finite sum and equals $\sum_{i=0}^e(\coeff_{tp^{i}}(f))^{1/p^i}$.
The lemma then follows by raising both sides to the $p^e$-th power.

Note that to prove \eqref{eq:coeff-identity}, we may assume all terms in $g$ have degree $tp^i$ for $i\in\NN$, as ignoring the other terms does not affect the coefficients appearing in \eqref{eq:coeff-identity}.
Moreover, as $\coeff_{tp^i}(\cdot)$ and the map $x\mapsto x^{1/p}$ are both linear in characteristic $p$, if \eqref{eq:coeff-identity} holds for $f_1,f_2\in\FF((T))$, then it also holds for $f_1+f_2$. So we may assume $g$ contains only one term $c T^{tp^i}$ where $c\in\FF$ and $i\in \NN$, and hence $f=c^p T^{tp^{i+1}}-c T^{tp^i}$. The LHS of \eqref{eq:coeff-identity} is then $(c^p)^{1/p^{i+1}}-c^{1/p^i}=0$.
\end{proof}

\begin{proof}[Proof of Lemma \ref{lem:bias-over-irred-curve}]
 Fix an irreducible component $C_0$ of the affine curve $C_{\overline{\FF}_q}$ over $\overline{\FF}_q$.
It is a standard fact that the natural inclusion $\FF_q[C]\hookrightarrow \overline{\FF}_q[C_{\overline{\FF}_q}]$ induces an inclusion $\FF_q[C]\hookrightarrow \overline{\FF}_q[C_0]$.\footnote{This uses the fact that $\overline{\FF}_q[C_{\overline{\FF}_q}]$ is an \emph{integral extension} of $\FF_q[C]$. The kernel of $\FF_q[C]\to \overline{\FF}_q[C_0]$ is $I\cap \FF_q[C]$, where $I$ is the minimal prime ideal of $\overline{\FF}_q[C_{\overline{\FF}_q}]$ defining the irreducible component $C_0$. As $C_{\overline{\FF}_q}$ is equidimensional of dimension one, $I$ is not a maximal ideal of $\overline{\FF}_q[C_{\overline{\FF}_q}]$.  Then $I\cap \FF_q[C]$ is a prime ideal of  $\FF_q[C]$ that is not maximal either by \cite[Corollary~5.8]{AM69}. As $C$ is an irreducible curve, this implies  $I\cap \FF_q[C]=0$. See \cite[Chapter~5]{AM69} for more details about integral extensions.}
In particular, as $f$ is not constant on $C$, it is not constant on $C_0$ either.

Denote by $\sigma$ be the character $x\mapsto e^{2\pi i x/p}$ of $\FF_p$.
For $\alpha\in\FF_q$, denote by $\chi_\alpha$ the character of $\FF_q$ sending $x$ to $(\sigma\circ \Tr)(\alpha x)$.
The map $\alpha\mapsto \chi_\alpha$ is a one-to-one correspondence between $\FF_q$ and $\widehat{\FF_q}$.
Consider a character $\chi_\alpha$ for which \eqref{eq:character-sum-bound} does not hold.
Then by Bombieri's estimate (Theorem~\ref{thm:bombieri}), there exists $g\in \overline{\FF}_q[X_1,\dots,X_n]$ such that $\alpha f-(g^p-g)$ vanishes identically on $C$, and hence also on $C_0$.

View $f$ and $g$ as elements of $\overline{\FF}_q(C_0)$.
By Lemma~\ref{lem:laurent-expansion}, there exists an $\overline{\FF}_q$-linear field embedding  $\tau: \overline{\FF}_q(C_0)\hookrightarrow \overline{\FF}_q((T))$ such that $\tilde{f}:=\tau(f)$ satisfies
\[
 -\deg(C_0)\cdot 
d_2 \leq \ord(\tilde{f})<0.
\]
Let $\tilde{g}=\tau(g)$. Then $\alpha \tilde{f}=\tilde{g}^p-\tilde{g}$ as $\tau$ is an $\overline{\FF}_q$-linear field embedding and $\alpha f=g^p-g$ in $\overline{\FF}_q[C_0]$.
Write $\ord(\tilde{f})=tp^e$ where $t<0$ is coprime to $p$.
By Lemma~\ref{lem:linearized-poly},
\begin{equation}\label{eq:poly-relation}
0=\sum_{i=0}^e(\coeff_{tp^{i}}(\alpha\tilde{f}))^{p^{e-i}}=\sum_{i=0}^e\alpha^{p^{e-i}}(\coeff_{tp^{i}}(\tilde{f}))^{p^{e-i}}.
\end{equation}
The RHS of \eqref{eq:poly-relation} is a nonzero polynomial in $\alpha$  independent of $g$, and its degree is bounded by $p^e$. Also note that this polynomial is a linearized polynomial, i.e., the degree of every monomial is a power of $p$. So its roots in $\FF_q$ form a subgroup $H$ whose size is at most $p^e$. 

Then the set $S:=\{\chi_\alpha: \alpha\in H\}$ contains all the characters for which \eqref{eq:character-sum-bound} fails to hold.
As $\chi_\alpha$ is defined by $x\mapsto (\sigma\circ \Tr)(\alpha x)$ and $\Tr(\cdot)$ is linear, the set $S$ is a subgroup of $\widehat{\FF_q}$, whose size is at most $|H|\leq p^e \leq |\ord(\tilde{f})|\leq \deg(C_0)\cdot d_2\leq d_1 d_2$.
\end{proof}

\begin{remark*}
We need the curve $C$ to be irreducible in Lemma~\ref{lem:bias-over-irred-curve} so that the ``bad'' characters are contained in a single subgroup of $\widehat{\FF_q}$ of size at most $d_1d_2$. For a reducible curve $C$, a similar proof  shows that these characters are contained in a subset $S\subseteq\widehat{\FF_q}$ of size at most $d_1d_2$ such that $S$ is the union of a collection of subgroups of $\widehat{\FF_q}$, one for each irreducible component of $C$.
\end{remark*}

\subsection{Proofs of Theorem~\ref{thm:single-polynomial} and Theorem~\ref{thm:single-polynomial-general}}

Theorem~\ref{thm:single-polynomial} can be proved easily using the techniques we have developed so far. First, we show that a distribution of the form $f(U_{C(\FF_q)})$ is a strongly $(\epsilon,d)$-biased source, and even an $\epsilon$-biased source if the characteristic $p$ is large, where $C$ is a low-degree absolutely irreducible affine curve over $\FF_q$ and $f$ is a low-degree polynomial. 

\begin{lemma}\label{lem:curve-epsilon-biased}
Let $C\subseteq\AA^n_{\FF_q}$ be an absolutely irreducible affine curve of degree $d_1$ over a finite field $\FF_q$ of characteristic $p$. 
Let $f\in \FF_q[X_1,\dots,X_n]$ be a polynomial of degree $d_2$ that is not constant on $C$. 
Let $d\geq d_1d_2$ and suppose $q\geq 20 d^{5}$.
Then $f(U_{C(\FF_q)})$ is strongly $(\epsilon, d)$-biased over the set $\FF_q$, where $\epsilon=8d^2/q^{1/2}$. Furthermore, if $p>d_1d_2$, then $f(U_{C(\FF_q)})$ is $\epsilon$-biased over the set $\FF_q$.
\end{lemma}

\begin{proof}
By Lemma~\ref{lem:bias-over-irred-curve}, there exists a subgroup $S\subseteq \widehat{\FF_q}$ of size most $d_1d_2\leq d$ such that for all $\chi\in  \widehat{\FF_q}\setminus S$, it holds that 
\[
\left|\sum_{x\in C(\FF_q)}\chi(f(x))\right|\leq (d_1^2+2d_1 d_2-3d_1)q^{1/2}+d_1^2\leq 4d^2q^{1/2}.
\]
By the effective Lang--Weil bound (Theorem~\ref{theorem:lang-weil}) and the assumption $q\geq 20d^5$, we have $|C(\FF_q)|\geq q/2$.
It follows that for $\chi\in  \widehat{\FF_q}\setminus S$,
\[
\left|\Ex[\chi(f(U_{C(\FF_q)}))]\right|=\frac{\left|\sum_{x\in C(\FF_q)}\chi(f(x))\right|}{|C(\FF_q)|}\leq \frac{4d^2q^{1/2}}{q/2} =8d^2/q^{1/2}=\epsilon.
\]
By definition, this means $f(U_{C(\FF_q)})$ is strongly $(\epsilon, d)$-biased.

Now assume $p>d_1 d_2$. As the size of $S\subseteq \widehat{\FF_q}$ is a power of $p$ and $|S|\leq d_1 d_2$, we must have $|S|=1$, i.e., $S$ contains only the trivial character.
So  $f(U_{C(\FF_q)})$ is $\epsilon$-biased.
\end{proof}

We are now ready to prove Theorem~\ref{thm:single-polynomial}. This follows by first reducing to the case of irreducibly minimal $(1,1,d)$ algebraic sources and then applying Lemma~\ref{lem:curve-epsilon-biased} together with the constructions in Section~\ref{sec:low-bias}.

\begin{proof}[Proof of Theorem~\ref{thm:single-polynomial}]
We will construct an $(\epsilon/2)$-extractor $\Ext$ for irreducibly minimal $(1,1,d)$ algebraic sources over $\FF_q$ with the claimed output length.
By Corollary~\ref{cor:decompose} and the assumption that $q\geq c_0 d^{5}/\epsilon^2$, every $(1,1,d)$ algebraic source over $\FF_q$ is $(\epsilon/2)$-close to a convex combination of irreducibly minimal  $(1,1,d)$ algebraic sources over $\FF_q$.
It follows that $\Ext$ is also an $\epsilon$-extractor for $(1,1,d)$ algebraic sources over $\FF_q$.

Let $\epsilon_0:=8d^2/q^{1/2}$ and let $c>0$ be a large enough absolute constant. We consider two different cases depending on how large the characteristic $p$ is.

\textbf{Case 1:} $p\leq  (d/\epsilon)^{c}$. In this case, let $\Ext: \FF_q\to \FF_p^t=\{0,1\}^m$ be the $(\epsilon/2)$-extractor for strongly $(\epsilon_0, d)$-biased sources given by Theorem~\ref{thm:extractor-for-strongly-biased}, where we set the parameters 
\begin{align*}
n' &=\min\{\lfloor 2\log_p(1/\epsilon_0)-2\log_p(16d/(\epsilon/2)^2)\rfloor, \log_p q\},\\
t &=\lfloor n'-3-2\log_p(2d/(\epsilon/2))\rfloor,~\text{and}\\
m &=t\log p.
\end{align*}
As $\epsilon_0=8d^2/q^{1/2}$ and  $p\leq  (d/\epsilon)^{c}$, we have
\[
m\geq \min\{2\log (1/\epsilon_0), \log q\}-O(\log (d/\epsilon))-O(\log p)=\log q-O(\log(d/\epsilon)).
\]

Consider an irreducibly minimal $(1,1,d)$ algebraic source $D$ over $\FF_q$.
By definition, there exist $r\in\NN^+$, an absolutely irreducible affine curve $C\subseteq\AA_{\FF_q}^r$ over $\FF_q$ of degree $d_1$, and a polynomial $f\in \FF_q[X_1,\dots,X_r]$ of degree $d_2$ that is not constant on $C$ such that $d_1d_2\leq d$ and $D=f(U_{C(\FF_q)})$. By Lemma~\ref{lem:curve-epsilon-biased}, the distribution $D$ is a strongly $(\epsilon_0, d)$-biased distribution over the set $\FF_q$, so that $\Ext(D)$ is $(\epsilon/2)$-close to the uniform distribution over $\FF_q$. It follows that $\Ext$ is an $(\epsilon/2)$-extractor for irreducibly minimal $(1,1,d)$ algebraic sources, as desired.

\textbf{Case 2:} $p> (d/\epsilon)^{c}\geq d_1 d_2$. 
In this case, identify $\FF_q$ with the abelian group $\ZZ_N^t$, where $N=p$ and $t=\log_p q$.
Let $\Ext$ be the map $\ZZ_N^t\to \ZZ_N^{t-1}\times \ZZ_M$ in Lemma~\ref{lem:mod-m-extractor}, where the parameter $M$ will be determined shortly.
By Lemma~\ref{lem:mod-m-extractor}, $\Ext$ is an $\epsilon'$-extractor for $\epsilon_0$-biased distribution, where
\[
\epsilon':=\epsilon_0\cdot (N^{t-1}M)^{1/2}\cdot C\log N+M/N=8d^2\cdot (M/p)^{1/2}\cdot C\log p+M/p.
\]
and $C>0$ is an absolute constant.
By Lemma~\ref{lem:curve-epsilon-biased} and the fact that $p>d_1d_2$, every irreducibly minimal $(1,1,d)$ algebraic source over $\FF_q$ is $\epsilon_0$-biased.
So $\Ext$ is also an $\epsilon'$-extractor for irreducibly minimal $(1,1,d)$ algebraic sources.

We want to choose $M\in\NN^+$ such that $\epsilon'\leq \epsilon/2$.
As $p> (d/\epsilon)^{c}$ and $c>0$ is a large enough constant,
such an integer $M$ exists and we can choose $M$ such that $\log M\geq \log p-2\log\log p-O(\log(d/\epsilon))$.
So $\Ext$ is an $(\epsilon/2)$-extractor for irreducibly minimal $(1,1,d)$ algebraic sources with the output length 
\[
m=(t-1)\log N+\log M=\log q-\log p+\log M\geq \log q-2\log\log p-O(\log (d/\epsilon)).
\]

In both cases above, we ignore the technicality that the size of the range of $\Ext$ may not be a power of two, i.e., $m$ may not be an integer. But one can always turn the size into a power of two at the cost of losing $O(\log (1/\epsilon))$ bits in the output by composing $\Ext$ with a suitable map. The details are left to the reader.
\end{proof}

\begin{remark*}
In the case where $p\leq  (d/\epsilon)^{c}$, the above proof shows that we could avoid losing $2\log\log p$ bits in the output. However, this does not make an essential difference as $2\log\log p$ is dominated by the term $O(\log (d/\epsilon))$ in this case.
\end{remark*}

We now prove Theorem~\ref{thm:single-polynomial-general} by composing the extractors in Theorem~\ref{thm:single-polynomial} with the deterministic rank extractors for varieties constructed in Section~\ref{sec:deterministic-rank-extractor}. 

\begin{proof}[Proof of Theorem~\ref{thm:single-polynomial-general}]

Choose sufficiently large $d'=\Theta(nd^2\log n)$.
Let $\Ext':\FF_q\to\{0,1\}^m$ be an explicit $(\epsilon/2)$-extractor for $(1,1,d')$ algebraic sources over $\FF_q$ as constructed in Theorem~\ref{thm:single-polynomial},
where 
\[
m\geq \log q-2\log\log p-O(\log(d'/\epsilon))=\log q-2\log\log p-O(\log(nd/\epsilon)).
\]
Let $\varphi:\AA^n_{\overline{\FF}_q}\to\AA^1_{\overline{\FF}_q}$ be an explicit $(n,1,d)$ deterministic rank extractor for varieties defined by a polynomial $F\in \FF_q[X_1,\dots,X_n]$ as constructed in Corollary~\ref{cor:MDS-extractor-for-curves}, where $\deg F=O(nd\log n)$. View $\varphi$ as a morphism $\AA^n_{\FF_q}\to\AA^1_{\FF_q}$ over $\FF_q$.

Let $\Ext:=\Ext'\circ \varphi|_{\FF_q^n}:\FF_q^n\to\{0,1\}^m$.
We claim that $\Ext$ is an $\epsilon$-extractor for $(n,1,d)$ algebraic sources over $\FF_q$.
To see this, consider an irreducibly minimal $(n,1,d)$ algebraic source $D=f(U_{C(\FF_q)})$ arising from an absolutely irreducible affine curve $C\subseteq\AA^n_{\FF_q}$ and a polynomial map $f: \AA^r_{\FF_q}\to \AA^n_{\FF_q}$ defined by polynomials $f_1,\dots,f_n$.
Let $d_1=\deg C$ and $d_2=\max\{\deg f_1,\dots, \deg f_n\}$. We have $d_1 d_2\leq d$ by Definition~\ref{defn:algebraic-source}.
Then $\varphi\circ f: \AA^r_{\FF_q}\to \AA^1_{\FF_q}$ is defined by the polynomial $F(f_1,\dots,f_n)$ of degree $O(d_2\cdot nd\log n)$. 
Note that $\deg C \cdot \deg F(f_1,\dots,f_n)=O(d_1d_2\cdot nd\log n)=O(nd^2\log n)\leq d'$.
So $\varphi(D)=(\varphi\circ f)(U_{C(\FF_q)})$
satisfies the third condition of $(1,1,d')$ algebraic sources over $\FF_q$ in Definition~\ref{defn:algebraic-source} with respect to $C$ and $\varphi\circ f$.

We have $\dim\overline{f(C)}=1$ by Definition~\ref{defn:algebraic-source}
and $\deg \overline{f(C)}\leq d_1d_2\leq d$ by Lemma~\ref{lem:deg-poly-image}.
As $\varphi$ is an $(n,1,d)$ deterministic rank extractor for varieties, we have $\dim \overline{(\varphi\circ f)(C)}=1$.
It follows that $\varphi(D)=(\varphi\circ f)(U_{C(\FF_q)})$ is a $(1,1,d')$ algebraic source over $\FF_q$.
As $\Ext'$ is an explicit $(\epsilon/2)$-extractor for $(1,1,d')$ algebraic sources over $\FF_q$,  $\Ext(D)=\Ext'(\varphi(D))$ is $(\epsilon/2)$-close to $U_m$.

The above proof shows that $\Ext$ is an $(\epsilon/2)$-extractor for irreducibly minimal $(n,1,d)$ algebraic sources over $\FF_q$.
By Corollary~\ref{cor:decompose}, every $(n,1,d)$ algebraic sources over $\FF_q$ is $(\epsilon/2)$-close to a convex combination of irreducibly minimal $(n,1,d)$ algebraic sources over $\FF_q$.
So $\Ext$ is an $\epsilon$-extractor for $(n,1,d)$ algebraic sources over $\FF_q$.
\end{proof}


\section{Deterministic Extractors for $(n,k,d)$ Algebraic Sources} \label{sec:full-extractor}

In this section, we provide our main construction of deterministic extractors for $(n,k,d)$ algebraic sources. Recall that in Section \ref{sec:extracting-seed} we considered the case of $(n,1,d)$ algebraic sources.

We start with the case of $(n,n,d)$ algebraic sources, and we follow our general proof technique as laid out in Section \ref{sec:techniques}: the first step of the construction is applying our extractor from Section \ref{sec:extracting-seed} to obtain a short output, which is then, in the second step, used as a seed for a seeded extractor for sources with high min-entropy (note that even though we have more structure in our source, since we are anyway applying a seeded extractor we might as well use an off-the-shelf construction which works for any source with high min-entropy). Proving that this indeed works requires analyzing the conditional distribution of an $(n,n,d)$ algebraic source under fixing of a subset of the coordinates, which is done in Lemma \ref{lem:blocks}. This construction is presented and analyzed in Section \ref{sec:full-rank}.

In order to remove the assumption that $k=n$ and handle general $(n,k,d)$ algebraic sources, we apply a rank extractor which, roughly speaking, condenses a $k$-dimensional source in an ambient $n$-dimensional space to a $k$-dimensional source in an ambient $k$-dimensional space, and this enables us to use the extractor from Section \ref{sec:full-rank}. As discussed at the end of Section \ref{sec:full-rank}, this can be done using the deterministic rank extractor of Section \ref{sec:deterministic-rank-extractor}, but it would have an undesirable effect on the field size. Thus, we opt to use a \emph{linear seeded rank extractor} (as defined in Section \ref{sec:seeded-rank-extractor}), where the seed of the rank extractor is chosen pseudorandomly using our extractor for $(n,1,d)$ algebraic sources from Section \ref{sec:extracting-seed}.

To summarize, in our composition theorem (Theorem \ref{thm:composition}), we start by applying the extractor for $(n,1,d)$ algebraic sources from Section \ref{sec:extracting-seed} in order to select a seed for the seeded linear rank extractor from Section \ref{sec:seeded-rank-extractor}, we apply the resulting linear map to the source, and then we use the extractor for full-rank sources from Section \ref{sec:full-rank} to obtain the final output. The details of this construction appear in Section \ref{sec:removing-full-rank}.
 
\subsection{Deterministic Extractors for Full-Rank Algebraic Sources}
\label{sec:full-rank}

The following lemma states that irreducible $(n,n,d)$ algebraic sources have a nice recursive structure.
The statement can be extended to general $(n,k,d)$ algebraic sources in some way, but this special case is simpler and suffices for us.

\begin{lemma}\label{lem:blocks}
Suppose $q\geq \max\{20d^5, 2(n+1)d^2/\epsilon^2\}$, where $\epsilon\in (0,1)$.
Let $D=(D_1, D_2)$ be an irreducible $(n,n,d)$ algebraic source over $\FF_q$, where $D_1$ and $D_2$ are distributions over $\FF_q^{n_1}$ and $\FF_q^{n_2}$ respectively and $n_1+n_2=n$. 
Then the following holds:
\begin{enumerate}
    \item $D_1$ is an irreducible $(n_1,n_1,d)$ algebraic source over $\FF_q$.
    \item With probability at least $1-\epsilon$ over $b\sim D_1$, the distribution $D_2|_{D_1=b}$ is $\epsilon$-close to an $(n_2,n_2,d)$ algebraic source over $\FF_q$.
\end{enumerate}
\end{lemma}

\begin{proof}
 Suppose $D=f(U_{V(\FF_q)})$ where $V\subseteq\AA^r_{\FF_q}$ and $f: \AA^r_{\FF_q}\to \AA^n_{\FF_q}$ are as in Definition~\ref{defn:algebraic-source}.
 So $V$ is absolutely irreducible and $\dim \overline{f(V)}=n$. And $f$ is defined by polynomials $f_1,\dots,f_n\in \mathcal{L}_{h_1,\dots,h_s,\FF_q}$, where $h_1,\dots,h_s\in \FF_q[X_1,\dots,X_r]$, $\deg h_1\geq \dots \geq \deg h_s$, and $\deg V\cdot \prod_{i=1}^{n} \deg h_i \leq d$. 
 
 Let $\varphi_1: V\to \AA^{n_1}_{\FF_q}$ be the morphism defined by $f_1,\dots,f_{n_1}$, and similarly, let $\varphi_2: V\to \AA^{n_2}_{\FF_q}$ be the morphism defined by $f_{n_1+1},\dots,f_{n}$. Then $f|_V=(\varphi_1,\varphi_2)$ and $D_i=\varphi_i(U_{V(\FF_q)})$ for $i=1,2$.
 
We know $V$ is absolutely irreducible. And the dimension of $\overline{\varphi_1(V)}$ must be $n_1$ since otherwise $\dim \overline{f(V)}$ cannot reach $n$. 
 By definition, $D_1$ is an irreducible $(n_1,n_1,d)$ algebraic source over $\FF_q$. This proves the first claim.
 
As $\epsilon^2\geq 2(n+1)d^2/q$ and $q\geq 20d^5$, by Lemma~\ref{lem:conditional-dist}, with probability at least $1-\epsilon$ over $b\sim D_1$,  
the distribution $D_2|_{D_1=b}$ is $\epsilon$-close to an $(n_2, n_2, d)$ algebraic source over $\FF_q$, proving the second claim.
\end{proof}

We also need the following explicit construction of seeded extractors given by Goldreich and Wigderson \cite{GW97}, which is based on expander graphs.

\begin{theorem}[\cite{GW97}]\label{thm:extractor-GW}
For $n\in \NN$, $0\leq \Delta\leq n$ and $\epsilon>0$,
there exists an explicit seeded $\epsilon$-extractor $\Ext:\{0,1\}^n\times \{0,1\}^\ell\to\{0,1\}^n$ for $(n-\Delta)$-sources with $\ell=O(\Delta+\log(1/\epsilon))$.
\end{theorem}

We now provide our construction for full-rank algebraic sources. Our construction follows the general paradigm mentioned in Section \ref{sec:techniques}: we first apply our extractor from Theorem \ref{thm:single-polynomial} to obtain a short output, which is then used as a seed to the extractor from Theorem \ref{thm:extractor-GW}. Proving that this ``randomness recycling'' technique works in this setting requires Lemma \ref{lem:blocks}.

\begin{theorem}[Extractor for $(n,n,d)$ algebraic sources]\label{thm:full-rank-extractor}
Let $n, d\in\NN^+$ and $\epsilon\in (0,1/2]$. 
Suppose $q\geq (nd/\epsilon)^{c_0}$, where $c_0>0$ is a large enough absolute constant.
Then there exists an explicit $\epsilon$-extractor $\Ext:\FF_q\to\{0,1\}^m$ for $(n,n,d)$ algebraic sources over $\FF_q$ such that  
 $m\geq n\log q-2\log\log p-O(\log (d/\epsilon))$.
\end{theorem}

\begin{proof}
If $n=1$, then the statement holds by Theorem~\ref{thm:single-polynomial}.
So assume $n>1$.
By Corollary~\ref{cor:decompose}, every $(n,n,d)$ algebraic source over $\FF_q$ is $(\epsilon/2)$-close to an irreducible $(n,n,d)$ algebraic source over $\FF_q$.
So it suffices to construct an explicit $(\epsilon/2)$-extractor $\Ext$ for irreducible $(n,n,d)$ algebraic sources over $\FF_q$ with the claimed output length.

Let $\epsilon'=\epsilon/10$.
We construct $\Ext$ as follows.
\begin{enumerate}
\item 
Let $m_1=\lceil (n-1)\log q\rceil$ and $\Delta=\log d + 3$.
Let $\Ext_1:\FF_q^{n-1}\times\{0,1\}^\ell\to \{0,1\}^{m_1}$ be an explicit seeded $\epsilon'$-extractor for $k$-sources, where $k=m_1-\Delta\leq (n-1)\log q-\log d-2$ and $\ell=O(\Delta+\log(1/\epsilon))$.
This can be done by using Theorem~\ref{thm:extractor-GW}
to construct an $\epsilon'$-extractor $\Ext_1':\{0,1\}^{m_1}\times\{0,1\}^\ell\to\{0,1\}^{m_1}$ for $k$-sources and then composing it with an injection $\FF_q^{n-1}\hookrightarrow \{0,1\}^{m_1}$.

\item Let $\Ext_2: \FF_q\to \{0,1\}^{m_2}$ be an explicit $\epsilon'$-extractor for $(1,1,d)$ algebraic sources over $\FF_q$  such that $m_2\geq \log q-2\log\log p- O(\log (d/\epsilon))$. Moreover, we assume the constant $c_0>0$ is large enough so that $m_2\geq \ell=O(\log(d/\epsilon))$.
Such an extractor can be constructed by Theorem~\ref{thm:single-polynomial}.
\item The map $\Ext$ takes $(x_1,x_2)\in\FF_q^{n-1}\times\FF_q$ and feeds $x_2$ to $\Ext_2$. Let $y=(y_1,y_2)$ be the output of $\Ext_2$, where $y_1\in\{0,1\}^\ell$. (This is possible as $m_2\geq \ell$.) Then $\Ext$ outputs $(\Ext_1(x_1,y_1), y_2)$. 
\end{enumerate}

$\Ext$ outputs $m:=m_1+m_2-\ell$ bits. Plugging in the values of $m_1,m_2,\ell$. We see that the output length of $\Ext$ is as claimed.

Let $D$ be an irreducibe $(n,n,d)$ algebraic sources over $\FF_q$. Write $D=(D_1,D_2)$ where $D_1$ is distributed over $\FF_q^{n-1}$ and $D_2$ is distributed over $\FF_q$.

By Lemma~\ref{lem:blocks}, with probability at least $1-\epsilon'$ over $x_1\sim D_1$, the  distribution $D_2|_{D_1=x_1}$ is $\epsilon'$-close to a $(1,1,d)$ algebraic source over $\FF_q$. As $\Ext_2$ is an $\epsilon'$-extractor for $(1,1,d)$ algebraic sources over $\FF_q$, we see that with  probability at least $1-\epsilon'$ over $x_1\sim D_1$, it holds that $\Ext_2(D_2)|_{D_1=x_1}=_{2\epsilon'} U_{m_2}$.
By Lemma~\ref{lem:approx-2},
\[
(D_1, \Ext_2(D_2))=_{3\epsilon'} D_1\times U_{m_2}.
\]
By Lemma~\ref{lem:blocks}, $D_1$ is an $(n-1,n-1,d)$ algebraic source over $\FF_q$.
By Lemma~\ref{lem:entropy-lower-bound} and the fact that $k\leq (n-1)\log q-\log d-2$, the distribution $D_1$ is $\epsilon'$-close to a $k$-source $D_1'$.
So
\begin{equation}\label{eq:dist-eq1}
(D_1, \Ext_2(D_2))=_{4\epsilon'} D'_1\times U_{m_2}.
\end{equation}
By the definition of $\Ext$, \eqref{eq:dist-eq1} implies
\begin{equation}\label{eq:dist-eq2}
\Ext(D)=_{4\epsilon'} \Ext_1(D_1'\times U_{\ell})\times U_{m_2-\ell}.
\end{equation}
As $\Ext_1$ is a seeded $\epsilon'$-extractor for $k$-sources, we see 
\begin{equation}\label{eq:dist-eq3}
\Ext_1(D_1'\times U_{\ell})=_{\epsilon'} U_{m_1}.
\end{equation}
It follows from \eqref{eq:dist-eq2} and \eqref{eq:dist-eq3} that $\Ext(D)=_{5\epsilon'} U_{m_1}\times U_{m_2-\ell}=U_m$.
As $5\epsilon'=\epsilon/2$, we see that $\Ext$ is an $(\epsilon/2)$-extractor for irreducible $(n,n,d)$ algebraic sources over $\FF_q$, and hence an $\epsilon$-extractor for $(n,n,d)$ algebraic sources over $\FF_q$.
\end{proof}

One can remove the full-rank assumption and construct an extractor for $(n,k,d)$ algebraic sources over $\FF_q$ by composing the extractor in Theorem~\ref{thm:full-rank-extractor} with the deterministic rank extractor for varieties in Section~\ref{sec:deterministic-rank-extractor}.
This argument was used by Dvir, Gabizon and Wigderson \cite{dvir-gabizon-wigderson}, except that they considered polynomial sources only and used a different construction of deterministic rank extractors.
The downside of this argument, however, is that such a deterministic rank extractor is necessarily nonlinear. In particular,  our rank extractor uses polynomials of degree at least $\poly(n)$, and so does the one in \cite{dvir-gabizon-wigderson}.
Composing with such a rank extractor increases the degree of each polynomial in the polynomial map by at least a $\poly(n)$ factor.
The resulting field size $q$ would then depend at least polynomially on $n^k$, or $n^n$ if $k=\Theta(n)$, assuming that we want to extract about $ k\log q$ bits.

In the next subsection, we show how to remove the full-rank assumption more efficiently using a linear seeded rank extractor for varieties.

\subsection{Removing the Full-Rank Assumption}
\label{sec:removing-full-rank}

We now remove the full-rank assumption in Theorem~\ref{thm:full-rank-extractor} without significantly increasing the required field size.
This is done by extending an argument in \cite{GRS06,gabizon-raz}.

\begin{lemma}[{\cite[Lemma~2.6]{GRS06}}]\label{lem:close-to-dist}
Let $D=(D_1, D_2)$ be a joint distribution over a finite product set $A\times B$. Suppose $D$ is $\epsilon$-close to $U_A\times D_2$. Then for all $y\in \supp(D_1)$, the conditional distribution $D_2|_{D_1=y}$ is $\epsilon'$-close to $D_2$, where $\epsilon'=2\epsilon\cdot |A|$.
\end{lemma}

We also need the following lemma.

\begin{lemma}\label{lem:conditioning}
Suppose $q\geq \max\{20d^5, 2(k+1)d^2/\epsilon^2\}$, where $\epsilon\in (0,1)$. 
Let $D=f(U_{V(\FF_q)})$ be an irreducible $(n,k,d)$ algebraic source over $\FF_q$ arising from an affine variety $V\subseteq\AA^r_{\FF_q}$ and a polynomial map $f: \AA^r_{\FF_q}\to \AA^n_{\FF_q}$  as in Definition~\ref{defn:algebraic-source}.
Let $\pi: \AA^n_{\FF_q}\to \AA^{k-1}_{\FF_q}$ be a linear map over $\FF_q$ such that $\dim\overline{(\pi\circ f)(V)}=k-1$.
Then with probability at least $1-\epsilon$ over $b\sim \pi(D)$, the distribution $D|_{\pi(D)=b}$ is $\epsilon$-close to an  $(n,1,d)$ algebraic source over $\FF_q$. 
\end{lemma}

\begin{proof}
The proof is similar to that of Lemma~\ref{lem:conditional-dist}, although we are not able to derive the statement directly from Lemma~\ref{lem:conditional-dist} for technical reasons.

By definition, $V$ is absolutely irreducible, and $f$ is defined by polynomials $f_1,\dots,f_n\in \mathcal{L}_{h_1,\dots,h_s,\FF_q}$, where $h_1,\dots,h_s\in \FF_q[X_1,\dots,X_r]$, $\deg h_1\geq \dots \geq \deg h_s$, and $\deg V\cdot \prod_{i=1}^{k} \deg h_i \leq d$.

Append a coordinate $Y_{u}$ of $\AA^n_{\FF_q}$ for some $u\in [n]$ to the output of $\pi$ to obtain a linear map $\pi':\AA^n_{\FF_q}\to \AA^k_{\FF_q}$ such that $\dim \overline{(\pi'\circ f)(V)}=k$. This is possible by Lemma~\ref{lem:dominant-projection}. 

Let $\varphi=(\pi'\circ f)|_V: V\to \AA^k_{\FF_q}$.
Then $\varphi$ is dominant. We can write $\varphi=(\varphi_1,\varphi_2)$ where $\varphi_1:V\to\AA^{k-1}_{\FF_q}$ equals $(\pi\circ f)|_V$ and $\varphi_2: V\to\AA^1_{\FF_q}$ is defined by the polynomial $f_u$.
As $\pi$ is a linear map over $\FF_q$ and the set $\mathcal{L}_{h_1,\dots,h_s,\FF_q}$ is closed under taking linear combinations over $\FF_q$, we know $\varphi$ is defined by polynomials in $\mathcal{L}_{h_1,\dots,h_s,\FF_q}$. 

Let $\delta=\Pr_{a\sim U_{V(\FF_q)}}[\dim\varphi^{-1}(\varphi(a))\neq \dim V-k]$.
We first bound $\delta$.
By the effective fiber dimension theorem (Corollary~\ref{cor:effective-FDT}), there exists a polynomial $P\in\overline{\FF}_q[X_1,\dots,X_r]$ of degree at most $k\cdot \deg V\cdot  \prod_{i=1}^{k} \deg h_i\leq kd$ that does not vanish identically on $V_{\overline{\FF}_q}$ such that for every $a\in V_{\overline{\FF}_q}$ satisfying $P(a)\neq 0$, the fiber $\varphi^{-1}(\varphi(a))$ is equidimensional of dimension $\dim V-k$.
 Let $B$ be the set of $a\in V(\FF_q)$ such that $\dim \varphi^{-1}(\varphi(a))\neq \dim V-k$. Then $B\subseteq V_{\overline{\FF}_q}\cap V(P)\cap \FF_q^r$.
 By B\'ezout's inequality (Lemma~\ref{lem:bezout}) and Lemma~\ref{lem:elementary-bound}, we have
 \[
|B|\leq \deg V\cdot \deg P\cdot q^{\dim V-1}\leq kd^2 q^{\dim V-1}.
 \]
Therefore, 
\[
\delta=\frac{|B|}{|V(\FF_q)|}\leq \frac{kd^2 q^{\dim V-1}}{q^{\dim V}/2}=2kd^2/q,
\]
where we use the fact $|V(\FF_q)|\geq q^{\dim V}/2$ that follows from Theorem~\ref{theorem:lang-weil}.
So $\epsilon\geq (2(k+1)d^2/q)^{1/2}\geq (2d^2/q+\delta)^{1/2}$.

 For $b\in\FF_q^{k-1}$, let  
$V_b=\varphi_1^{-1}(b)$ and let $V'_b$ be the union of the irreducible components $Z$ of $V_b$ such that $Z$ is absolutely irreducible of dimension $\dim V-(k-1)$ and $\dim \overline{\varphi_2(Z)}=1$.
By Lemma~\ref{lem:restricting-to-fiber},  with probability at least $1-\epsilon$ over $b\sim \varphi_1(U_{V(\FF_q)})=\pi(D)$, it holds that 
\begin{equation}\label{eq:bounding-bad-set}
|V_b(\FF_q)\setminus V_b'(\FF_q)|\leq \epsilon\cdot |V_b(\FF_q)|.
\end{equation}
Fix $b$ such that \eqref{eq:bounding-bad-set} holds. Note that $D|_{\pi(D)=b}=f(U_{V_b(\FF_q)})$.
By \eqref{eq:bounding-bad-set} and Lemma~\ref{lem:dist-distance}, the distributions $f(U_{V_b(\FF_q)})$ and $f(U_{V_b'(\FF_q)})$ are $\epsilon$-close.
So it suffices to verify that $f(U_{V_b'(\FF_q)})$ is an $(n,1,d)$ algebraic source over $\FF_q$.

The set $V_b'(\FF_q)$ is nonempty as $|V_b'(\FF_q)|\geq (1-\epsilon)|V_b(\FF_q)|>0$.
By definition, every irreducible component $Z$ of $V'_b$ is absolutely irreducible of dimension $\dim V-(k-1)$ and satisfies $\dim \overline{\varphi_2(Z)}=1$. 
This also implies that $\dim \overline{f(Z)}\geq 1$ for every irreducible component $Z$ of $V'_b$ as $\varphi_2$ is defined by $f_u$ and hence its output is part of that of $f|_V$.
So the distribution $f(U_{V'_b(\FF_q)})$ satisfies the first two conditions of $(n,1,d)$ algebraic sources in Definition~\ref{defn:algebraic-source}.

We now verify the third condition in Definition~\ref{defn:algebraic-source}.
As in the proof of Lemma~\ref{lem:conditional-dist},
by Gaussian elimination, we can find integers  $1\leq j_1< \dots<  j_t\leq n$, where $0\leq t\leq k-1$, and polynomials $g_1,\dots,g_t\in\FF_q[X_1,\dots,X_r]$ such that $V_b=V\cap V(g_1,\dots,g_t)$,
and each $g_i$ can be written as a linear combination 
\[
g_i=c_{i, j_i} h_{j_i}+ c_{i,j_i+1} h_{j_i+1}+\dots+c_{i,s} h_{s} +c_i
\]
with $c_{i,j}, c_i\in\FF_q$ and $c_{i,j_i}\neq 0$.
B\'ezout's inequality (Lemma~\ref{lem:bezout}) then gives
\begin{equation}\label{eq:degbound-fiber}
\deg V_b \leq \deg V\cdot \prod_{i=1}^t \deg g_i = \deg V\cdot \prod_{i=1}^t \deg h_{j_i}.
\end{equation}
Let $\{\widehat{j}_1,\dots,\widehat{j}_{s-t}\}=[s]\setminus\{j_1,\dots,j_t\}$, where $\widehat{j}_1<\dots < \widehat{j}_{s-t}$.
As $g_1,\dots,g_t$ vanish identically on $V_b$, adding to each $f_i$ a multiple of $g_j$ for $j\in [t]$ does not change $f_i|_{V_b}$.
In particular, for $i\in [n]$, 
we can eliminate the dependence of $f_i$ on $h_{j_1},\dots,h_{j_t}$ and find $\widetilde{f}_i\in \mathcal{L}_{h_{\widehat{j}_1},\dots,h_{\widehat{j}_{s-t}},\FF_q}$ such that $\widetilde{f}_i|_{V_b}=f_i|_{V_b}$.
Then the morphism $f|_{V_b}: V_b\to\AA^{n}_{\FF_q}$ is defined by the polynomials $\widetilde{f}_{1},\dots,\widetilde{f}_{n}$.
And
\[
\deg V'_b\cdot \deg h_{\widehat{j}_1} \leq 
\deg V_b\cdot  \deg h_{\widehat{j}_1} 
\stackrel{\eqref{eq:degbound-fiber}}{\leq}
\deg V\cdot \left(\prod_{i=1}^t \deg h_{j_i}\right) \cdot \deg h_{\widehat{j}_1}\leq \deg V\cdot \prod_{i=1}^k \deg h_i\leq d.
\]
So the third condition in Definition~\ref{defn:algebraic-source} is also satisfied (with respect to the polynomials $\widetilde{f}_i$).
This shows that $f(U_{V'_b(\FF_q)})$ is an $(n,1,d)$ algebraic source over $\FF_q$, as desired.
\end{proof}

The following theorem shows how to compose all the ingredients in our construction: an extractor $\Ext_1$ for $(n,1,d)$ algebraic sources, an extractor $\Ext_2$ for full-rank algebraic sources, and a linear seeded rank extractor $\varphi$, in order to obtain extractors for $(n,k,d)$ algebraic sources. The construction uses $\Ext_1$ in order to select the seed for $\varphi$, applies $\varphi$ on the input, and then applies $\Ext_2$ on the resulting ``condensed'' source.

\begin{theorem}[Composition of extractors]\label{thm:composition}
Let $n\geq k>1$ be integers. Let $\epsilon, \epsilon'\in (0,1)$.
Suppose we are given the following objects:
\begin{itemize}
    \item an $\epsilon$-extractor $\Ext_1: \FF_q^n\to \{0,1\}^{m_1}$ for $(n,1,d)$ algebraic sources over $\FF_q$, 
    \item an $\epsilon$-extractor $\Ext_2: \FF_q^{k-1}\to \{0,1\}^{m_2}$ for $(k-1,k-1,d)$ algebraic sources over $\FF_q$, and
    \item an $(n,k-1,k,\epsilon')$ linear seeded rank extractor $(\varphi_y)_{y\in \{0,1\}^{\ell}}$ for varieties over $\overline{\FF}_q$ (see Definition~\ref{defn:seeded-rank-extractor}) such that $\ell\leq m_1$ and each $\varphi_y$ is defined by linear polynomials over $\FF_q$.
\end{itemize}
Write $\Ext_1=(\Ext_1', \Ext_1'')$, where $\Ext_1'$ and $\Ext_1''$ output the first $\ell$ bits and the last $m_1-\ell$ bits of $\Ext_1$ respectively.
Assume $q\geq \max\{20d^5, 2(k+1)d^2/\epsilon^2\}$.
Then the map $\Ext:\FF_q^n\to \{0,1\}^{m_1}\times\{0,1\}^{m_2}=\{0,1\}^{m_1+m_2}$ defined
by 
\[
\Ext(x):=(\Ext_1(x), \Ext_2(\varphi_{\Ext_1'(x)}(x)))
\]
is a $(6\epsilon\cdot 2^\ell+4\epsilon+\epsilon')$-extractor for $(n,k,d)$ algebraic sources over $\FF_q$.
\end{theorem}

Towards the proof of Theorem \ref{thm:composition}, we start with the following lemma.

\begin{lemma}\label{lem:distribution-closeness}
Use the notations and assumptions in Theorem~\ref{thm:composition}.
Let $D$ be an irreducible $(n,k,d)$ algebraic source over $\FF_q$ arising from an affine variety $V\subseteq\AA^r_{\FF_q}$ and a polynomial map $f: \AA^r_{\FF_q}\to \AA^n_{\FF_q}$  as in Definition~\ref{defn:algebraic-source}.
Let $y\in\supp(\Ext'_1(D))$. Assume $\dim \overline{\varphi_y(f(V))}=k-1$.
Then 
\[
(\Ext_1''(D), \Ext_2(\varphi_{y}(D)))|_{\Ext_1'(D)=y} =_{3\epsilon\cdot 2^\ell+\epsilon} U_{m_1+m_2-\ell}.
\]
\end{lemma}

\begin{proof}
By Lemma~\ref{lem:conditioning}, with probability at least $1-\epsilon$ over $b\sim \varphi_y(D)$, the distribution $D|_{\varphi_y(D)=b}$ is $\epsilon$-close to an $(n,1,d)$ algebraic source $D'_b$ over $\FF_q$. Fix $b\in \supp(\varphi_y(D))$ such that this happens, i.e., 
\begin{equation}\label{eq:distribution-eq1}
D|_{\varphi_y(D)=b}=_{\epsilon} D'_b.
\end{equation}
As $\Ext_1$ is an $\epsilon$-extractor for $(n,1,d)$ algebraic sources over $\FF_q$,
we have 
\begin{equation}\label{eq:distribution-eq2}
\Ext_1(D_b')=_{\epsilon} U_{m_1}.
\end{equation}
Combining
\eqref{eq:distribution-eq1} and \eqref{eq:distribution-eq2}, we conclude that with probability at least $1-\epsilon$ over $b\sim \varphi_y(D)$,
\[
\Ext_1(D)|_{\varphi_y(D)=b}=_{2\epsilon} U_{m_1}.
\]
By Lemma~\ref{lem:approx-2}, this implies
\[
(\Ext_1'(D), \Ext_1''(D), \varphi_y(D))=(\Ext_1(D), \varphi_y(D))=_{3\epsilon} U_{m_1}\times\varphi_y(D).
\]
Therefore,
\[
(\Ext_1'(D), \Ext_1''(D), \Ext_2(\varphi_y(D)))=_{3\epsilon} U_{m_1}\times \Ext_2(\varphi_y(D))=U_{\ell}\times U_{m_1-\ell}\times \Ext_2(\varphi_y(D)).
\]
By Lemma~\ref{lem:close-to-dist},
\begin{equation}\label{eq:distribution-eq3}
(\Ext_1''(D), \Ext_2(\varphi_y(D))|_{\Ext_1'(D)=y}=_{6\epsilon\cdot 2^\ell} U_{m_1-\ell}\times \Ext_2(\varphi_y(D)).
\end{equation}
By assumption, $V$ is irreducible and $\dim \overline{\varphi_y(f(V))}\geq k-1$. 
As $\varphi_y$ is a linear map over $\FF_q$ and the polynomials that define $f$ are in $\mathcal{L}_{h_1,\dots,h_s,\FF_q}$, which is closed under taking linear combinations over $\FF_q$, we see that $\varphi_y(D)=(\varphi_y\circ f)(U_{V(\FF_q)})$ is a $(k-1,k-1,d)$ algebraic source over $\FF_q$.
As $\Ext_2$ is an $\epsilon$-extractor for $(k-1,k-1,d)$ algebraic source over $\FF_q$, we have
\begin{equation}\label{eq:distribution-eq4}
\Ext_2(\varphi_y(D))=_{\epsilon} U_{m_2}.
\end{equation}
Combining \eqref{eq:distribution-eq3} and \eqref{eq:distribution-eq4} yields
\[
(\Ext_1''(D), \Ext_2(\varphi_y(D))|_{\Ext_1'(D)=y}=_{6\epsilon\cdot 2^\ell+\epsilon} U_{m_1+m_2-\ell}
\]
as desired.
\end{proof}

\begin{proof}[Proof of Theorem~\ref{thm:composition}]
Let $D$ be an irreducible $(n,k,d)$ algebraic source over $\FF_q$ arising from an affine variety $V\subseteq\AA^r_{\FF_q}$ and a polynomial map $f: \AA^r_{\FF_q}\to \AA^n_{\FF_q}$  as in Definition~\ref{defn:algebraic-source}.
Then $\dim \overline{f(V)}\geq k$.
Let $T$ be the set of $y\in\{0,1\}^\ell$ such that $\dim\overline{\varphi_y(f(V))}\geq k-1$.
As $(\varphi_y)_{y\in \{0,1\}^\ell}$ is an $(n,k-1,k,\epsilon')$ linear seeded rank extractor, we have
\begin{equation}\label{eq:prob-in-T}
\Pr_{y\sim U_\ell}[y\in T]\geq 1-\epsilon'.
\end{equation}

As $D$ is an $(n,k,d)$ algebraic source and hence an $(n,1,d)$ algebraic source over $\FF_q$, and $\Ext_1$ is an $\epsilon$-extractor for $(n,1,d)$ algebraic source over $\FF_q$, we have $\Ext_1(D)=_\epsilon U_{m_1}$.
So  $\Ext_1'(D)=_\epsilon U_{\ell}$. Combining this with \eqref{eq:prob-in-T} yields
\begin{equation}\label{eq:prob-in-T-2}
\Pr_{y\sim \Ext_1'(D)}[y\in T]\geq 1-\epsilon-\epsilon'.
\end{equation}
By Lemma~\ref{lem:distribution-closeness}, we have
\begin{equation}\label{eq:distribution-closeness}
(\Ext_1''(D), \Ext_2(\varphi_{y}(D)))|_{\Ext_1'(D)=y} =_{6\epsilon\cdot 2^\ell+\epsilon} U_{m_1+m_2-\ell} \quad \text{for all}~y\in \supp(\Ext_1'(D))\cap T.
\end{equation}
By Lemma~\ref{lem:approx-2}, \eqref{eq:prob-in-T-2} and \eqref{eq:distribution-closeness} together yield
\[
\Ext(D)=
(\Ext_1'(D), \Ext_1''(D), \Ext_2(\varphi_{\Ext'_1(D)}(D)))=_{6\epsilon\cdot 2^\ell+2\epsilon+\epsilon'} \Ext_1'(D)\times U_{m_1+m_2-\ell}.
\]
Using the fact $\Ext_1'(D)=_\epsilon U_{\ell}$ again, we obtain
\[
\Ext(D)=_{6\epsilon\cdot 2^\ell+3\epsilon+\epsilon'} U_{m_1+m_2}.
\]

The above proof shows that $\Ext$ is a $(6\epsilon\cdot 2^\ell+3\epsilon+\epsilon')$-extractor for irreducible $(n,k,d)$ algebraic sources over $\FF_q$.
By Lemma~\ref{lem:decompose-irreducible}, every $(n,k,d)$ algebraic source over $\FF_q$ is $\epsilon$-close to a convex combination of irreducible $(n,k,d)$ algebraic sources over $\FF_q$.
So $\Ext$ is a $(6\epsilon\cdot 2^\ell+4\epsilon+\epsilon')$-extractor for $(n,k,d)$ algebraic sources over $\FF_q$.
\end{proof}

\paragraph{Putting it together.}

 We now instantiate the objects in Theorem~\ref{thm:composition} and prove Theorem~\ref{thm:intro:extractor-main}.

\begin{proof}[Proof of Theorem~\ref{thm:intro:extractor-main}]
If $k=0$, the theorem holds trivially, and we may even use an extractor with an empty output.
If $k=1$, the theorem follows from Theorem~\ref{thm:single-polynomial-general}.
So assume $k>1$.

Let $\ell=\lceil \log(2n^2/\epsilon)\rceil$.
Construct an explicit linear $(n, k-1, k-1, \epsilon_0)$ seeded rank extractor $(\varphi_y)_{y\in\{0,1\}^\ell}$ for varieties with $\epsilon_0=(k-1)(n-k+1)/2^\ell\leq \epsilon/2$ using the construction in Lemma~\ref{lem_condenser} (see also Corollary~\ref{cor:linear-seeded}). This is possible as $|\FF_q^\times|=q-1\geq \max\{n, 2^\ell\}$.
By definition, $(\varphi_y)_{y\in\{0,1\}^\ell}$ is also a linear $(n, k-1, k, \epsilon_0)$ seeded rank extractor for varieties.

Let $\epsilon_1=\frac{\epsilon/2}{6\cdot 2^\ell+4}$.
Construct an explicit $\epsilon_1$-extractor $\Ext_1:\FF_q^n\to\{0,1\}^{m_1}$ for $(n,1,d)$ algebraic sources over $\FF_q$ using Theorem~\ref{thm:single-polynomial-general} such that 
\[
m_1\geq \log q-2\log\log p-O(\log (nd/\epsilon_1))=\log q-2\log\log p-O(\log (nd/\epsilon)).
\]
We may assume $m_1\geq \ell$ as $q\geq (nd/\epsilon)^c$ where $c>0$ is a sufficiently large constant. Write $\Ext_1=(\Ext_1', \Ext_1'')$, where $\Ext_1'$ and $\Ext_1''$ output the first $\ell$ bits and the last $m_1-\ell$ bits of $\Ext_1$ respectively.

Finally, construct an explicit $\epsilon_1$-extractor $\Ext_2: \FF_q^{k-1}\to \{0,1\}^{m_2}$ for $(k-1,k-1,d)$ algebraic sources over $\FF_q$ using Theorem~\ref{thm:full-rank-extractor} such that
\[
m_2\geq (k-1)\log q-2\log\log p-O(\log(d/\epsilon_1))=(k-1)\log q-2\log\log p-O(\log(nd/\epsilon)).
\]
The choice of $\epsilon_0$ and $\epsilon_1$ guarantees that $6\epsilon_1\cdot 2^\ell+4\epsilon_1+\epsilon_0\leq \epsilon$.
By Theorem~\ref{thm:composition},
the map $\Ext:\FF_q^n\to \{0,1\}^{m_1+m_2}$ defined
by 
$\Ext(x):=(\Ext_1(x), \Ext_2(\varphi_{\Ext_1'(x)}(x)))$
is an $\epsilon$-extractor for $(n,k,d)$ algebraic sources over $\FF_q$, whose output length is
\[
m_1+m_2\geq k\log q-4\log\log p-O(\log(nd/\epsilon))
\]
as desired.
\end{proof}

\begin{remark*}
If $d=1$, an $(n,k,d)$ algebraic source over $\FF_q$ is simply an affine source over $\FF_q$. In this case, our output length in Theorem~\ref{thm:intro:extractor-main} is $k\log q-O(\log\log p+\log(n/\epsilon))$, which is slightly better than the output length $(k-1)\log q$ in \cite{gabizon-raz}, This is due to a more careful analysis that we use. Namely, we use the fact that a linear seeded rank extractor is \emph{strong} in the sense that most seeds are good. This allows us to include the seed in the output, which yields the improved output length.
We remark that the analysis of \cite{gabizon-raz} can be easily modified to achieve such an improvement too. Finally, when $d=1$, an $(k-1,k-1,d)$ algebraic source is already the uniform distribution, as observed in \cite{gabizon-raz}. So one can simply use the identity map $\FF_q^{k-1}\to \FF_q^{k-1}$ as $\Ext_2$ and get a slightly better parameter $m_2=(k-1)\log q$ instead of $m_2=(k-1)\log q-2\log\log p-O(\log(n/\epsilon))$.
\end{remark*}


\section{Affine Extractors with Exponentially Small Error for Quasipolynomially Large Fields}
\label{sec:affine-extractors}

In this section, we construct affine extractors with exponentially small error, over prime fields of size $q = n^{O(\log \log(n))}$ and any characteristic. Our construction is in fact identical to the extractor of Bourgain, Dvir and Leeman \cite{bourgain-dvir-leeman-2016}, but our analysis is slightly improved. Specifically, Bourgain, Dvir and Leeman constructed an affine extractor over prime fields $\FF_q$ where $q = n^{O(\log \log n)}$ is a so-called ``typical'' prime. Our construction works over any prime finite field of the same size.

Recall that ``$\log$'' denotes logarithms in base $2$. We use ``$\ln$'' in this section to denote natural logarithms.

\begin{definition}[Divisor counting function]
For positive integer $n \geq 1$, let $\omega(n)$ count the number of distinct prime factors of $n$, not counting multiplicity. 
\end{definition}

Bourgain, Dvir and Leeman use average-case bounds on $\omega(q-1)$ for a ``typical'' prime $q$. For our purposes we need the following worst-case upper bound on $\omega(n)$. 

\begin{lemma}[\cite{robin-1983-estimate}, Theorem 13]\label{divisor-bound}
Let $n \geq 26$ be an integer. Then
$$\omega(n) \leq \frac{\ln n}{\ln \ln n - 1.1714}$$
\end{lemma}

The following proposition replaces the use of \cite{bourgain-dvir-leeman-2016} by finding a set of degrees $d_1, \ldots, d_n$ with useful properties for the construction.

\begin{proposition}
\label{prop:good-degrees}
Let $q$ be a prime number. Fix $\eps > 0$ . Then, if $q \geq n^{\frac{2}{\eps} \log\log(n)}$, there exists an efficient deterministic algorithm that, in time polynomial in $n$, finds $n$ integers $d_1 < d_2 < \dots < d_n \in \NN$ such that  $\text{LCM}(d_1, \dots, d_n) \leq q^{\eps}$ and each $d_i$ is coprime to $q-1$.
\end{proposition}

Notice that these properties of $d_1, \dots, d_n$ are precisely those needed for the affine extractor for \cite{bourgain-dvir-leeman-2016}. 

\begin{proof}
Suppose that $q = n^{C \log\log(n)}$ for a constant $C > 0$ to be specified later. Let $r = \ceil*{\log n}$. Let $p_1, \dots, p_r$ be the least $r$ primes that are coprime to $q-1$. Then $p_1, \dots, p_r$ all belong to the first $\omega(q-1) + r$ primes. 

By Lemma \ref{divisor-bound}, we know that for large enough $n$,
\[
\omega(q-1) \le \frac{\ln q}{\ln \ln q - 1.1714}
\le 1.01 \frac{\log q}{\log \log q}.
\]
(Recall that we use $\log$ for logarithm in base 2).
Therefore, $p_1, \dots, p_r$ are among the first $m$ primes, where 
\begin{align*}
m & \leq \ceil*{\log n} + 1.01 \cdot \frac{\log(q)}{\log\log(q)} \\
& \leq 1 + \log(n) + \frac{1.01C \log(n) \cdot \log\log(n)}{\log(C) + \log\log(n) + \log\log\log(n)} \\
&\leq 1 + \log(n) + 1.01C\log(n) \leq 2C \log(n)
\end{align*}

The last two inequalities are for large enough $n$ and $C$.

Next, we would like to bound the magnitude of these $m$ primes.
By the Chebyshev bound on the prime counting function (see, e.g., Theorem 5.4 in \cite{shoup-number-theory}), $[t]$ contains at least $\frac{t \ln(2)}{2\ln(t)}$ primes. Therefore, taking $t = C' \log(n) \cdot \log\log(n)$, where $C' = 10C$, we get that for large enough $n$, $[t]$ contains at least $m$ primes.

Let $D = p_{1} p_{2} \cdots p_{r}$. Notice that $D$ has $2^{\lceil \log n \rceil} \ge n$ distinct divisors, each of which is coprime to $q-1$. These distinct divisors  are the $d_1, \dots, d_n$ in the theorem statement. 
In particular, we have $\text{LCM}(d_1,\dots,d_n)\leq D$.

We can upper-bound $D$ as $p_r^r$. 
We want to choose $q$ large enough that $D \leq q^{\eps}$.
Since $p_r \leq t$ and $r \leq \log(n) + 1$, 
we obtain 
\begin{align*}
D  & \leq (C^\prime \log(n) \cdot \log\log(n))^{\log(n) + 1} 
\\
& \le n^{\log(C^\prime \log(n) \cdot \log\log(n)) + 1} \\
&= n^{\log\log(n) + \log\log\log(n) + \log(C^\prime) + 1}
\end{align*}

Moreover, $q^\eps = n^{C\eps \cdot \log\log(n)}$. For large enough $n$, it follows that choosing $C \geq \frac{2}{\eps}$ is enough to imply that $D \leq q^\eps$.

Finally, observe that indeed $d_1, \ldots, d_n$ can be found in time which is polynomial in $n$, by checking all integers up to $t = O(\log n \log \log n)$ for primality and co-primality with $q-1$ in order to compute $p_1, \ldots, p_r$, and then multiplying all non-empty subsets of $p_1, \ldots, p_r$ to output $d_1, \ldots, d_n$
\end{proof}

The rest of the proof continues in a very similar manner to the proof of Theorem 3.1 of \cite{bourgain-dvir-leeman-2016}. For completeness, we provide the main details of the construction and its proof. The following lemma from \cite{bourgain-dvir-leeman-2016} gives a convenient form for representing affine subspaces.

\begin{lemma}[\cite{bourgain-dvir-leeman-2016}, Lemma 3.4]
\label{lem:echelon-form}
Let $V \subseteq \FF^n$ be an affine subspace of dimension $k$. Then, there is an affine map $\ell : \FF^k \to \FF^n$ whose image is $V$ such that there exist $k$ indices $1 \le j_1 < j_2 < \dots < j_k \le n$, such that
\begin{enumerate}
    \item For all $i < j_1$, $\ell_i(t)$ is a constant in $\FF_q$.
    \item For all $i \in [k]$, $\ell_{j_i}(t) = t_i$.
    \item For every $i>1$ and $j < j_i$, $\ell_j(t)$ is an affine function which depends only on $t_1, \ldots, t_{i-1}$.
\end{enumerate}
\end{lemma}

We also need the following exponential sum estimate due to Deligne (see \cite{deligne1974conjecture, moreno-kumar-1993, bourgain-dvir-leeman-2016}). For $b \in \FF_q^n$, define $\chi_b : \FF_q^n \to \CC$ to be the additive character $\chi_b(x) = \omega_q^{\langle b,x \rangle}$ where $\omega_q = e^{2\pi i / q}$ is a primitive $q$-th root of unity. Note that this definition is valid even for $n=1$.

\begin{theorem}\label{thm:deligne-exp-sum}
Let $f$ be an $n$-variate polynomial over $\FF_q$ of degree $d$. Let $f_d$ denote the degree-$d$ homogeneous component of $f$ and suppose that $f_d$ is \emph{smooth}, that is, the only common zero of the $n$ polynomials $\set{\partial f_d / \partial x_i}_{i \in [n]}$ is the all-zero vector. Then for every non-zero $b \in \FF_q$,
\[
\left| \sum_{x \in \FF_q^n} \chi_b(f(x)) \right| \le (d-1)^n q^{n/2}.
\]
\end{theorem}

Let $A \in \FF^{m \times n}$ be a matrix where every $m$ columns are linearly independent (e.g., a Vandermonde matrix). Let $d_1 < d_2 < \dots < d_n$ be as in Proposition \ref{prop:good-degrees} and define the function $E: \FF^n \to \FF^m$ by
\begin{equation}\label{eq:affine-extractor}
E(x_1, \ldots, x_n) = A \cdot
\begin{pmatrix}
x_1^{d_1} \\
\vdots \\
x_n^{d_n}
\end{pmatrix}.
\end{equation}

\begin{theorem}
\label{thm:affine-extractor}
For every $0 < \beta < 1/2$, there exists a constant $C$ such that the following holds: Let $k \le n$ be integers and $\FF$ be a prime field of size $q \ge n^{C \log \log n}$. Then for $m=\beta k$ the function $E : \FF^n \to \FF^m$ as in \eqref{eq:affine-extractor} is an affine extractor for min-entropy $k$ with error $q^{-\Omega(k)}$. That is, for every affine subspace $V \subseteq \FF^n$ of dimension $k$, if $X_V$ is a random variable uniformly distributed on $V$, $E(X_V)$ is $q^{-\Omega(k)}$-close to uniform on $\FF^k$.    
\end{theorem}

\begin{proof}
As we mentioned earlier, our proof is identical to the proof of Theorem 3.1 of \cite{bourgain-dvir-leeman-2016}. We provide the details for completeness.

Let $V \subseteq \FF^n$ be an affine subspace of dimension $k$ and let $X$ be a random variable obtained by picking a random element from $V$ and applying the extractor $E$ above. Let $\ell$ be an affine map whose image is $V$ as in Lemma \ref{lem:echelon-form}.

By Lemma \ref{lem:xor}, it is enough to give an upper bound on $\left| \Ex[\chi_c(X)] \right|$ for every non-zero $c \in \FF_q^n$.

Denote $b=c^T A$, and observe that 
\[
\chi_c(E(x)) = \omega_q^{b_1 x_1^{d_1} + \cdots + b_n x_n^{d_n}} = \chi_1(b_1 x_1^{d_1} + \cdots + b_nq x_n^{d_n}).
\]
Hence,
\begin{equation}\label{eq:affine-character-exponential-sum}
|\Ex[\chi_c(X)]| = 
 \left| q^{-k} \sum_{t \in \FF_q^k} \chi_1(b_1 \ell_1(t)^{d_1} + \cdots + b_n \ell_n(t)^{d_n}) \right|. 
\end{equation}

Denote $D = \text{LCM}(d_1, \ldots, d_n) \le q^{\varepsilon}$ and $D_i = D/d_{j_i}$ for every $i \in [k]$. By performing the change of variables $t_i = s_i^{D_i}$ (which is invertible since $d_1, \ldots, d_n$ are all coprime to $q-1$), we define $\tilde{\ell}_j = \ell_j(s_1^{D_1}, \ldots, s_k^{D_k})$, so that
\eqref{eq:affine-character-exponential-sum} becomes
\begin{equation}\label{eq:affine-exp-sum-after-change}
|\Ex[\chi_c(X)]| = \left|
q^{-k}  \sum_{s \in \FF_q^k} \chi_1(b_1 \tilde{\ell}_1(s)^{d_1} + \cdots + b_n \tilde{\ell}_n(s)^{d_n})
\right|,
\end{equation}
and the functions $\tilde{\ell}_1, \ldots, \tilde{\ell}_k$ have the following properties (as in Claim 3.5 in \cite{bourgain-2007-affine-extractor}):

\begin{enumerate}
    \item For every $i \in [k]$, $\tilde{\ell}^{d_{j_i}}_{j_i} = s_i^D$.
    \item For all $j \not\in \set{j_1, \ldots, j_k}$, $\tilde{\ell}^{d_j}_j$ is a polynomial in $s_1, \ldots, s_k$ of degree strictly less than $D$.
\end{enumerate}

This implies that we can write \eqref{eq:affine-exp-sum-after-change} as
\begin{equation}\label{eq:affine-exp-sum-low-deg-poly}
|\Ex[\chi_c(X)]| = \left|
q^{-k}  \sum_{s \in \FF_q^k} \chi_1(b_{j_1} s_1^D + \cdots + b_{j_k}s_k^D + g(s))
\right|,
\end{equation}
where $g$ is a polynomial of degree strictly less than $D$.

Since $c$ is non-zero and every $m$ columns of $A$ are linearly independent, the vector $b = c^T A$ has at most $m-1 < k/2$ zero coordinates. Hence, at least $k/2$ of the values $b_{j_1}, \ldots, b_{j_k}$ are non-zero. Suppose without loss of generality that these are the first $k/2$ coordinates. Thus, we estimate \eqref{eq:affine-exp-sum-low-deg-poly} as
\begin{equation*}\label{eq:affine-exp-sum-smooth}
\begin{aligned}
&|\Ex[\chi_c(X)]|\\ 
&\le
q^{-k/2} \sum_{s_{k/2+1}, \ldots, s_{k} \in \FF_q}
\left|
q^{-k/2}  \sum_{s_1, \ldots, s_{k/2} \in \FF_q} \chi_1(b_{j_1} s_1^D + \cdots + b_{j_{k/2}}s_{k/2}^D + g_{s_{k/2}+1, \ldots, s_k}(s_1, \ldots, s_{k/2}))
\right|
\end{aligned}
\end{equation*}
with $b_{j_1}, \ldots, b_{j_{k/2}}$ non-zero and $\deg(g_{s_{k/2+1}, \ldots, s_k}) < d$ for every $s_{k/2+1}, \ldots, s_k$. That is, for every choice of $s_{k/2+1}, \ldots, s_k$, the degree $D$ homogeneous component of the polynomial
\[
b_{j_1} s_1^D + \cdots + b_{j_{k/2}}s_{k/2}^D + g_{s_{k/2}+1, \ldots, s_k}(s_1, \ldots, s_{k/2})
\]
is smooth. By Theorem \ref{thm:deligne-exp-sum},

\[
|\Ex[\chi_c(X)]| \le q^{-k/2} \cdot D^{k/2} \cdot q^{k/4} \le q^{(-1/4+\varepsilon/2)k}.
\]
Setting $\varepsilon = 1/4 - \beta/2 > 0 $ and applying Lemma \ref{lem:xor}, the statistical distance of $X$ from the uniform distribution on $\FF_q^m$ is at most
\[
q^{(-1/4+\varepsilon/2)k} \cdot q^{m/2} \le q^{-(\varepsilon/2) k},
\]
which concludes the proof.
\end{proof}


\section{Explicit Noether Normalization for Affine Varieties and Affine Algebras}
\label{sec_rank_variety}

The Noether normalization lemma \cite{Noe26, Nag62} is a cornerstone of commutative algebra and algebraic geometry. It states that any finitely generated commutative algebra over a field $\FF$, or what we call an \emph{affine algebra} over $\FF$, is not too far from a polynomial ring, in the sense that it is always a finitely generated module over a subring that is isomorphic to a polynomial ring $\FF[Y_1,\dots, Y_k]$.
The geometric interpretation of this statement is that any affine variety $V$ over $\FF$ is a ``branched covering'' of an affine space $\AA^k_{\FF}$, or more precisely, $V$ admits a surjective finite morphism $\varphi_V: V\to \AA^k_{\FF}$. 

When $\FF$ is an infinite field (or more generally, a sufficiently large field), the polynomials that define the finite morphism $\varphi_V$ may be chosen to be linear polynomials (see, e.g., Lemma~\ref{lem:finite-linear-map}). In general, $\varphi_V$ can always be chosen to be defined by polynomials of sufficiently large degrees. In fact, counting arguments show that given the variety,  a ``random'' polynomial map defined by polynomials of sufficiently large degrees would almost surely yield such a finite morphism. See \cite{BE19} for a quantitative analysis. However, it is not known how to completely ``derandomize'' such counting arguments.

The first  proof of the Noether normalization lemma for general affine algebras over arbitrary fields was given by Nagata \cite{Nag53, Nag56, Nag62}.
This proof has the interesting feature that it actually constructs a ``universal'' polynomial map $\varphi:\AA^n_\FF\to \AA^k_\FF$ that works for all low-degree affine varieties. Namely, for any low-degree affine variety $V\subseteq\AA^n_\FF$ of dimension $k$, the restriction of $\varphi$ to $V$ gives a finite morphism $\varphi|_V: V\to \AA^k_{\FF}$.
The existence of such a polynomial map $\varphi$ that is independent of $V$ appears to be stronger and more intriguing than the existence of finite morphisms $V\to \AA_\FF^k$.
In fact, we do not know how to prove the existence of $\varphi$ via a counting argument.

While the polynomial map $\varphi$ constructed by Nagata gives a uniform way of constructing finite morphisms, a drawback is that the degrees of the polynomials that define $\varphi$ can get extremely high due to the iterative nature of the construction.
More specifically, the map $\varphi$ is constructed as a composition of polynomial maps $\varphi_i: \AA^{i+1}_\FF\to \AA^{i}_\FF$, $i=n-1,\dots,k$ such that their restrictions $\varphi_i|_{V_{i+1}}$ are finite morphisms, where we inductively define $V_n=V$ and $V_i=\overline{\varphi_i(V_{i+1})}$ for $i=n-1,\dots, k$. The problem is that composing with a polynomial map can increase the degree of a variety exponentially (see Lemma~\ref{lem:deg-poly-image}). The degree bound for the polynomials defining $\varphi$ is at least doubly exponential for this reason.

Thus, it is a natural question to ask if there is a more efficient construction of the universal polynomial map $\varphi$.
In this section, we show that the DKL construction in Section~\ref{sec:deterministic-rank-extractor} is indeed such a construction, which always works when $|\FF|\geq n$.

\paragraph{The construction of $\varphi$.}
We first recall the DKL construction in Section~\ref{sec:deterministic-rank-extractor}.
Let $\FF$ be a field.
Let $n, d\in \NN^+$ and $m,k\in [n]$. Let $ d_1,\dots, d_n$ be $n$ pairwise coprime integers greater than $d$.
Let $M=(c_{i,j})_{i\in [m], j\in [n]}\in\FF^{m\times n}$ be a $k$-regular matrix, i.e., any $k$ distinct columns of $M$ are linearly independent.
Let $\varphi=\varphi(M): \AA^n_\FF\to \AA^m_\FF$ be the polynomial map defined by $f_1,\dots,f_m\in\FF[X_1,\dots,X_n]$, where $f_i:=\sum_{j=1}^n c_{i,j} X_j^{d_j}$.
In other words, $\varphi$ is given by
\[
\varphi: (a_1,\dots,a_n)\mapsto \left(\sum_{j=1}^n c_{1,j} a_j^{d_j}, \dots, \sum_{j=1}^n c_{m,j} a_j^{d_j}\right).
\]

The main results of this subsection are the following theorems.

\begin{theorem}[Explicit Noether normalization for affine varieties]\label{thm_finite}
Let $V$ be an affine variety of dimension at most $k$ and degree at most $d$ over a field $\FF$.
Then $\varphi|_V: V\to\AA^m_\FF$ is a finite morphism. 
\end{theorem}

Theorem~\ref{thm_finite} translates into the following algebraic statement, Theorem~\ref{thm_algNNL}, which gives an explicit Noether normalization lemma for affine algebras, i.e., finitely generated commutative algebras over a field. 

Recall that the \emph{Krull dimension} of a commutative ring $A$ is the supremum of the lengths of all chains of prime ideals in $A$. If $V$ is an affine variety over a field $\FF$, then the Krull dimension of its coordinate ring $\FF[V]$ is just the dimension of $V$.

\begin{theorem}[Explicit Noether normalization for affine algebras]\label{thm_algNNL}
Suppose $A$ is a commutative $\FF$-algebra generated by $a_1,\dots,a_n\in A$ such that the Krull dimension of $A$ is at most $k$.
Let the ideal $I$ of $\FF[X_1,\dots,X_n]$ be the ideal of all polynomial relations satisfied by $a_1,\dots,a_n$.
Also suppose the degree of the affine variety $V(I)\subseteq\AA^n$ is at most $d$. Then $A$ is a finitely generated module over its subring $S=\FF[f_1(a),\dots, f_m(a)]$, where $f_1,\dots,f_m$ are the polynomials defining $\varphi$ and $a=(a_1,\dots,a_n)$.
\end{theorem}

The fact that $A$ is a finitely generated module over $S$ implies that the Krull dimension of $S$ equals that of $A$. In the case where the Krull dimension of $A$ is $k$ and $k=m$, this means $f_1(a),\dots,f_m(a)$ are algebraically independent over $\FF$, and hence $S$ is isomorphic to a polynomial ring $\FF[Y_1,\dots,Y_m]$ via $f_i(a)\mapsto Y_i$.

Theorem~\ref{thm_finite} and Theorem~\ref{thm_algNNL} are proved in Appendix~\ref{sec:NNL}. The proof is inspired by and closely follows a geometric proof sketched in \cite[Remark~1]{KRS96}.

\paragraph{Smaller fields.} 
While $k\times n$ MDS matrices are generally not known over small finite fields $\FF_q$, which prevents us from choosing $m=k$ over $\FF_q$, it may still be possible to choose larger $m$ for which (explicit) $k$-regular $m\times n$ matrices over $\FF_q$ exist, and this would yield a finite morphism $\varphi|_V: V\to \AA^m_{\FF_q}$ by Theorem~\ref{thm_finite}. As compositions of finite morphisms are finite \cite[Corollary~5.4]{AM69}, by replacing $n$ with $m$ and $V$ with $V'=\overline{\varphi(V)}$, we reduce the problem of constructing a finite morphism on $V\subseteq \AA^n_{\FF_q}$ to constructing that on $V'\subseteq \AA^m_{\FF_q}$, where $V'$ has the same dimension as $V$ but lives in a possibly much smaller affine space $\AA^m_{\FF_q}$.
The degree of $V'$, however, may be significantly larger than that of $V$. See Lemma~\ref{lem:deg-poly-image} for a general upper bound on the degree.

For example,  while we do not know the existence of $k\times n$ MDS matrices over small finite fields $\FF_q$, one can still use a BCH-code-like construction to obtain an $m\times n$ $k$-regular matrix with $m=O(k\log_q n)$, which can be much smaller than $n$ if $k\ll n$.
Applying the resulting map $\varphi$ reduces the dimension of the ambient space from $n$ to $m$. 

However, when $q$ is really small and $k$ is close to $n$, it may be possible that one can only choose $m=n-1$ and hence only reduce the dimension of the ambient space by one at each step. This is essentially the same method used in Nagata's construction.  Currently, all constructions of the universal polynomial map $\varphi: \AA^n_{\FF_q}\to \AA^k_{\FF_q}$ with $k=\dim V$ that we know over a constant-size field $\FF_q$ use polynomials of degree at least doubly exponential in $\min\{k, n-k\}$ due to the blow-up of the degree of the variety. It is an interesting question to ask if there exist constructions with a better degree bound over constant-size fields.

\begin{appendices}
\section{Noether Normalization via Linear Maps}\label{sec:noether-linear-map}

In this section, we prove the quantitative Noether normalization lemma (Lemma~\ref{lem:finite-linear-map}) and its variant, Lemma~\ref{lem:finite-linear-map-variant}. We also explain how to remove the condition $q > 2(k + 1)d^2$ in  \cite[Theorem~7.1]{CM06} and obtain Theorem~\ref{theorem:lang-weil}.

Although Lemma~\ref{lem:finite-linear-map} concerns only affine varieties, we need to deal with projective varieties that ``complete" affine varieties.

\paragraph{Projective spaces and projective varieties.} 
Fix $\FF$ to be an algebraically closed field.
The \emph{projective $n$-space} $\PP^n$ over $\FF$, as a set, is the quotient set $\left(\FF^{n+1}\setminus\{\mathbf{0}\}\right)/\sim$, where $\mathbf{0}$ is the origin $(0,\dots,0)$ and $\sim$ is the equivalence relation defined by scaling, i.e., $u\sim v$ if $u=cv$ for some $c\in\FF^\times$. 
We use $(n+1)$-tuples $(x_0,\dots,x_{n})$ to represent points in $\PP^n$.  

We equip $\PP^n$ with the \emph{Zariski topology} over $\FF$, where a subset is closed if it is the set of common zeros of a set of homogeneous polynomials in  $\FF[X_0, X_1,\dots, X_n]$. 
Call $X_0,\dots, X_n$ the \emph{homogeneous coordinates} of $\PP^n$.
A closed subset $V\subseteq \PP^n$ is said to be a \emph{projective variety} over $\FF$. 

The projective space $\PP^n$ is covered by open subsets $U_i:=\{(x_0,\dots,x_n)\in \PP^n: x_i\neq 0\}$, $i=0,\dots, n$.
Each $U_i$ can be identified with the affine space $\AA^n$ via the map
\[
(x_0,\dots,x_n)\mapsto \left(\frac{x_0}{x_i},\dots,\frac{x_{i-1}}{x_i},\frac{x_{i+1}}{x_i}, \dots,\frac{x_n}{x_i}\right).
\]
And the subspace topology of $U_i$ induced from that of $\PP^n$ is precisely the Zariski topology of $\AA^n$ if we identify $U_i$ with $\AA^n$ this way.

While the sets $U_i$ are symmetric, we often fix $i=0$ and regard  $\AA^n$ as an open subset of $\PP^n$ by identifying it with $U_0\subseteq\PP^n$ as above. Its complement $\PP^n\setminus \AA^n$ may be identified with $\PP^{n-1}$ (for $n>0$) via the map $(x_0,\dots,x_n)\mapsto (x_1,\dots,x_n)$, and is called the \emph{hyperplane at infinity}.
The \emph{projective closure} of an affine variety $V\subseteq\AA^n$, which we denote by $V_{\mathrm{cl}}$, is the smallest projective subvariety of $\PP^n$ that contains $V$ as a subset. And we have $V_{\mathrm{cl}}\cap \AA^n=V$.

The notions of irreducibility, irreducible components, degree, and dimension all extend to projective varieties. We have $\deg V_{\mathrm{cl}}=\deg V$ and $\dim V_{\mathrm{cl}}=\dim V$ for an affine variety $V\subseteq\AA^n$, i.e., taking the projective closure preserves the degree and the dimension.

The notions of morphisms and finite morphisms also extend to projective varieties. See, e.g., \cite{Sha13}.
If $\varphi: V\to V'$ is a finite morphism  between projective varieties, and $\varphi|_{\varphi^{-1}(U)}: \varphi^{-1}(U)\to U$ is a morphism between affine varieties for some open subset $U$ of $V'$, then $\varphi|_{\varphi^{-1}(U)}$ is also a finite morphism (\cite[\S I.5.3, Theorem~5]{Sha13}). 

The following lemma gives a way of finding a finite morphism $V\to \PP^k$ defined by linear polynomials.

\begin{lemma}[{\cite[\S I.5.3, Theorem~7]{Sha13}}]\label{lem_finitepr}
Let $V\subseteq \PP^n$ be a projective variety over $\FF$.
Suppose $\ell_0,\dots,\ell_k\in\FF[X_0,\dots, X_n]$ are (homogeneous) linear polynomials that have no common zero on $V$. Then they define a finite morphism $V\to \PP^k$ that sends $a\in V$ to $(\ell_0(a),\dots,\ell_k(a))$.
\end{lemma}

Thus, the problem of finding a finite morphism from a projective variety $V$ to $\PP^k$ reduces to finding $k+1$ linear polynomials that have no common zero on $V$. We use \emph{Chow forms} to show the existence of such linear polynomials.

\begin{lemma}
Let $V\subseteq\PP^n$ be an irreducible projective variety over $\FF$ of dimension $k$ and degree $d$. There exists a nonzero polynomial $R_V\in \FF[Y_{0,0},\dots,Y_{k,n}]$ such that 
for all $c_{0,0},\dots,c_{k, n}\in \FF$, the linear polynomials $\ell_0,\dots,\ell_k$ defined by $\ell_i:=\sum_{j=0}^n c_{i,j} X_j$ have no common zero on $V$ iff $R_V(c_{0,0},\dots,c_{k,n})\neq 0$.
Moreover, $R_V$ is multihomogeneous of degree $(d,\dots,d)$ in
\[\{Y_{0,0},\dots,Y_{0,n}\},\dots,\{Y_{k,0},\dots,Y_{k,n}\},
\]
i.e., it is homogeneous of degree $d$ in each of the $k+1$ group of variables.
\end{lemma}

The polynomial $R_V$ is unique up to a scalar and is called the \emph{Chow form} of $V$.
See, e.g., \cite{KPS01}.

We also need the following lemma on the existence of a non-root of a polynomial.

\begin{lemma}\label{lem:nonroot}
Suppose $P\in\FF[X_{1,1},\dots, X_{k,n}]$ is a polynomial such that for all $i\in [k]$, the (total) degree of $P$ in the variables $X_{i,1},\dots,X_{i,n}$ is at most $d$. Let $S$ be a finite subset of $\FF$ of size greater than $d$.
Then $P$ has a non-root in $S^{kn}$. 
\end{lemma}

\begin{proof}
We assign the values $c_{i,1},\dots, c_{i,n}\in S$ to the $k$ groups of variables $\{X_{i,1},\dots, X_{i,n}\}$ for $i=1,\dots, k$ one by one, such that at each step, the polynomial $P$ remain nonzero. 
This is obviously true at the beginning. Now at the beginning of the $i$-th step, assume that $P$ remains nonzero after the partial assignment $X_{i',j}=c_{i',j}$ for $i'<i$ and $j\in [n]$, and call this polynomial $P_i$.
View $P_i$ as a polynomial in the variables $X_{i+1,1},\dots, X_{k,n}$ over the ring $\FF[X_{i,1},\dots,X_{i,n}]$. Then as $P_i\neq 0$, it has a term whose coefficient $Q$ is a nonzero polynomial of degree at most $d$ in $X_{i,1},\dots, X_{i,n}$.
By the DeMillo--Lipton--Schwartz--Zippel lemma \cite{Sch80,Zip79,Dem78},
there exist $c_{i,1},\dots,c_{i,n}\in S$ such that $Q(c_{i,1},\dots,c_{i,n})\neq 0$, and hence $P_i$ remain nonzero after the assignment $X_{i,j}=c_{i,j}$, $j=1,\dots, n$. 
Continuing this process, we see that $P$ has a non-root in $S^{kn}$.
\end{proof}

We are now ready to prove Lemma~\ref{lem:finite-linear-map}. For convenience, we first restate this lemma.

\finitelinearmap*

\begin{proof}
We may assume that $\FF$ is algebraically closed. This is because finiteness of morphisms over $\FF$ follows from that over $\overline{\FF}$ by a descent argument (see Appendix~\ref{sec:NNL}).

Suppose $V_1,\dots,V_s$ are the irreducible components of $V$, where $\dim V_t=k_t\leq k$ and $\deg V_t=d_t$ for $t\in [s]$.
Identify $\AA^n$ with the open subset $U_0=\{(x_0,\dots,x_n)\in \PP^n: x_0\neq 0\}$ of $\PP^n$, and let $H=\PP^n\setminus U_0$ be the hyperplane at infinity.

For $t\in [s]$, let 
\[
P_t=R_{(V_t)_{\mathrm{cl}}}\in \FF[Y_{0,0},\dots, Y_{k_t, n}]\subseteq \FF[Y_{0,0},\dots, Y_{k, n}].
\]
That is, $P_t$ is the Chow form of the projective closure of $V_t$.
Then $P_t$ is multihomogeneous of degree $(d_t,\dots,d_t)$.
For $t\in [s]$, let $\widehat{P}_t\in \FF[Y_{1,1},\dots, Y_{n,n}]$ be the polynomial obtained from $P_t$ by assigning $(Y_{0,0},Y_{0,1}, \dots, Y_{0,n})=(1,0,\dots,0)$ and $Y_{1,0}=\dots=Y_{k,0}=0$.
Then each $\widehat{P}_t$ remains a nonzero polynomial. In fact, $\widehat{P}_t$ may be viewed as the Chow form of $(V_t)_{\mathrm{cl}}\cap H$ as a projective subvariety of $H\cong\PP^{n-1}$ of dimension $k-1$. (If $\dim V_t=0$, then $(V_t)_{\mathrm{cl}}\cap H=\emptyset$ and $\widehat{P}_t$ is a nonzero constant.)

Suppose $c_{1,1},\dots,c_{k,n}\in S$ satisfy $\widehat{P}_t(c_{0,0},\dots,c_{n,k})\neq 0$ for all $t\in [s]$. We claim the linear forms $\ell_i=\sum_{j=1}^n c_{i,j} X_i$ define a finite morphism $V\to \AA^k$.
To see this, note that by the property of Chow forms, the polynomials $\ell_0:=X_0$ and $\ell_1,\dots,\ell_k$ have no common zero on $V_{\mathrm{cl}}=\bigcup_{t=1}^s (V_t)_{\mathrm{cl}}$. So by Lemma~\ref{lem_finitepr}, $\ell_0,\dots, \ell_k$ define a finite morphism $\varphi: V_{\mathrm{cl}}\to \PP^k$.
Let $U_0'$ be the open subset $\{(x_0,\dots,x_k)\in \PP^k: x_0\neq 0\}$ of $\PP^k$. 
As $\ell_0=X_0$, we have $\ell_0(x)\neq 0$ for $x\not\in U_0$ and hence $\varphi(H)\cap U_0'=\emptyset$.
So $\varphi^{-1}(U_0')=U_0\cap V_{\mathrm{cl}}=V$. As $\varphi$ is a finite morphism, so is $\varphi|_{V}: V\to U_0'$.
Finally, note that $\varphi|_{V}$ is exactly the morphism $V\to \AA^k$ defined by $\ell_1,\dots,\ell_k$ if we identify $U_0$ with $\AA^k$.

So it remains to prove the existence of $c_{1,1},\dots,c_{k,n}\in S$ such that $\widehat{P}_t(c_{1,1},\dots,c_{k,n})\neq 0$ for $t\in [s]$.
This follows by applying Lemma~\ref{lem:nonroot} to $P:=\prod_{t=1}^s \widehat{P}_t$  and noting that $P$ is multihomogeneous of degree $(d,\dots,d)$.
\end{proof}

One can prove an analogue of Lemma~\ref{lem:finite-linear-map} for projective varieties with a similar, and in fact, simpler proof. But we only need Lemma~\ref{lem:finite-linear-map} for affine varieties in this paper.

\begin{proof}[Proof of Lemma~\ref{lem:finite-linear-map-variant}]
The proof is the same as that of Lemma~\ref{lem:finite-linear-map}, except that we consider $V_1$ and $V_2$ simultaneously and apply the union bound when picking each linear polynomial $\ell_i$.
\end{proof}

\paragraph{The effective Lang--Weil bound.}

Lemma~\ref{lem:nonroot} above can also be used to prove the effective Lang--Weil bound (Theorem~\ref{theorem:lang-weil}).
We  restate the theorem  below for convenience.

\langweil*

This bound was proved as \cite[Theorem~7.1]{CM06} with an extra condition that $q > 2(k + 1)d^2$.
We first explain why this condition was assumed in \cite{CM06}.

To prove \cite[Theorem~7.1]{CM06}, Cafure and Matera first proved an effective Lang--Weil bound when $V$ is an absolutely irreducible \emph{hypersurface} in $\AA^n_{\FF_q}$ without the condition $q > 2(k + 1)d^2$  (see \cite[Theorem~5.2]{CM06}).
To extend it to the general case, they further argued that if $q>2(k+1)d^2$, then there exists an affine linear map $\pi: \AA^n_{\FF_q}\to \AA^{k+1}_{\FF_q}$ over $\FF_q$, defined by polynomials $\sum_{j=1}^{n} \lambda_{ij} X_j+\gamma_i\in \FF_q[X_1,\dots, X_n]$, $i=1,\dots,k+1$, that induces a \emph{birational equivalence} $\pi|_V: V\dashrightarrow H$ between $V$ and a hypersurface $H\subseteq \AA^{k+1}_{\FF_q}$. Such a birational equivalence may be thought of as an ``almost isomorphism" between $V$ and $H$. In particular, $\pi|_V$ is an ``almost bijection" between the set of rational points of $V$ and that of $H$, i.e., there exist dense open subsets $U\subseteq V$ and $U'\subseteq H$ such that  the rational points in $U$ are mapped bijectively to those in $U'$ under $\pi|_V$. Furthermore, the sets $(V\setminus U)(\FF_q)$ and $(H\setminus U')(\FF_q)$ are small.
The effective Lang--Weil bound for affine varieties then easily follows from the bound for hypersurfaces.

To see intuitively why such a morphism $\pi$ should exist, note that the quantitative Noether normalization lemma (Lemma~\ref{lem:finite-linear-map}) already guarantees the existence of a morphism $V\to \AA^{k}_{\FF_q}$ that is finite (and hence surjective). However, this map is  finite-to-one instead of (almost) one-to-one.
But note that a morphism $\AA^n_{\FF_q}\to \AA^{k+1}_{\FF_q}$ has one extra coordinate in the output, and we can use this extra coordinate to distinguish the finitely many preimages of a general point $b\in \AA^{k}_{\FF_q}$. Thus, we should expect that a general affine linear map $\pi: \AA^n_{\FF_q}\to \AA^{k+1}_{\FF_q}$ defines an ``almost isomorphism" from $V$ to a hypersurface.
Indeed, Cafure and Matera used the \emph{Chow form} and the \emph{discriminant} to prove the existence of such a morphism $\pi$. Specifically, they showed that there exists a nonzero polynomial $G$ in the variables $(\Lambda_{i,j})_{i\in [k+1], j\in [n]}$ and $(\Gamma_i)_{i\in [k+1]}$ over $\overline{\FF}_q$ such that $\pi$ is a desired morphism whenever the coefficients $(\lambda_{i,j})_{i\in [k+1], j\in [n]}$ and $(\gamma_i)_{i\in [k+1]}$ of the polynomials that define $\pi$ form a non-root of $G$ \cite[Theorem~6.1]{CM06}.
Moreover, $G$ has the property that its degree in $\Lambda_{i,1},\dots,\Lambda_{i,n},\Gamma_i$ is at most $2d^2$ for $i\in [k+1]$ (see the proof of \cite[Theorem~6.1]{CM06}). In particular, its total degree is at most $2(k+1)d^2$.

Next, Cafure and Matera used essentially the DeMillo--Lipton--Schwartz--Zippel lemma \cite{Sch80,Zip79,Dem78} and the fact that $\deg G\leq 2(k+1)d^2$ to argue that a non-root of $G$ exists in $\FF_q^{(k+1)(n+1)}$, and this is where they need the condition $q>2(k+1)d^2$ \cite[Corollary~6.2]{CM06}.

However,  we can take advantage of the fact that the degree of $G$ in each group of variables $\Lambda_{i,1},\dots,\Lambda_{i,n},\Gamma_i$ is at most $2d^2$ for $i\in [k+1]$, and use Lemma~\ref{lem:nonroot} instead. This immediately allows us to relax the condition  $q>2(k+1)d^2$ to $q>2d^2$.

Finally, observe that when $q\leq 2d^2$, the bound $|V(\FF_q)-q^k|<(d-1)(d-2)q^{k-1/2}+5d^{13/3}q^{k-1}$ follows from the trivial lower bound $|V(\FF_q)|\geq 0$ and the upper bound $|V(\FF_q)|\leq d q^k$ in Lemma~\ref{lem:elementary-bound}.
So we can remove the condition about $q$ completely.

\section{The Effective Fiber Dimension Theorem}\label{sec:proof-subsec-1}

We prove the general form of the effective fiber dimension theorem (Theorem~\ref{thm:effective-FDT-general}) in this section.
The base field $\FF$ is assumed to be an algebraically closed field.

\paragraph{Generalized Perron Theorem.}

We need the following result proved by Jelonek \cite{jelonek-2005-effective-nullstellensatz}, generalizing a classical result of Perron \cite{Per27} that bounds the degree of annihilating polynomials for algebraic dependent polynomials.

\begin{theorem}[Generalized Perron Theorem {\cite[Theorem~3.3]{jelonek-2005-effective-nullstellensatz}}]\label{thm:generalized-Perron}
Suppose $f: \AA^n\to\AA^m$ is a polynomial map defined by polynomials $f_1,\dots,f_m\in\FF[X_1,\dots,X_n]$, where $\deg f_i=d_i>0$. 
Let $V\subseteq \AA^n$ be an irreducible affine variety over $\FF$, and let $W=\overline{f(V)}\subseteq\AA^m$.
Suppose $\dim V=\dim W=m-1$.
Then there exists a nonzero polynomial $Q\in\FF[Y_1,\dots,Y_m]$  
such that $Q(f_1,\dots,f_m)$ vanishes identically on $V$
and $\deg(Q(Y_1^{d_1},\dots,Y_m^{d_m}))\leq \deg V\cdot \prod_{i=1}^m d_i$.
\end{theorem}

\paragraph{Proof of the effective fiber dimension theorem.}

Now we are ready to prove Theorem~\ref{thm:effective-FDT-general}. For convenience, we first restate the theorem.

\effectiveFDTgeneral*

One special form of the effective fiber dimension theorem with $V=W=\AA^k$ was essentially proved in \cite{guo-saxena-sinhababu} using Perron's bound \cite{Per27}. Our proof of Theorem~\ref{thm:effective-FDT-general} can be seen as a generalization of this proof. It can also be seen as an effective version of a standard proof of the fiber dimension theorem (see \cite[Proof of Theorem~12.4.1]{Vak17}).

Towards proving Theorem~\ref{thm:effective-FDT-general}, we first prove the following lemma using the Generalized Perron Theorem.

\begin{lemma}\label{lem:effective-FDT-special}
Use the notations in Theorem~\ref{thm:effective-FDT-general} and assume that $\deg h_i>0$ for $i\in [s]$. Suppose $\dim V=\dim \overline{f(V)}=k=m$. 
Also suppose at least $u$ polynomials among $f_1,\dots,f_k$ are linear.
Let $g\in\FF[X_1,\dots,X_n]$ such that $\deg g>0$.
Then there exists a nonzero polynomial $Q\in\FF[Y_1,\dots,Y_{k+1}]$ satisfying the following:
\begin{enumerate}
\item $Q(f_1,\dots,f_k,g)\in\FF[X_1,\dots,X_n]$ vanishes identically on $V$.
\item View $Q$ as a univariate polynomial in $Y_{k+1}$ over $\FF[Y_1,\dots,Y_k]$ and let $Q^*\in \FF[Y_1,\dots,Y_k]$ be its leading coefficient. 
Then the degree of $Q^*(f_1,\dots,f_k)\in\FF[X_1,\dots,X_n]$ is at most $\deg V\cdot \deg g\cdot \prod_{i=1}^{\min\{k-u,s\}} \deg h_{i}$.
\end{enumerate}
\end{lemma}

\begin{proof}
Note that changing the coordinate system of $\AA^m=\AA^k$ via an invertible linear transformation does not affect the statement. So by permuting the polynomials $f_i$, we may assume the linear polynomials $f_i$ appear at the end of the list $f_1,\dots,f_k$. By applying Gaussian elimination, we may further assume $f_i\in \mathcal{L}_{h_i,\dots,h_s,\FF}$ (i.e., $f_i$ does not depend on $h_1,\dots,h_{i-1}$) for $i=1,\dots,\min\{k-u,s\}$, and $\deg(f_i)\leq 1$ for $i=\min\{k-u,s\}+1,\dots,k$.
In particular, we have $\deg f_i \leq \deg h_i$ for $i=1,\dots,\min\{k-u,s\}$.

Let $\psi$ be the polynomial map $\AA^n\to\AA^{k+1}$  defined by $f_1,\dots,f_k, g$.
The dimension of $\overline{\psi(V)}$ is $k$ since it cannot exceed $k=\dim V$ and the dimension of $\overline{f(V)}$ is already $k$.
As $\dim \overline{f(V)}=k$, we necessarily have $\deg f_i>0$ for $i\in [k]$.
Applying the Generalized Perron Theorem (Theorem~\ref{thm:generalized-Perron}) to $\psi$, we see that there exists a nonzero polynomial $Q\in \FF[Y_1,\dots,Y_{k+1}]$  such that $Q(f_1,\dots,f_k, g)$ vanishes identically on $V$ and 
\[
\deg\left(Q\left(Y_1^{\deg f_1},\dots,Y_k^{\deg f_k}, Y_{k+1}^{\deg g}\right)\right)\leq \deg V\cdot \deg g\cdot \prod_{i=1}^{k} \deg f_i\leq \deg V\cdot \deg g\cdot \prod_{i=1}^{\min\{k-u,s\}} \deg h_{i}.
\]

Let $Q^*\in \FF[Y_1,\dots,Y_k]$ be the leading term of $Q$ in $Y_{k+1}$. 
Then
\begin{align*}
\deg(Q^*(f_1,\dots,f_k))&\leq  \deg\left(Q^*\left(Y_1^{\deg f_1},\dots,Y_k^{\deg f_k}\right)\right)\\ &\leq \deg\left(Q\left(Y_1^{\deg f_1},\dots,Y_k^{\deg f_k}, Y_{k+1}^{\deg g}\right)\right)\\
&\leq \deg V\cdot \deg g\cdot \prod_{i=1}^{\min\{k-u,s\}} \deg h_{i}.
\end{align*}
To see that the first inequality holds, note that the monomials of $Q^*\left(Y_1^{\deg f_1},\dots,Y_k^{\deg f_k}\right)$ correspond one-to-one to the monomials of $Q^*(Y_1,\dots,Y_k)$, and substituting $f_i$ for $Y_i^{\deg f_i}$ within each monomial does not increase its degree.
\end{proof}

We use Lemma~\ref{lem:effective-FDT-special} to prove Theorem~\ref{thm:effective-FDT-general}.

\begin{proof}[Proof of Theorem~\ref{thm:effective-FDT-general}]
We may assume $\deg h_i>0$ for $i\in [s]$ by removing all polynomials $h_i$ that are constants.
Let $\psi: \AA^{n}\to \AA^{k}$
be the polynomial map defined by  $\ell_{1},\dots,\ell_{k-k'}, f_{j_1},\dots,f_{j_{k'}}$.
As $\tau$ is finite, we have $\dim \overline{\psi(V)}=\dim V=k$.
And we have $\dim \overline{\varphi(V)}=t+k'$ as otherwise the dimension of $\overline{\psi(V)}$ cannot achieve $k$.

Consider $i\in \{k-k'+1,\dots,k\}$. Applying Lemma~\ref{lem:effective-FDT-special} to  $\psi:\AA^n\to \AA^k$ (which is defined by $\ell_{1},\dots,\ell_{k-k'},f_{j_1},\dots,f_{j_{k'}}\in \mathcal{L}_{h_1,\dots,h_s,\ell_{1},\dots,\ell_{k-k'},\FF}$) with $g=\ell_i$ and $u=k-k'$,
we see that there exists a nonzero polynomial $Q_i\in\FF[Y_1,\dots,Y_{k+1}]$ satisfying the following:
\begin{enumerate}
\item $Q_i(\ell_{1},\dots,\ell_{k-k'},f_{j_1},\dots,f_{j_{k'}}, \ell_i)$ vanishes identically on $V$.
\item Let $Q_i^*\in \FF[Y_1,\dots, Y_k]$ be the leading coefficient of $Q_i$ in $Y_{k+1}$. Then the degree of $Q^*_i(\ell_{1},\dots,\ell_{k-k'},f_{j_1},\dots,f_{j_{k'}})\in\FF[X_1,\dots,X_n]$ is at most $\deg V\cdot \prod_{i=1}^{k'} \deg h_{i}$.
\end{enumerate}
For $g\in \FF[X_1,\dots,X_n]$, let $\bar{g}:=g+I(V)\in \FF[V]$.
Then $\bar{\ell}_{1},\dots,\bar{\ell}_{k-k'},\bar{f}_{j_1},\dots,\bar{f}_{j_{k'}}$ are algebraically independent over $\FF$ as $\psi|_V: V\to\AA^k$ is dominant.
In particular, we have 
$Q_i^*(\bar{\ell}_{1},\dots,\bar{\ell}_{k-k'},\bar{f}_{j_1},\dots,\bar{f}_{j_{k'}})\neq 0$, or equivalently, $Q_i^*(\ell_{1},\dots,\ell_{k-k'},f_{j_1},\dots,f_{j_{k'}})$ does not vanish on $V$.

The Zariski closure of the image of $V$ under the polynomial map $\AA^n\to\AA^{k+1}$ defined by $\ell_{1},\dots,\ell_{k-k'},f_{j_1},\dots,f_{j_{k'}}, \ell_i$ is a hypersurface of $\AA^{k+1}$ defined by a single polynomial, which we may assume to be $Q_i$. So $Q_i$ generates the ideal of the polynomial relations satisfied by $\bar{\ell}_{1},\dots,\bar{\ell}_{k-k'},\bar{f}_{j_1},\dots,\bar{f}_{j_{k'}}, \bar{\ell}_i$.

Let $\KK:=\FF(\bar{f}_{j_1},\dots,\bar{f}_{j_{k'}})$, which can be viewed as the function field of the target $\AA^{k'}$ of the dominant morphism $f'$.
By the finiteness of $\tau$, the coordinate ring $\KK[V_{f'}]$ of $V_{f'}$ is a finitely generated module over $\KK[\bar{\ell}_{1},\dots,\bar{\ell}_{k-k'}]$.
So $\bar{\ell}_i$ is integral over $\KK[\bar{\ell}_{1},\dots,\bar{\ell}_{k-k'}]$, i.e., it is a root of a monic polynomial over the ring $\KK[\bar{\ell}_{1},\dots,\bar{\ell}_{k-k'}]$ \cite[Proposition~5.1]{AM69}.
It follows that the leading coefficient $Q^*_i$ of $Q_i$ does not depend on $Y_1,\dots, Y_{k-k'}$, although it may depend on $Y_{k-k'+1},\dots,Y_{k}$ since $\bar{f}_{j_1},\dots,\bar{f}_{j_{k'}}$ are invertible in $\KK$. So $Q_i^*\in \FF[Y_{k-k'+1},\dots,Y_{k}]$.

Let $P=\prod_{i=k-k'+1}^{k} Q_i^*(f_{j_1},\dots,f_{j_{k'}})$. Then $P$ does not vanish identically on $V$ and has degree at most $k'\cdot\deg V\cdot \prod_{i=1}^{k'} \deg h_{i}$.

Consider $a\in V$ such that $P(a)\neq 0$, and let $Z$ be an irreducible component of $\varphi|_V^{-1}(\varphi(a))$. Note that  $\dim Z\geq \dim V - \dim \overline{\varphi(V)} =\dim V-(t+k')=k-k'-t$ by the first claim in the fiber dimension theorem (Theorem~\ref{thm_dimfiber}).

It remains to prove that $\dim Z\leq k-k'-t$. 
For $g\in \FF[X_1,\dots,X_n]$, let $\bar{\bar{g}}:=g+I(Z)\in \FF[Z]$.
Note that $\bar{\bar{\ell}}_i=\ell_i(a)\in\FF$ for $i\in [t]$ since $\ell_i-\ell_i(a)$ vanishes identically on $\varphi|_V^{-1}(\varphi(a))\supseteq Z$.
Similarly, we have $\bar{\bar{f}}_{j_i}=f_{j_i}(a)\in \FF$ for $i\in [k']$ since $f_{j_i}-f_{j_i}(a)$ vanishes identically on $\varphi|_V^{-1}(\varphi(a))\supseteq Z$.

For $i\in \{k-k'+1,\dots,k\}$, as $Q_i(\ell_1,\dots,\ell_{k-k'},f_{j_1},\dots,f_{j_{k'}},\ell_i)$ vanishes identically on $V\supseteq Z$, we see that $\bar{\bar{\ell}}_i$ is a root of the univariate polynomial 
\begin{align*}
&Q_i(\bar{\bar{\ell}}_1,\dots,\bar{\bar{\ell}}_{k-k'},\bar{\bar{f}}_{j_1},\dots,\bar{\bar{f}}_{j_{k'}}, Y_{k+1})\\
=&Q_i(\ell_1(a),\dots,\ell_t(a), \bar{\bar{\ell}}_{t+1}, \dots,\bar{\bar{\ell}}_{k-k'},f_{j_1}(a),\dots,f_{j_{k'}}(a), Y_{k+1})
\in (\FF[\bar{\bar{\ell}}_{t+1}, \dots,\bar{\bar{\ell}}_{k-k'}])[Y_{k+1}].
\end{align*}
Its leading coefficient is 
$Q^*_i(f_{j_1}(a),\dots,f_{j_{k'}}(a))$, which is a nonzero element in $\FF$ since $P(a)\neq 0$.
So $\bar{\bar{\ell}}_i$ is a root of a monic polynomial over $\FF[\bar{\bar{\ell}}_{t+1}, \dots,\bar{\bar{\ell}}_{k-k'}]$ for $i\in \{k-k'+1,\dots,k\}$.
It follows that $\FF(\bar{\bar{\ell}}_1,\dots,\bar{\bar{\ell}}_k)$ is a finite extension of $\FF(\bar{\bar{\ell}}_{t+1}, \dots,\bar{\bar{\ell}}_{k-k'})$.
On the other hand, as the morphism $\pi: V\to \AA^k$ defined by $\ell_1,\dots,\ell_k$ is finite and $Z$ is an affine subvariety of $V$, we know $\FF[Z]$ is a finitely generated module over $\FF[\bar{\bar{\ell}}_1,\dots,\bar{\bar{\ell}}_k]$.
So $\FF(Z)$ is a finite extension of $\FF(\bar{\bar{\ell}}_1,\dots,\bar{\bar{\ell}}_k)$, and hence also a finite extension of $\FF(\bar{\bar{\ell}}_{t+1}, \dots,\bar{\bar{\ell}}_{k-k'})$.
Therefore, the trancendence degree of $\FF(Z)$ over $\FF$ is at most $k-k'-t$, i.e., $\dim Z\leq k-k'-t$.
\end{proof}

\section{Miscellanea}\label{sec:misc}

\paragraph{Absolute irreducibility.}
The following lemma gives alternative characterizations of absolute irreducibility when the finite field $\FF_q$ is large enough.

\begin{lemma}\label{lem:char}
Let $V$ be a nonempty affine variety over $\FF_q$ of dimension $k$ and degree $d$. 
Suppose $q\geq 20d^5$.
Then the following are equivalent:
\begin{enumerate}[(1)]
    \item At least one irreducible component of $V$ of dimension $k$ is absolutely irreducible.
    \item $|V(\FF_q)|\geq q^k/2$.
    \item $|V(\FF_q)|> d^2q^{k-1}$.
\end{enumerate}
\end{lemma}

\begin{proof}
The fact that (1) implies (2) follows from the effective Lang--Weil bound (Theorem~\ref{theorem:lang-weil}).
And (2) implies (3) as $q\geq 20d^5$.
Finally, to see that (3) implies (1), suppose that none of the irreducible components of $V$ of dimension $k$ is absolutely irreducible. Let $V'$ be the union of the irreducible components of $V$ of dimension $k$, and let $V''$ be the union of the remaining irreducible components. Then we have $|V'(\FF_q)|\leq (\deg V')^2 q^{k-1}$ by Lemma~\ref{lem:not-absolute-irreducible} and $|V''(\FF_q)|\leq \deg V'' q^{k-1}$ 
 by Lemma~\ref{lem:elementary-bound}. It follows that 
\[
|V(\FF_q)|\leq |V'(\FF_q)|+|V''(\FF_q)|\leq (\deg V'+\deg V'')^2 q^{k-1}=d^2 q^{k-1}.
\]
So (3) implies (1).
\end{proof}

Lemma \ref{lem:char} shows that the condition of absolute irreducibility in the definition of $(n,k,d)$ algebraic sources (Definition~\ref{defn:algebraic-source}) is useful as it guarantees the existence of enough rational points.
On the other hand, even if none of the irreducible components of $V$ is absolutely irreducible, it may still be possible that the variety $V$ has a substantial number of rational points, so that the randomness extraction question is meaningful.

\begin{example*}
Let $q$ be an odd prime power, and let $a$ be a non-square in $\FF_q$, i.e., $a\in\FF_q^\times\setminus (\FF_q^\times)^2$. Let $r\geq 2$.
The affine variety $V=V(X_1^2-a X_2^2)\subseteq\AA^r_{\FF_q}$ over $\FF_q$ is irreducible but not absolutely irreducible as $V_{\overline{\FF}_q}$ consists of the two hyperplanes of $\AA^r_{\overline{\FF}_q}$ defined by $X_1+\sqrt{a} X_2$ and $X_1-\sqrt{a} X_2$ respectively.
Let $W=V(X_1,X_2)\subseteq \AA^r_{\FF_q}$, which is an affine subspace of codimension two and is absolutely irreducible.
Then $V(\FF_q)=W(\FF_q)$ and hence $|V(\FF_q)|=q^{r-2}$.
\end{example*}

There is a general algorithm that, given an affine variety $V\subseteq\AA^r_{\FF_q}$ over $\FF_q$ such that none of the irreducible components of $V$ of top dimension (i.e., dimension $\dim V$) is absolutely irreducible, outputs an affine subvariety $W$ of lower dimension such that $W(\FF_q)=V(\FF_q)$ and $W$ has an irreducible component of top dimension that is absolutely irreducible. Namely, for every irreducible component $V_0$ of $V$ of top dimension, replace $V_0$ by the intersection of the irreducible components of $(V_0)_{\overline{\FF}_q}$. This yields an affine variety over $\FF_q$ as it is fixed by the Frobenius map over $\FF_q$. 
Repeating this process, we would eventually obtain an affine subvariety $W\subseteq V$ such that $W(\FF_q)=V(\FF_q)$ and at least one irreducible component of $W$ of top dimension is absolutely irreducible, as required by Definition~\ref{defn:algebraic-source}.
One can then choose the parameters $k$ and $d$ such that $D=f(U_{V(\FF_q)})=f(U_{W(\FF_q)})$ is an $(n,k,d)$ algebraic source over $\FF_q$.

However, the problem is that we do not have a good bound on $\deg W$ or $d$. The general bound we know on $\deg W$ is doubly exponential in $\dim V$. It seems to be an interesting question to understand how large $\deg W$ can be given other parameters such as $\deg V$, $\dim V$, and the dimension $r$ of the ambient space that contains $V$.

\paragraph{Proof of Lemma~\ref{lem:deg-poly-image}.}

We prove Lemma~\ref{lem:deg-poly-image} now.
First, we need the following lemma.

\begin{lemma}\label{lem:reducing-dimension}
Let $f:\AA^n\to\AA^m$ be a polynomial map over an algebraically closed field $\FF$.
Let $V\subseteq \AA^n$ be an irreducible affine variety over $\FF$, and let $W=\overline{f(V)}\subseteq\AA^m$.
Then there exists an affine subspace $H\subseteq\AA^n$ of codimension $\dim V-\dim W$ such that $\dim(H\cap V)=\dim W$ and $\overline{f(H\cap V)}=W$.
\end{lemma}

\begin{proof}
Consider a general point $x\in W$ and let  $t=\dim f|_V^{-1}(x)$. Then $t=\dim V-\dim W$ by the fiber dimension theorem (Theorem~\ref{thm_dimfiber}).
Then there exists an affine subspace $H\subseteq \AA^n$ of codimension $t$ such that $H\cap f|_V^{-1}(x)\neq \emptyset$,  $\dim (H\cap f|_V^{-1}(x))=\dim f|_V^{-1}(x)-t=0$, and $\dim (H\cap V)=\dim V - t=\dim W$.
(This can be shown by, e.g., taking $H$ to be the intersection of $t$ general hyperplanes containing a fixed point of $f|_V^{-1}(x)$ and then applying Lemma~\ref{lem:krull-PIT}.)
Then
\[
\dim (f|_{H\cap V}^{-1}(x))=\dim (H\cap f|_V^{-1}(x))=0=\dim (H\cap V)-\dim W.
\]
Let $W'=\overline{f(H\cap V)}$.
By the fiber dimension theorem (applied to the morphism $f|_{H\cap V}: H\cap V\to W'$), we must have $\dim W'=\dim W$. As $V$ is irreducible, so is $W$. It follows that $W'=W$.
\end{proof}

For convenience, we restate Lemma~\ref{lem:deg-poly-image} below before presenting its proof.

\degpolyimage*

\begin{proof}
By replacing $V$ with $V_{\overline{\FF}}$, we may assume that $\FF$ is algebraically closed.
By considering the irreducible components of $V$ individually, we may assume that $V$ is irreducible. Then $W$ is also irreducible.
Also, noting that performing an invertible linear transformation on $\AA^m_\FF$ (and replacing $f_1,\dots,f_m$ by their linear combinations accordingly) does not change the degrees of the subvarieties of $\AA^m_\FF$. So by Gaussian elimination, we may assume  $f_i\in \mathcal{L}_{h_i,\dots,h_s,\FF}$ for $i\in [m]$, where we let $\mathcal{L}_{h_i,\dots,h_s,\FF}=\FF$ if $i>s$.

Next, we reduce to the case where $\dim V=\dim W$. 
By Lemma~\ref{lem:reducing-dimension}, there exists an affine subspace $H\subseteq\AA^n_\FF$ of codimension $\dim V-\dim W$ such that $\dim(H\cap V)=\dim W$ and $\overline{f(H\cap V)}=W$. Also note that $\deg(H\cap V)\leq \deg V$ by B\'ezout's inequality (Lemma~\ref{lem:bezout}). Thus, by replacing $V$ with $H\cap V$, we may assume $\dim V=\dim W$.

Let $k=\dim W$. By the fiber dimension theorem (Theorem~\ref{thm_dimfiber}), there exists a dense open subset $U$ of $W$ contained in $f(V)$ such that $\dim f|_V^{-1}(b)=0$ for all $b\in U$.
Let $B=W\setminus U$, which is a proper subvariety of $W$. As $W$ is irreducible, we have $\dim B<k$.
A general affine subspace $L\subseteq \AA^m$ of codimension $k$ then satisfies $L\cap B=\emptyset$ and $|L\cap W|=\deg W$.
Fix such $L$.
As $L\cap B=\emptyset$. We have $L\cap W\subseteq U\subseteq f(V)$.
It follows that $L\cap W=f(f|_V^{-1}(L))$.
Moreover, as $\dim f|_V^{-1}(b)=0$ for all $b\in U$ and $L\cap W$ is a finite subset of $U$, the set $f|_V^{-1}(L)=f|_V^{-1}(L\cap W)$ is a finite set.

Suppose $L$ is defined by degree-$1$ polynomials $\ell_1,\dots,\ell_k$. Then $f^{-1}(L)$ is defined by the polynomials $\ell_1(f_1,\dots,f_m),\dots,\ell_k(f_1,\dots,f_m)$.
By Gaussian elimination, we may assume that for each $i\in [k]$, $\ell_i(f_1,\dots,f_m)$ does not involve the polynomials $f_1,\dots,f_{i-1}$, and hence $\ell_i(f_1,\dots,f_m)\in \mathcal{L}_{h_i,\dots,h_s, \FF}$. In particular, we have $\deg \ell_i(f_1,\dots,f_m)\leq \deg h_i$ for $i\in [k]$.

By B\'ezout's inequality, we have $|f|_V^{-1}(L)|=|f^{-1}(L)\cap V|\leq \deg V\cdot \prod_{i=1}^{k} \deg h_i$.
Therefore,
\[
\deg(W)=|L\cap W|=|f(f|_V^{-1}(L))|\leq |f|_V^{-1}(L)|\leq \deg V\cdot \prod_{i=1}^{k} \deg h_i
\]
as desired.
\end{proof}

\paragraph{Proof of Lemma~\ref{lem:laurent-expansion}.}

We now prove Lemma~\ref{lem:laurent-expansion}, which is restated below.

\laurentexpansion*

Let $C_0$ be as in Lemma~\ref{lem:laurent-expansion}.
Regard the affine space $\AA^n$ as an open subset of the projective space $\PP^n$ via $(a_1,\dots,a_n)\mapsto (1,a_1,\dots,a_n)$.
Let $C$ be the \emph{projective closure} of $C_0\subseteq \AA^n$ in  $\PP^n$, which is an irreducible projective curve over $\FF$ whose degree in $\PP^n$ equals $\deg C_0$ (cf. Section~\ref{sec:noether-linear-map}).

For a point $y\in C$, the field $\FF(C)$ of rational functions on $C$ has a subring $\mathcal{O}_{y, C}$, called the \emph{local ring} of $C$ at $y$, which consists of the rational functions that have no pole at $y$ \cite[Section~5]{LL89}. It has a unique maximal ideal $\mathfrak{m}_{y,C}$, which consists of the rational functions on $C$ that vanishes at $y$.

Let $\pi: \widetilde{C}\to C$ be the \emph{normalization} of $C$, which is a finite and surjective morphism from a smooth and irreducible projective curve  $\widetilde{C}$ to $C$ \cite[\S II.5]{Sha13}.

We need the following facts.  
\begin{lemma}[{\cite[Section~5]{LL89}}]\label{lem:field-embedding}
For $x\in \widetilde{C}$ and $y=\pi(x)\in C$, there exists a field embedding $\rho_x: \FF(C)\hookrightarrow \FF((T))$ such that $\rho_x(\mathcal{O}_{y,C})\subseteq \FF[[T]]$ 
and $\rho_x(\mathfrak{m}_{y,C})\subseteq T\cdot \FF[[T]]$.
\end{lemma}

In addition, the following statement holds for the maps $\rho_x$ chosen in \cite[Section~5]{LL89}.

\begin{lemma}[{\cite[Proposition~3]{LL89}}]\label{lem:laurent-lower-bound}
Let $x\in \widetilde{C}$ and $y=\pi(x)\in C$. 
Let $f$ and $g$ be linear homogeneous polynomials nonzero on $C$ such that $g(y)\neq 0$.
Then $f/g$ restricted to $C$ is in $\mathcal{O}_{y,C}$, and $\ord(\rho_x(f/g))\leq \deg C$. 
\end{lemma}

\begin{proof}[Proof of Lemma~\ref{lem:laurent-expansion}]

Let $f\in\FF[X_1,\dots,X_n]$ be a polynomial of degree $d$.
Identify $X_i$ with $Y_i/Y_0$ for $i\in [n]$, where $Y_0,\dots,Y_n$ are the $n$ homogeneous coordinates of $\PP^n$. 
Consider arbitrary $x\in \widetilde{C}$.
We first show that  
\[
\ord(\rho_x(f))\geq -\deg(C_0)\cdot d.
\]
As  $f$ is a polynomial of degree $d$ in $Y_1/Y_0,\dots,Y_n/Y_0$, it suffices to show that  $\ord(\rho_x(Y_i/Y_0))\geq -\deg C_0 $ for $i \in [n]$.
Fix $i\in [n]$. Choose $j \in \{0,1,\dots,n\}$ such that $Y_j$ does not vanish at $y$. Such an index $j$ exists as $Y_0,\dots,Y_n$ do not simultaneously vanish on $\PP^n$. 
By Lemma~\ref{lem:field-embedding} and Lemma~\ref{lem:laurent-lower-bound}, we have
\[
0\leq \ord(\rho_x(Y_i/Y_j)), \ord(\rho_x(Y_0/Y_j))\leq \deg C=\deg C_0.
\]
Then 
\[
\ord(\rho_x(Y_i/Y_0))=\ord(\rho_x(Y_i/Y_j))-\ord(\rho_x(Y_0/Y_j))\geq -\deg C_0,
\]
as desired.

Now assume that $f\in\FF[X_1,\dots,X_n]$ is not constant on $C_0$. It remains to show that there exists $x\in \widetilde{C}$ such that $\ord(\rho_x(f))<0$.
This follows from the following standard argument: Assume to the contrary that $\ord(\rho_x(f))\geq 0$ for all $x\in \widetilde{C}$. Then $f$ is a regular function on $\widetilde{C}$ and hence defines a morphism from $\widetilde{C}$ to $\AA^1$.
Viewing $\AA^1$ as an open subset of $\PP^1$, we get a morphism  $\varphi: \widetilde{C}\to \PP^1$ whose image is contained in $\AA^1$.
On the other hand, it is well-known that the image of a morphism from a projective variety is closed \cite[\S I.5.2, Theorem~2]{Sha13}. And as $\widetilde{C}$ is irreducible, we know $\varphi( \widetilde{C})$ is also irreducible. 
The only irreducible closed subsets of $\PP^1$ are single points and  $\PP^1$ itself.
As $f$ is not constant on $C_0$,  we know  $\varphi( \widetilde{C})$  is not a single point. 
So $\varphi( \widetilde{C})=\PP^1$, which contradicts the fact $\varphi( \widetilde{C})\subseteq\AA^1$.
\end{proof}

\begin{remark*}
The last part of the above proof can be strengthened to show that for a nonzero rational function $g$ on $C$,
\[
\sum_{x\in  \widetilde{C}: \ord(\rho_x(g))\neq 0} \ord(\rho_x(g))=0,
\] 
i.e.,  $g$ has as many zeros as poles on  $\widetilde{C}$, counting  multiplicities.
This implies the existence of $x\in \widetilde{C}$ such that $ \ord(\rho_x(f))<0$ as follows: As $f$ is not constant on $C_0$, there exists $c\in \FF$ such that $g:=f-c$ is not identically zero on $C_0$ but has at least one zero. Then $g$ must also have a pole on $ \widetilde{C}$, i.e., $\ord(\rho_x(g))<0$ for some $x\in \widetilde{C}$.
Then as $f=g+c$, we have $\ord(\rho_x(f))=\min\{\ord(\rho_x(g)), \ord(\rho_x(c))\}=\ord(\rho_x(g))<0$.
\end{remark*}

\section{Explicit Noether Normalization}\label{sec:NNL}

We first prove Theorem~\ref{thm_finite} in the case where $\FF$ is algebraically closed.

\paragraph{Proof of Theorem~\ref{thm_finite} when $\FF$ is algebraically closed.}
Assume that $\FF$ is algebraically closed.
For each $u\in\NN$, identify $\AA^u$ with an open subset of $\PP^u$ via $(x_1,\dots,x_u)\mapsto (1, x_1,\dots,x_u)$.
Define 
\[
\Gamma=\{(x,y)\in V\times \AA^m: y=\varphi(x)\}\subseteq \AA^n\times \AA^m\subseteq \PP^n\times \PP^m.
\]
Let $\widetilde{\Gamma}$ be the (Zariski-)closure of $\Gamma$ in $\PP^n\times \PP^m$. Then $\widetilde{\Gamma}\cap (\AA^n\times \AA^m)=\Gamma$ as $\Gamma$ is closed in $\AA^n\times \AA^m$.

Let $\iota: V\to \Gamma$ be the morphism $x\mapsto (x,\varphi(x))$, which is an isomorphism between the affine varieties $V$ and $\Gamma$ over $\FF$. Let $\pi_1: \widetilde{\Gamma} \to \PP^n$ and $\pi_2: \widetilde{\Gamma} \to \PP^m$ be the projections from $\widetilde{\Gamma}$ to the first factor and the second factor respectively. Then $\pi_2|_{\Gamma}\circ \iota=\varphi|_V$.

\begin{claim}\label{claim_finite}
If $\pi_2^{-1}(\AA^m)\subseteq \Gamma$, then $\varphi|_V$ is a finite morphism.
\end{claim}

To prove Claim~\ref{claim_finite}, we need a result from algebraic geometry. For a closed set $Z$ of $\AA^n\times \AA^m$,
the projection from $Z$ to $\AA^m$ is an example of an \emph{affine morphism} to $\AA^m$, making $Z$ an \emph{affine variety over $\AA^m$}. Similarly, for a closed set $Z$ of $\PP^n\times\AA^m$, the projection from $Z$ to $\AA^m$ is an example of a \emph{projective morphism} to $\AA^m$, making $Z$ a \emph{projective variety over $\AA^m$}. See \cite{Vak17} for the definitions of these objects with various degrees of generality. 
We need the following fact.

\begin{lemma}\label{lem:affine-projective}
A morphism to a variety is affine and projective iff it is finite.
\end{lemma}

See \cite[Corollay~19.1.6]{Vak17}.  The projectivity of $\pi$ can be relaxed to \emph{properness} \cite[Proposition~4.4.2]{EGA3}.
To shed some light on Lemma~\ref{lem:affine-projective}, consider an affine and projective morphism from a variety $Z$ to $\AA^0$ (i.e., a point) over $\FF$, which just means that $Z$ is an affine and projective variety over $\FF$. In this case, the lemma simply states that the coordinate ring $\FF[Z]$ is a finite-dimensional vector space over $\FF$. As $Z$ is affine, the set of regular functions (i.e., functions with no poles) on $Z$ is  $\FF[Z]$. On the other hand, it is well-known that any regular function on a connected projective variety over $\FF$ is constant.
It follows that the dimension of $\FF[Z]$ is indeed finite and equals the number of connected components of $Z$. (In fact, $Z$ is just a finite collection of points.)
Lemma~\ref{lem:affine-projective} states that finiteness holds more generally for any variety as the target of the morphism.

\begin{proof}[Proof of Claim~\ref{claim_finite}]
The map $\pi_2|_\Gamma: \Gamma\to \AA^m$ is an affine morphism.  Suppose $\pi_2^{-1}(\AA^m)\subseteq \Gamma$. Then $\Gamma=\widetilde{\Gamma}\cap (\PP^n \times \AA^m)$. So $\pi_2|_{\Gamma}$ is also a projective morphism. 
It follows from Lemma~\ref{lem:affine-projective} that $\pi_2|_\Gamma$ is finite.
As $\iota$ is an isomorphism between $V$ and $\Gamma$ over $\FF$, the map $\varphi|_V=\pi_2|_{\Gamma}\circ \iota$ is also finite.
\end{proof}

By Claim~\ref{claim_finite}, we may assume that $\pi_2^{-1}(\AA^m)$ contains a point $u\in \widetilde{\Gamma}\setminus \Gamma$.
Let $x=\pi_1(u)\in\PP^n$ and $y=\pi_2(u)\in\AA^m$.
Note that $x\not\in\AA^n$ as otherwise we would have $u=(x,y)\in \widetilde{\Gamma}\cap (\AA^n\times\AA^m)=\Gamma$.
 See the following diagrams for an illustration.
\[
\begin{tikzcd}
 \widetilde{\Gamma} \arrow[r, "\pi_1"] \arrow[d, "\pi_2"] & \PP^n
&  u \arrow[rr, "\pi_1", mapsto] \arrow[d, "\pi_2", mapsto] &[-30pt]  &[-20pt]  x\in \PP^n\setminus\AA^n  \\
 \PP^m & & y &\in\AA^m
\end{tikzcd}
\]

\begin{claim}\label{claim_cond}
There exist $h_1(T), \dots , h_n(T) \in \FF((T))$ satisfying the following conditions:
\begin{enumerate}
\item At least one $h_i(T)$ has a pole, i.e., $h_i(T)\not\in\FF[[T]]$.
\item $P(h_1(T),\dots,h_n(T))=0$ for all $P\in I(V)$. 
\item $f_i(h_1(T),\dots, h_n(T))\in \FF[[T]]$ for $i\in [m]$.
 \end{enumerate}
\end{claim}
\begin{proof} We follow the argument in \cite{KRS96}.
As $\widetilde{\Gamma}$ is the closure of $\Gamma$ in $\PP^n\times\PP^m$, the point $u\in \widetilde{\Gamma}\setminus\Gamma$ cannot possibly be an isolated point of $\widetilde{\Gamma}$ (i.e. a point that is an irreducible component of $\widetilde{\Gamma}$).
So there exists an irreducible curve $C\subseteq \widetilde{\Gamma}$ that passes through $u$ and intersects $\Gamma$. This can be shown by intersecting $\widetilde{\Gamma}$ with general hyperplanes containing $u$ to reduce the dimension, and then picking an irreducible component of the intersection that contains $u$.

Let $\mathcal{O}_{u, C}$ be the local ring of $C$ at $u$, and let $\mathfrak{m}_{u, C}$ be its unique maximal ideal.
Then there exists a field embedding $\rho: \FF(C)\hookrightarrow \FF((T))$ such that $\rho(\mathcal{O}_{u, C})\subseteq \FF[[T]]$ and $\rho(\mathfrak{m}_{u,C})\subseteq T\cdot \FF[[T]]$, as we have seen in the proof of Lemma~\ref{lem:laurent-expansion} in Appendix~\ref{sec:misc}.

Write $x=(x_0,\dots,x_n)$. As $x\in\PP^n\setminus \AA^n$. We have $x_0=0$ and $x_j\neq 0$ for some $j\in [n]$.
Let $X_0,\dots, X_n$ be the homogeneous coordinates of $\PP^n$.
Then $X_j$ does not vanish on $x$, and hence not on $u$ either.
It follows that $X_i/X_j$ restricted to $C$ is in $\mathcal{O}_{u,C}$ for $i=0,\dots, n$.
Let $g_i(T)=\rho(X_i/X_j)\in \FF[[T]]$ for $i=0,\dots, n$.
As $\pi_1(\Gamma)\subseteq \AA^n$, the function $X_0$ does not vanish on any point in $\Gamma$. As $C$ intersects $\Gamma$, we see that $X_0$ does not vanish identically on $C$.
So $X_0/X_j$ restricts to a nonzero rational function on $C$ and hence $g_0(T)=\rho(X_0/X_j)\neq 0$.
For $i\in [n]$, let $h_i(T)=g_j(T)/g_0(T)$. 

Recall that $g_0(T)=\rho(X_0/X_j)$. As $x_0=0$, we know $X_0/X_j$ restricted to $C$ is in $\mathfrak{m}_{u,C}$ and hence $g_0(T)\in T\cdot \FF[[T]]$. And $g_j(T)=\rho(1)=1$ by definition. So $h_j(T)=g_j(T)/g_0(T)$ has a pole, proving the first condition.

Consider arbitrary $P\in I(V)$. 
Let $\widetilde{P}(X_0,\dots,X_n)=X_0^{\deg(P)}P(X_1/X_0,\dots, X_n/X_0)$ be the homogenization of $P$. 
We know $P$ vanishes identically on $V$ and hence also on $\Gamma$.
So $\widetilde{P}$ vanishes identically on its closure $\widetilde{\Gamma}$, which contains the curve $C$.
So $\widetilde{P}(X_0/X_j,\dots,X_n/X_j)$ restricted to $C$ is zero.
It follows that
$
\widetilde{P}(g_0(T),\dots,g_n(T))=\rho(\widetilde{P}(X_0/X_j,\dots,X_n/X_j))=0
$.
Therefore,
\[
P(h_1(T),\dots,h_n(T))=P(g_1(T)/g_0(T),\dots, g_n(T)/g_0(T))=g_0(T)^{-\deg(P)}\widetilde{P}(g_0(T),\dots,g_n(T))=0,
\]
proving the second condition.

Let $Y_1,\dots,Y_m$ be the coordinates of $\AA^m$, which (restricted to $C$) are in $\mathcal{O}_{u, C}$ as $u\in \PP^n\times \AA^m$.
Consider arbitrary $i\in [m]$. Let $F_i(X_1,\dots,X_n,Y_i)=f_i(X_1,\dots,X_n)-Y_i$. Let $\widetilde{F}_i(X_0,\dots,X_n,Y_i)=X_0^{\deg(f_i)}(f_i(X_1/X_0,\dots, X_n/X_0)-Y_i)$ be the homogenization of  $F_i$ with respect to  $X_1,\dots, X_n$.  By   definition, $F_i$ vanishes identically on $\Gamma\subseteq\AA^n\times\AA^m$ and hence $\widetilde{F}_i$ vanishes identically on $\widetilde{\Gamma}\cap (\PP^n\times \AA^m)$. As $C\cap (\PP^n\times \AA^m)$ is a subset of $\widetilde{\Gamma}\cap (\PP^n\times \AA^m)$ and is dense in $C$, we see that $\widetilde{F}(X_0/X_j,\dots,X_n/X_j, Y_i)$ restricted to $C$ is zero.
Then  $\widetilde{F}_i(g_0(T),\dots,g_n(T),\rho(Y_i))=\rho(\widetilde{F}(X_0/X_j,\dots,X_n/X_j, Y_i))=0$. Therefore,
\begin{align*}
f_i(h_1(T),\dots,h_n(T))-\rho(Y_i)&=F_i(h_1(T),\dots,h_n(T),\rho(Y_i))\\
&=F_i(g_1(T)/g_0(T),\dots, g_n(T)/g_0(T),\rho(Y_i))\\
&=g_0(T)^{-\deg(f_i)}\widetilde{F}_i(g_0(T),\dots,g_n(T),\rho(Y_i))\\
&=0.
\end{align*}
So $f_i(h_1(T),\dots,h_n(T))=\rho(Y_i)\in \rho(\mathcal{O}_{u,C})\subseteq\FF[[T]]$, proving the third condition.
\end{proof}

However, Lemma~\ref{lem:three-conditions} states that  Claim~\ref{claim_cond} cannot be true. So we obtain a contradiction, implying that the assumption that $\pi_2^{-1}(\AA^m)$ contains a point $u\in \widetilde{\Gamma}\setminus \Gamma$ is false. This concludes the proof of Theorem~\ref{thm_finite} when $\FF$ is algebraically closed.

\paragraph{Proof of Theorem~\ref{thm_finite} for arbitrary $\FF$.}
Theorem~\ref{thm_finite} for an aritrary field $\FF$ follows from the special case where $\FF$ is algebraically closed and the fact that finiteness descends under a ``faithfully flat base change." We explain this now. 

In the following, all rings and algebras are commutative with unity.

Let $R$ be a ring. For $R$-modules $M$ and $M'$, their \emph{tensor product} $M\otimes_R M'$ over $R$ is defined to be the $R$-module generated by the set of elements $\{a\otimes b: a\in M, b\in M'\}$ subject to the $R$-bilinear relations $a\otimes b+a'\otimes b=(a+a')\otimes b$, $a\otimes b+a\otimes b'=a\otimes(b+b')$, and $c(a\otimes b)=(ca)\otimes b=a\otimes (cb)$ for $a,a'\in M$, $b,b'\in M'$, and $c\in R$.  For an $R$-algebra $S$ and an $R$-module $M$, the tensor product $M\otimes_R S$ is an $S$-module.

An $R$-module is \emph{free} if it has a generating set that is linearly independent over $R$. We need the following fact about nonzero free modules (which holds more generally for \emph{faithfully flat modules} \cite[Definition~25.5.1]{Vak17}).

\begin{lemma}\label{lem:surjectivity}
Let $F$ be a nonzero free $R$-module. Let $M\to M'$ be an $R$-module homomorphism. Then $M\to M'$ is surjective iff the induced map $M\otimes_R F\to M'\otimes_R F$ is surjective.
\end{lemma}

Lemma~\ref{lem:surjectivity} implies the following fact about descent of finiteness.
 
\begin{lemma}\label{lemma:finiteness-descent}
Let $S$ be an $R$-algebra that is also a nonzero free $R$-module.  Let $M$ be an $R$-module. Suppose $M\otimes_R S$ is a finitely generated $S$-module. Then $M$ is a finitely generated $R$-module.
\end{lemma}

\begin{proof}
Let $\{r_1,\dots,r_n\}$ be a finite set of generators of $M\otimes_R S$. By definition, we may write each $r_i$ as a finite sum $r_i=\sum_{j\in I_i} m_{ij}\otimes s_{ij}$ over an index set $I_i$ with $m_{ij}\in M$ and $s_{ij}\in S$.
Form the free $R$-module $F$ with the basis $\{e_{ij}: i\in [n], j\in I_i\}$. Consider the $R$-module homomorphism $F\to M$ sending $e_{ij}$ to $m_{ij}$. The induced map $F\otimes_R S\to M\otimes_R S$ is surjective as $\sum_{j\in I_i} e_{ij}\otimes s_{ij}$ is sent to $r_i$ for $i\in [n]$.
By Lemma~\ref{lem:surjectivity}, the map $F\to M$ is also surjective.
So $M$ is generated by the finite set $\{m_{ij}: i\in [n], j\in I_i\}$ as an $R$-module.
\end{proof}

Now we are ready to prove Theorem~\ref{thm_finite} in full generality.

\begin{proof}[Proof of Theorem~\ref{thm_finite}]
Let $\widetilde{\varphi}=\varphi|_V: V\to \AA^m_\FF$, which is associated with the $\FF$-algebra homomorphism $\widetilde{\varphi}^\sharp:\FF[Y_1,\dots,Y_m]\to\FF[V]$.
Let $R=\widetilde{\varphi}^\sharp(\FF[Y_1,\dots,Y_m])\subseteq \FF[V]$. By definition, we want to show that $\FF[V]$ is a finitely generated module over $R$.

For an affine variety $W$ over $\FF$, the coordinate ring  of $W_{\overline{\FF}}$ may be identified with $\FF[W]\otimes_\FF \overline{\FF}$. As we already know that Theorem~\ref{thm_finite} holds over algebraically closed fields by the discussion above, applying Theorem~\ref{thm_finite} to the morphism $V_{\overline{\FF}}\to \AA^m_{\overline{\FF}}$ defined by $f_1,\dots,f_m$ shows that $\FF[V]\otimes_{\FF} \overline{\FF}$ is a finitely generated module over $R\otimes_\FF \overline{\FF}$.

The ring $R\otimes_\FF \overline{\FF}$ is a free module over $R$ since $\overline{\FF}$ is free over $\FF$ and tensor products commute with direct sums.
By Lemma~\ref{lemma:finiteness-descent}, to show that $\FF[V]$ is a  finitely generated module over $R$, it suffices to show that
$\FF[V]\otimes_R (R\otimes_\FF \overline{\FF})$ is a finitely generated module over $R\otimes_\FF \overline{\FF}$.

Finally, we have the canonical isomorphisms 
\[
\FF[V]\otimes_R (R\otimes_\FF\overline{\FF})\cong (\FF[V]\otimes_R R)\otimes_\FF\overline{\FF}\cong \FF[V]\otimes_\FF \overline{\FF}.
\]
See \cite[Proposition~2.14]{AM69}. 
So $\FF[V]\otimes_R (R\otimes_\FF \overline{\FF})$ is a finitely generated module over $R\otimes_\FF \overline{\FF}$, as desired.
\end{proof}

\paragraph{Proof of Theorem~\ref{thm_algNNL}.}
Theorem~\ref{thm_algNNL} is almost equivalent to Theorem~\ref{thm_finite} except that the coordinate ring of an affine variety is always reduced (i.e., it has no nonzero nilpotent elements) while the algebra $A$ in Theorem~\ref{thm_algNNL} may be non-reduced. Nevertheless, we can easily derive Theorem~\ref{thm_algNNL} from Theorem~\ref{thm_finite} as follows.

Let $R'$ be an algebra over a ring $R$. We say $R'$ is \emph{integral} over $R$ if every $b\in R'$ is a root of a \emph{monic} polynomial $X^t+a_{t-1} X^{t-1}+\dots+a_0$ with the coefficients $a_i\in R$. We need the following fact, which states that integrality is equivalent to finiteness for finitely generated algebras.

\begin{lemma}[{\cite[Proposition~5.1 and Corollary~5.2]{AM69}}]\label{lem:finite-iff-integral}
A finitely generated algebra over a ring $R$ is a finitely generated module over $R$ iff it is integral over $R$.
\end{lemma}

Now we give the proof of Theorem~\ref{thm_algNNL}.
 
\begin{proof}[Proof of Theorem~\ref{thm_algNNL}] 
If $I$ is radical, then $I=I(V(I))$ and hence $A\cong\FF[V(I)]$. In this case, the theorem follows from  Theorem~\ref{thm_finite} and the definition of finite morphisms. 

Now consider general $I$. Identify $A$ with $\FF[X_1,\dots,X_n]/I$ and let $\bar{A}=\FF[X_1,\dots,X_n]/\mathrm{rad}(I)=A/\mathrm{nil}(A)$, where $\mathrm{rad}(I)$ denotes the radical of $I$ and $\mathrm{nil}(A)$ denotes the nilradical of $A$, i.e., the radical of the zero ideal. As $\mathrm{rad}(I)$ is radical and $V(\mathrm{rad}(I))=V(I)$, we see that $\bar{A}$ is a finitely generated module over $S$ by Theorem~\ref{thm_finite}. So $\bar{A}$ is integral over $S$ by Lemma~\ref{lem:finite-iff-integral}.
Now consider arbitrary $b\in A$. As $\bar{A}$ is integral over $S$, there exists a monic polynomial $F\in S[X]$ such that the image of $F(b)$ in $\bar{A}$ is zero, or equivalently, $F(b)\in \mathrm{nil}(A)$. So $F(b)^N=0$ for some $N\in\NN^+$. Then $b$ is a root of the monic polynomial $F^N$. Therefore, $A$ is integral over $S$ and hence is a finitely generated module over $S$ by Lemma~\ref{lem:finite-iff-integral}. 
\end{proof}

\end{appendices}

\bibliographystyle{alpha}
\bibliography{algebraic-extractors}

\end{document}